\documentclass{jocg}

\usepackage{amsmath,amsthm,amssymb}
\usepackage{verbatim}
\usepackage{color}
\usepackage{url}
\usepackage{enumitem}
\usepackage{subcaption}
\usepackage[pdftex]{graphicx}
\usepackage{verbatim}
\usepackage{ifthen}
\usepackage{comment}

\theoremstyle{plain}
\newtheorem{theorem}{Theorem}[section]
\newtheorem{lemma}[theorem]{Lemma}

\newtheorem{construction}[theorem]{Construction}

\newtheorem*{question}{Open Question}
\newtheorem*{theorem*}{Theorem}

\theoremstyle{definition}
\newtheorem{definition}[theorem]{Definition}
\newtheorem{convention}[theorem]{Convention}
\newtheorem*{note}{Note}

\DeclareTextFontCommand{\term}{\color{Plum}\em}

\usepackage{ifthen}

\newcommand\headingboldmath[2]{#1[{#2}]{\boldmath #2}}

\renewcommand\epsilon{\varepsilon}
\newcommand\eps{\epsilon}

\renewcommand\phi{\varphi}

\usepackage[e]{esvect}
\newcommand\widevec[1]{\vv{#1}}

\newcommand\inv{^{-1}}

\newcommand\bR{\mathbb{R}}

\newcommand\bQ{\mathbb{Q}}
\newcommand\bZ{\mathbb{Z}}
\newcommand\bN{\mathbb{N}}

\newcommand\cL{\mathcal{L}}
\newcommand\cT{\mathcal{T}}
\newcommand\cG{\mathcal{G}}
\newcommand\cM{\mathcal{M}}
\newcommand\cA{\mathcal{A}}

\newcommand\cP{\mathcal{P}}
\newcommand\cE{\mathcal{E}}

\newenvironment{descriptionflush}{\begin{description}[leftmargin=0cm,listparindent=\parindent,style=unboxed,itemsep=6pt]}{\end{description}} %

\DeclareMathOperator\Conf{Conf}
\DeclareMathOperator\NXConf{NXConf}
\DeclareMathOperator\Offset{Offset}

\DeclareMathOperator\Con{Con}
\DeclareMathOperator\RigidCon{RigidCon}

\DeclareMathOperator\SliceCon{SliceCon}
\DeclareMathOperator\AngleCon{AngleCon}
\DeclareMathOperator\NXCon{NX}
\DeclareMathOperator\Rot{Rot}
\DeclareMathOperator\poly{poly}
\DeclareMathOperator\Rect{Rect}
\DeclareMathOperator\Coeffs{Coeffs}
\DeclareMathOperator\area{area}

\newcommand\probname[1]{\ensuremath{\mathsf{#1}}}
\newcommand\HtwoN{\probname{H_2N}}
\newcommand\CommonZero{\probname{CommonZero}}
\newcommand\Stretchability{\probname{Stretchability}}
\newcommand\PointConfiguration{\probname{PointConfiguration}}
\newcommand\ETR{\probname{ETR}}
\newcommand\SAT{\probname{SAT}}

\newcommand\subscriptfont[1]{\text{\rm #1}}
\newcommand\Ucopy{U_{\subscriptfont{copy}}}
\newcommand\Ucross{U_{\subscriptfont{cross}}}
\newcommand\Uangular{U_{\subscriptfont{angular}}}
\newcommand\Uvecrot{U_{\subscriptfont{rot}}}
\newcommand\Uangleavg{U_{\angle\subscriptfont{avg}}}
\newcommand\Uvecavg{U_{\widevec{\subscriptfont{avg}}}}
\newcommand\Uanglesum{U_{\angle\subscriptfont{sum}}}
\newcommand\Uvecsum{U_{\widevec{\subscriptfont{sum}}}}

\newcommand\Ustart{U_{\subscriptfont{start}}}
\newcommand\Uveccreate{U_{\subscriptfont{create}}}

\newcommand\cLhook{\cL_{\subscriptfont{hook}}}

\newcommand\cLcopy{\cL_{\subscriptfont{copy}}}
\newcommand\cLcross{\cL_{\subscriptfont{cross}}}
\newcommand\cLangleavg{\cL_{\angle\subscriptfont{avg}}}
\newcommand\cLanglesum{\cL_{\angle\subscriptfont{sum}}}
\newcommand\cLangular{\cL_{\subscriptfont{angular}}}

\newcommand\cLstart{\cL_{\subscriptfont{start}}}
\newcommand\cLvecavg{\cL_{\widevec{\subscriptfont{avg}}}}
\newcommand\cLvecsum{\cL_{\widevec{\subscriptfont{sum}}}}
\newcommand\cLvecrot{\cL_{\subscriptfont{rot}}}
\newcommand\cLveccreate{\cL_{\subscriptfont{create}}}

\newcommand\cLend{\cL_{\subscriptfont{end}}}
\newcommand\cLendcrossing{\cL_{\subscriptfont{X-end}}}

\newcommand\cLslice{\cL_{\angle\subscriptfont{slice}}}
\newcommand\cLanglerestrictor{\cL_{\angle\subscriptfont{restrict}}}

\newcommand\cLparallel{\cL_{\subscriptfont{parallel}}}

\newcommand\Gcell{G_{\subscriptfont{cell}}}

\newcommand\ifnotempty[2]{\ifthenelse{\equal{#1}{}}{}{#2}}

\newcommand\mysubfigure[5][]{%
  \begin{subfigure}[{#1}]{#2}%
    \centering%
    #3%
    \ifnotempty{#4}{\caption{#4}}%
    \ifnotempty{#5}{\label{#5}}%
  \end{subfigure}%
}

\title{%
  \MakeUppercase{Who Needs Crossings?: Noncrossing Linkages are Universal, and Deciding (Global) Rigidity is Hard}%
    \thanks{A preliminary version of this paper appeared at the 32nd International Symposium on Computational Geometry, June 2016.
    Originally invited to the corresponding special issue,
    but the preparation of the final paper was delayed.}
}

\author{%
  Zachary~Abel,%
  \thanks{%
    \affil{MIT Department of Electrical Engineering and Computer Science, 50 Vassar St., Cambridge, MA 02139, USA, \email{zabel@mit.edu}. Partially supported by an NSF Graduate Research Fellowship.}
  }\,
  Erik~D.~Demaine,%
  \thanks{
    \affil{MIT Computer Science and Artificial Intelligence Laboratory,
      32 Vassar St., Cambridge, MA 02139, USA,}
    \email{\{edemaine,mdemaine,seisenst,jaysonl,neboat\}@mit.edu}%
  }\,
  Martin~L.~Demaine,\footnotemark[3]\,
  Sarah~Eisenstat,\footnotemark[3]\,
  Jayson~Lynch,\footnotemark[3]\,
  and Tao~B.~Schardl.\footnotemark[3]
}

\begin{document}
\maketitle

\begin{abstract}
  We exactly settle the complexity of graph realization, graph rigidity, and graph global rigidity as applied to three types of graphs: ``globally noncrossing'' graphs, which avoid crossings in all of their configurations; matchstick graphs, with unit-length edges and where only noncrossing configurations are considered; and unrestricted graphs (crossings allowed) with unit edge lengths (or in the global rigidity case, edge lengths in $\{1,2\}$). We show that all nine of these questions are complete for the class $\exists\bR$, defined by the Existential Theory of the Reals, or its complement $\forall\bR$; in particular, each problem is (co)NP-hard.

  One of these nine results---that realization of unit-distance graphs is $\exists\bR$-complete---was shown previously by Schaefer (2013), but the other eight are new. We strengthen several prior results. Matchstick graph realization was known to be NP-hard (Eades \& Wormald 1990, or Cabello et al.\ 2007), but its membership in NP remained open; we show it is complete for the (possibly) larger class $\exists\bR$. Global rigidity of graphs with edge lengths in $\{1,2\}$ was known to be coNP-hard (Saxe 1979); we show it is $\forall\bR$-complete.

  The majority of the paper is devoted to proving an analog of Kempe's Universality Theorem---informally, ``there is a linkage to sign your name''---for globally noncrossing linkages. In particular, we show that any polynomial curve $\phi(x,y)=0$ can be traced by a noncrossing linkage, settling an open problem from 2004. More generally, we show that the regions in the plane that may be traced by a noncrossing linkage are precisely the compact semialgebraic regions (plus the trivial case of the entire plane). Thus, no drawing power is lost by restricting to noncrossing linkages. We prove analogous results for matchstick linkages and unit-distance linkages as well.

\end{abstract}

\section{Introduction}
\label{sec:intro}

\begin{table}
  \centering
  \tabcolsep=0.73\tabcolsep
  \def\rowantiskip{\vspace{-0.5ex}}
  \let\oldexists=\exists
  \let\oldforall=\forall
  \def\pmbexists{\pmb{\oldexists}}
  \def\pmbforall{\pmb{\oldforall}}
  \def\pmbbR{\pmb{\mathbb{R}}}
  \def\NEW#1{\textbf{\let\exists=\pmbexists\let\forall=\pmbforall\let\bR=\pmbbR #1}}
  \begin{tabular}{l|c|c|c|c}
    \bf Graph type & \bf Realization & \bf Rigidity & \bf Global rigidity & \bf Universality
    \\ \hline
    General              & $\exists\bR$-complete
                         & $\forall\bR$-complete
                         & \NEW{$\forall\bR$-complete}
                         & Compact
    \rowantiskip\\
                         & \cite{Schaefer-2013}
                         & \cite{Schaefer-2013}
                         & (CoNP-hard \cite{Saxe-1979})
                         & semialg. \cite{King-1999}
    \\
    \hline
    Globally noncrossing & \NEW{$\exists\bR$-complete}
                         & \NEW{$\forall\bR$-complete}
                         & \NEW{$\forall\bR$-complete}
                         & \NEW{Compact}
    \rowantiskip\\
    \em (no configs.\ cross) & 
                         &
                         &
                         & \NEW{semialg.}
    \\
    \hline
    Matchstick graph     & \NEW{$\exists\bR$-complete}
                         & \NEW{$\forall\bR$-complete}
                         & \NEW{$\forall\bR$-complete}
                         & \NEW{Bounded}
    \rowantiskip\\
    \em (unit + noncrossing) & (NP-hard \cite{Eades-Wormald-1990})
                         &
                         &
                         & \NEW{semialg.}
    \\
    \hline
    Unit edge lengths    & $\exists\bR$-complete
                         & \NEW{$\forall\bR$-complete}
                         & Open (do they
                         & \NEW{Compact}
    \rowantiskip\\
    \em (allowing crossings) & \cite{Schaefer-2013}
                         &
                         & even exist?)
                         & \NEW{semialg.}
    \\
    \hline
    Edge lengths in $\{1,2\}$ & $\exists\bR$-complete
                         & \NEW{$\forall\bR$-complete}
                         & \NEW{$\forall\bR$-complete}
                         & \NEW{Compact}
    \rowantiskip\\
    \em (allowing crossings) & \cite{Schaefer-2013}
                         & 
                         & (CoNP-hard \cite{Saxe-1979})
                         & \NEW{semialg.}
    \\
    \hline
  \end{tabular}
  \caption[Summary of our results (bold) compared with earlier results (cited).]{Summary of our results (bold) compared with earlier results (cited).
    The rows give the special types of graphs considered.
    The middle three columns give complexity results for the three natural
    decision problems about graph embedding; all completeness results are
    strong.
    The rightmost column gives the exact characterization of drawable sets.}
  \label{tab:results}
\end{table}

The rise of the steam engine in the mid-1700s led to an active study of
\emph{mechanical linkages}, typically made from rigid bars connected together
at hinges.  For example, steam engines need to convert the linear motion of
a piston into the circular motion of a wheel, a problem solved approximately
by Watt's parallel motion linkage (1784) and exactly by Peaucellier's inversor (1864)
\cite[Section~3.1]{Demaine-O'Rourke-2007}.
These and other linkages are featured in an 1877 book called
\emph{How to Draw a Straight Line} \cite{Kempe-1877} by Alfred Bray Kempe---a
barrister and amateur mathematician in London, perhaps most famous for his
false ``proof'' of the Four-Color Theorem \cite{Kempe-1879} that nonetheless
introduced key ideas used in the correct proofs of today
\cite{Appel-Haken-1977,Robertson-Sanders-Seymour-1997}.

\paragraph{Kempe's Universality Theorem.}
Kempe investigated far beyond drawing a straight line by turning a circular crank.
In 1876, he claimed a universality result, now known as Kempe's Universality
Theorem: every polynomial curve $\phi(x,y) = 0$ can be traced by a vertex of
a 2D linkage \cite{Kempe-1876}.  Unfortunately, his ``proof'' was again flawed:
the linkage he constructs indeed traces the intended curve, but also traces
finitely many unintended additional curves.  Fortunately, his idea was spot on.

Many researchers have since solidified and/or strengthened Kempe's Universality
Theorem \cite{Kapovich-Millson-2002,Jordan-Steiner-1999,King-1999,Abbott-2008,Schaefer-2013}.
In particular, small modifications to Kempe's gadgets lead
to a working proof \cite[Section~3.2]{Abbott-2008,Demaine-O'Rourke-2007}.
Furthermore, the regions of the plane drawable by a 2D linkage
(other than the entire plane $\bR^2$) are exactly
compact semialgebraic regions%
\footnote{A compact planar region is \term{semialgebraic} if it can be
  obtained as the intersection and/or union of finitely many basic sets defined by
  polynomial inequalities $p(x, y) \geq 0$.}
\cite{King-1999,Abbott-2008}.
By carefully constructing these linkages to have rational coordinates,
Abbott, Barton, and Demaine~\cite{Abbott-2008} showed how to reduce the problem of testing isolatedness of a point in an algebraic set%
\footnote{A set $S \subseteq \mathbb R^d$ is \term{algebraic} if it can be
  written as the set of solutions to a polynomial equation
  $p(x_1, \dots, x_d) = 0$.  (Algebraic sets are closed under intersection
  and finite union.)}
to testing
rigidity of a linkage.  Isolatedness was proved coNP-hard \cite{Koiran-2000}
and then $\forall \bR$-complete%
\footnote{The class $\forall \bR = \text{co-}\exists \bR$ consists
  of decision problems whose complement (inverting yes/no instances) belong to
  $\exists \bR$.
  The class $\exists \bR$ refers to the problems (Karp) reducible to
  the \term{existential theory of the reals}
  ($\exists x_1 : \cdots \exists x_n : \pi(x_1,\dots,x_n)$ for an arithmetic predicate $\pi : \bR^n \to \{\text{true},\text{false}\}$), which is somewhere between NP
  and PSPACE (by \cite{Canny-1988-pspace}).  The classic example of an
  $\exists \bR$-complete problem is pseudoline stretchability
  \cite{Mnev-1988}. The classes $\exists\bR$ and $\forall\bR$ are discussed more thoroughly in Section~\ref{sec:exists-r}.}
\cite{Schaefer-2013}; thus linkage rigidity is $\forall \bR$-complete.

\paragraph{Our results: no crossings.}
See Table~\ref{tab:results} for a summary of our results in comparison to
past results.  Notably,
all known linkage constructions for Kempe's Universality Theorem (and its
various strengthenings) critically need to allow the bars to cross each other.
In practice, certain crossings can be made physically possible,
by placing bars in multiple parallel planes and constructing vertices as
vertical pins.  Without extreme care, however, bars can still be blocked by
other pins, and it seems difficult to guarantee crossing avoidance
for complex linkages.
Beyond these practical issues, it is natural to wonder whether allowing
bars to cross is necessary to achieve linkage universality.
Don Shimamoto first posed this problem in April 2004, and it was highlighted as
a key open problem in the first chapter of \emph{Geometric Folding Algorithms}
\cite{Demaine-O'Rourke-2007}.

We solve this open problem by strengthening most of the results mentioned above
to work for \term{globally noncrossing graphs}, that is, graphs plus
edge-length constraints that alone force all configurations to be
(strictly) noncrossing.%
\footnote{Thus, the noncrossing constraint can be thought of as being
  ``required'' or not of a configuration; in either case, the configurations
  (even those reachable by discontinuous motions) will be noncrossing.}
In particular, we prove the following universality and complexity results.
\begin{enumerate}[listparindent=\parindent]
\item The planar regions drawable by globally noncrossing linkages
      are exactly ($\bR^2$ and) the compact semialgebraic regions
      (Theorem~\ref{thm:noncrossing-universality}),
      settling Shimamoto's 2004 open problem.
\item Testing whether a globally noncrossing graph with constant-sized integer edge lengths has any valid configurations
      is $\exists \bR$-complete (Theorem~\ref{thm:noncrossing-realizability}).
\item Testing rigidity is strongly\footnote{A problem is said to be \emph{strongly} hard if it remains hard when all integers in the input are specified in unary rather than binary. In other words, all integers are polynomially bounded in the size of the input rather than exponentially bounded. All of our hardness results show strong hardness when applicable.} $\forall \bR$-complete even for
      globally noncrossing graphs with constant-sized integer edge lengths that are drawn with integer vertex coordinates (Theorem~\ref{thm:noncrossing-rigidity}).
\item Testing global rigidity (uniqueness of a given embedding) is
      strongly $\forall \bR$-complete even for globally noncrossing graphs with constant-sized integer edge lengths that are
      drawn with integer vertex coordinates
      (Theorem~\ref{thm:noncrossing-global-rigidity}).
\end{enumerate}
Our techniques are quite general and give us results for two other restricted
forms of graphs as well.  First, \term{matchstick graphs} are graphs with
\emph{unit} edge-length constraints, and where only (strictly) noncrossing
configurations are considered valid.  We prove in Section~\ref{sec:matchstick} the following universality and complexity results:
\begin{enumerate}[resume,listparindent=\parindent]
\item The planar regions drawable by matchstick graphs
      are exactly ($\bR^2$ and) the bounded semialgebraic regions (Theorem~\ref{thm:matchstick-universality-full}).
  Notably, unlike all other models considered, matchstick graphs enable
  the representation of open boundaries in addition to closed (compact)
  boundaries.
\item Deciding whether an abstract graph can be draw as a matchstick graph is $\exists\bR$-complete (Theorem~\ref{thm:matchstick-realizability}).
  This result strengthens a 25-year-old NP-hardness result
  \cite{Eades-Wormald-1990,Cabello-Demaine-Rote-2007},
  and settles an open question of \cite{Schaefer-2013}.
\item Testing rigidity or global rigidity of a matchstick graph is
  strongly $\forall\bR$-complete (Theorems~\ref{thm:matchstick-rigidity} and~\ref{thm:matchstick-global-rigidity}).
\end{enumerate}
Second, we consider restrictions on edge lengths to be either all equal (unit)
or all in $\{1,2\}$, but at the price of allowing crossing configurations.
Recognizing unit-distance graphs is already known to be $\exists\bR$-complete
\cite{Schaefer-2013}.  We prove in Section~\ref{sec:unit} the
following additional universality and complexity results:
\begin{enumerate}[resume,listparindent=\parindent]
\item The planar regions drawable by unit-edge-length linkages
      are exactly the compact semialgebraic regions (and $\bR^2$) (Theorem~\ref{thm:unit-universality}),
      proving a conjecture of Schaefer \cite{Schaefer-2013}.
\item Testing rigidity of unit-edge-length graphs is strongly
      $\forall\bR$-complete (Theorem~\ref{thm:unit-rigidity}),
      proving a conjecture of Schaefer \cite{Schaefer-2013}.
\item Testing global rigidity of graphs with edge lengths in $\{1,2\}$
      is strongly $\forall\bR$-complete (Theorem~\ref{thm:unit-global-rigidity}).
  This result strengthens a 35-year-old strong-coNP-hardness result
  for the same scenario~\cite{Saxe-1979}.
  While it would be nice to
  strengthen this result to
  unit edge lengths, we have been unable to find even a single globally
  rigid equilateral linkage larger than a triangle.
\end{enumerate}

We introduce several techniques to make noncrossing linkages manageable
in this setting.  In Section~\ref{sec:defining-extended-linkages} we define
\emph{extended linkages} to allow additional joint types, in particular,
requiring angles between pairs of bars to stay within specified intervals.
Section~\ref{sec:detailed-overview} then shows how to draw a polynomial curve and
obtain Kempe's Universality Theorem with these powerful linkages while
avoiding crossings, by following the spirit of Kempe's original construction
but with specially designed modular gadgets to guarantee no crossings
between (or within) the gadgets. We simulate extended linkage with linkages that have chosen subgraphs marked as rigid. In turn, in  Sections~\ref{sec:noncrossing}--\ref{sec:matchstick}, we simulate these ``partially rigidified'' linkages with the three desired linkage types: globally noncrossing, unit-distance or $\{1,2\}$-distance, and matchstick.

\section{Description of the Main Construction}
\label{sec:main-theorem}

The heart of this paper is a single, somewhat intricate linkage construction. In this section, we describe and discuss the properties of this construction in detail, after building up the necessary terminology.

\subsection{Definitions: Linkages and Graphs}
\label{sec:linkage-definitions}

Unless otherwise specified, all graphs $G = (V(G), E(G), \ell)$ in this text are connected, edge-weighted with positive edge lengths $\ell(e)>0$, and contain no self-loops. We first recall and establish notation for our primary objects of study, linkages.

\begin{definition}[Linkages]
  An \term{abstract linkage}, or simply a \term{linkage}, is a triple $\cL = (G,W,P)$ consisting of a weighted graph $G = (V(G),E(G),\ell)$ together with a choice of \term{pin locations} $P(w)\in\bR^2$ for vertices $w$ in a chosen subset $W\subset V(G)$ of \term{pinned vertices}.
\end{definition}

\begin{definition}[Linkage Configurations]
  A \term{configuration} of a linkage $\cL = (G,W,P)$ is an assignment of vertex locations $C:V(G)\to\bR^2$ respecting the edge-length and pin assignments: $|C(u)-C(v)| = \ell(u v)$ for each edge $u v$, and $C(w) = P(w)$ for each pinned vertex $w\in W$.
  Two configurations are \term{congruent} if they differ only by a Euclidean transformation, i.e., a translation, rotation, and possibly reflection.

  The \term{configuration space} $\Conf(\cL)\subseteq (\bR^2)^{|V(G)|}$ is the set of all configurations; it is a closed, algebraic subset of $(\bR^2)^{|V(G)|}$. (Be warned, however, that some modified types of linkages to follow, notably \emph{NX-constrained} linkages or \emph{extended linkages}, will have only \emph{semi}algebraic configuration spaces.)

  An abstract linkage is called \term{configurable} or \term{realizable} if its configuration space is nonempty, and a \term{configured linkage} is a linkage together with a chosen \term{initial configuration} $C_0\in\Conf(\cL)$. Configured linkage $(\cL,C_0)$ is \term{rigid} if there is no nontrivial continuous deformation of $\cL$, i.e., $C_0$ is isolated in $\Conf(\cL)$. Similarly, $(\cL,C_0)$ is \term{globally rigid} if $C_0$ is the \emph{only} configuration in $\Conf(\cL)$.
\end{definition}

\begin{convention}
  As a convenient abuse of notation, we often write $v$ instead of $C(v)$ when configuration $C$ is understood from context.
\end{convention}

We consider \term{abstract graphs} and \term{configured graphs} as abstract or configured linkages without pins ($P = \emptyset$), with one key difference: rigidity and global rigidity for graphs are more liberally defined to allow Euclidean motions.

\begin{definition}[Graph Rigidity and Global Rigidity]
  A configured graph $(G,C_0)$ is \term{rigid} if the only continuous deformations of $(G,C_0)$ are rigid Euclidean motions, i.e., if there is a neighborhood of $C_0$ in $\Conf(G)$ consisting only of configurations congruent to $C_0$. Likewise, $(G,C_0)$ is \term{globally rigid} if all configurations are congruent to $C_0$.
\end{definition}

As defined above, configurations may have coincident vertices,
vertices in the middle of edges, and/or properly crossing edges.
The following notion forbids these undesirable features:

\begin{definition}[Noncrossing Configurations] \label{def:noncrossing}
  A configuration $C$ of a linkage $\cL$ is \term{noncrossing} if distinct edges intersect only at common endpoints: for any pair of incident edges $u v \ne u v' \in E(G)$ sharing a vertex $u$, the segments $C(u)C(v)$ and $C(u)C(v')$ intersect only at $C(u)$ in $\bR^2$, and for any pair $u v, u' v'$ of disjoint edges in $G$, segments $C(u)C(v)$ and $C(u')C(v')$ are disjoint in $\bR^2$.

  Let $\NXConf(\cL)\subseteq\Conf(\cL)$ denote the subset of noncrossing configurations. Then we say $\cL$ is \term{globally noncrossing} if all of its configurations are noncrossing, i.e., $\NXConf(\cL) = \Conf(\cL)$. If $C$ is a noncrossing configuration, the \term{minimum feature size} of $C$ is the shortest distance from a segment $C(u)C(v)$ to a point $C(w)$ for $u v \in E(G)$ and $w\in V(G)\setminus\{u,v\}$. If linkage $\cL$ is globally noncrossing, its \term{global minimum feature size} is defined as the infimum of the minimum feature size of its configurations. If $\Conf(\cL)$ is compact, this infimum is achieved and is strictly greater than~$0$. (As a special case, if $\cL$ is not realizable, then it is vacuously globally noncrossing and its global minimum feature size is $+\infty$.)
\end{definition}

\begin{definition}[Combinatorial Embeddings, Corners]
  \label{def:combinatorial-embedding}
  A \term{combinatorial embedding} $\sigma$ of a graph $G$ consists of a cyclic ordering $\sigma_v$ of $v$'s incident edges for each vertex $v\in V(G)$. A noncrossing configuration $C$ of $G$ \term{agrees with $\sigma$} if the counterclockwise cyclic ordering of segments $C(v)C(w)$ around $C(v)$ matches $\sigma_v$ for each vertex $v$.

  Whenever edge $v u$ is followed by $v w$ in $\sigma_v$, the two-edge path $\Lambda = u v w$ is a \term{corner} of $\sigma$ at $v$; in the special case where $v$ is incident to exactly one edge $v u$, there is a single corner, $u v u$, at $v$.
  If $C$ is a noncrossing configuration of $G$ agreeing with $\sigma$, the \term{angles} of $C$ at $v$ are the angles $\angle C(\Lambda) := \angle C(u)C(v)C(w)$ for each corner $\Lambda = u v w$. When there is only one corner $u v u$ at $v$, we define $\angle C(u)C(v)C(u)$ to have measure $360^\circ$. With this convention, the angles of $C$ at $v$ add to $360^\circ$ no matter the degree of $v$. (Recall that all graphs are connected, so $\deg(v) \ge 1$.)
\end{definition}

\subsection{Constrained Linkages}
\label{sec:constrained-linkages}

We will make use of a number of special-purpose ``constraints'' or ``annotations'' that may be attached to linkages to artificially modify their behavior, such as ``rigid constraints'' that  ``rigidify'' a subgraph into a chosen configuration while allowing the rest of the linkage to move freely. These annotations do not affect the linkage itself; instead, they merely indicate which configurations of the linkage they consider acceptable.
The language of constraints allows us to separate a desired \emph{effect} from the \emph{implementation} or \emph{construction} that enforces that effect. For example, Sections~\ref{sec:noncrossing} through~\ref{sec:matchstick} develop three different techniques to force subgraphs to remain rigid---three different ``implementations'' of the rigidifying constraint---allowing a majority of the work, namely the Main Theorem (Theorem~\ref{thm:main-theorem}), to be reused in all three contexts. We now define constraints in general, and the rigid constraint in particular, more formally.

\begin{definition}
  \label{def:constraint}
  A \term{constraint}~$\Con$ on an abstract linkage~$\cL$ is specified by a subset of the configuration space, $\Con\subseteq\Conf(\cL)$, and we say the configurations~$C\in \Con$ \term{satisfy} constraint~$\Con$. A \term{constrained linkage}~$\cL$ is an abstract linkage~$\cL_0$ together with a finite set~$K$ of constraints on~$\cL_0$, and the \term{constrained configuration space} is defined as $\Conf(\cL) := \Conf(\cL_0)\cap\bigcap_{\Con\in K} \Con$. In other words, constrained linkage~$\cL$ simply ignores any configurations of~$\cL_0$ that don't satisfy all of its constraints.

  All terms discussed in Section~\ref{sec:linkage-definitions}---realizability, rigidity, global rigidity, etc.---apply equally well to constrained linkages via their \emph{constrained} configuration space.
\end{definition}

\begin{definition}
  A \term{rigid constraint} $\RigidCon_\cL(H,C_H)$ on a linkage $\cL = (G,W,P)$ is specified by a connected subgraph $H\subseteq G$ (whose edge lengths match those of $G$) together with a configuration $C_H$ of $H$. A configuration $C\in\Conf(\cL)$ satisfies the rigid constraint when $C$ induces a configuration $C|_H$ on $H$ that is congruent to the given $C_H$, i.e., differs only by a (possibly orientation-reversing) Euclidean transformation. When a constrained linkage~$\cM$ possesses constraint $\RigidCon_\cM(H,C_H)$, we say $(H,C_H)$ is a \term{rigidified subgraph} of $\cM$. A constrained linkage all of whose constraints are rigid constraints is called a \term{partially rigidified linkage}.
\end{definition}

Other constraint types---the noncrossing constraint, angle constraint, and sliceform constraint---will be introduced as they are needed, in Sections~\ref{sec:matchstick} and~\ref{sec:defining-extended-linkages}.

\subsection{Drawing with (Constrained) Linkages}
\label{sec:drawing-with-linkages}

\begin{definition}[Linkage Trace and Drawing]
  \label{def:trace}
  For a (possibly constrained) linkage $\cL$ and a tuple $X = (v_1,\ldots,v_k)$ of distinct vertices of $\cL$, the \term{trace} of $X$ is defined as the image $\pi_X(\Conf(\cL))\subset(\bR^2)^k$, where $\pi_X$ is the projection map sending a configuration $C\in\Conf(\cL)$ to $\pi_X(C) := (C(v_1),\ldots,C(v_k))$. This trace is semialgebraic, and if $\Conf(\cL)$ is compact, the trace is also compact. A linkage $(\cL,X)$ is said to \term{draw} its trace, and a set $R\subseteq(\bR^2)^k$ is \term{drawable (by a linkage)} if it can be expressed as the trace of some $k$ vertices of a linkage.
\end{definition}

The term ``draw'' is somewhat misleading, because a linkage might not be able to ``draw'' its entire trace through a single continuous motion. For example, some linkages have a disconnected trace, such as the linkage $\cT$ formed by an equilateral triangle $abc$ pinned at $a$ and $c$, so that $X=\{b\}$ draws two separate points corresponding to the two possible triangle orientations. Even some linkages with connected traces cannot draw their trace with a continuous motion: for example, if we modify $\cT$ by adding a longer edge $b d$, then $d$ draws two intersecting circles, even though it can only draw one circle at a time in a continuous motion (depending on the orientation of triangle $abc$). %

Because of the complicated relationship between a linkage and its trace, some types of drawings have been singled out as particularly nice:

\begin{definition}[Liftable Drawing]
  Say $(\cL,X)$ draws its trace \term{liftably} if the map $\pi_X$ has the \term{path lifting property}: for any configuration $C\in\Conf(\cL)$ and path $\gamma:[0,1]\to\pi_X(\Conf(\cL))$ in the trace starting at $\gamma(0)=\pi_X(C)$, there is a path $\overline \gamma:[0,1]\to\Conf(\cL)$ starting at $\overline \gamma(0) = C$ and lifting $\gamma$, meaning $\gamma = \pi_X\circ\overline\gamma$.
\end{definition}

For example, node $d$ of linkage $\cT$ above does not draw its trace \emph{liftably}, because $d$ cannot move from one circle to the other through a continuous motion of $\cT$. By contrast, node $c$ \emph{does} draw liftably (vacuously, as there are no nontrivial paths in $c$'s trace), even though $c$'s trace is disconnected.

Liftable drawing was introduced in~\cite{Abbott-2008} (under the name ``continuous'' drawing) for its usefulness in studying linkage rigidity. Indeed, if $(\cL,X)$ draws liftably and a point $p\in\pi_X(\Conf(\cL))$ is \emph{not} isolated in the trace, then any configuration $C$ with $\pi_X(C) = p$ is \emph{not} rigid, because a nontrivial continuous path beginning at $p$ can be lifted to a nontrivial motion beginning at $C$.

Though we have changed the name from ``continuous'' to ``liftable'' for greater clarity between \emph{continuous} motions and \emph{liftable} drawings, the two notions are closely related: indeed, liftability implies that each \emph{connected component} of the trace can be drawn with a continuous motion. Note, however, that the implication does not go the other way. For example, the central vertex $v$ in Watt's famous linkage traces a lemniscate, and can do so with a continuous motion. However,  a path through the trace that ``turns a corner'' at the central intersection point cannot be lifted to a motion of the full linkage, so $v$ does not draw its lemniscate liftably.

The complementary notion is \emph{rigid} drawing:

\begin{definition}[Rigid Drawing]
  Say $(\cL,X)$ draws its trace \term{rigidly} if the map $\pi_X$ has finite fibers, i.e., for any $p\in\pi_X(\Conf(\cL))$, there are only finitely many configurations $C\in\Conf(\cL)$ with $\pi_X(C) = p$.
\end{definition}

For example, node $d$ of linkage $\cT$ draws its trace rigidly, because there are at most two configurations corresponding to each possible position of $d$. By contrast, node $c$ does \emph{not} draw rigidly, because infinitely many configurations exist for each position of $c$.

Rigid drawing was also introduced in~\cite{Abbott-2008}, for similar reasons: if $(\cL,X)$ draws \emph{rigidly} and if $p$ is isolated in $\pi_X(\Conf(\cL))$, then any configuration $C$ with $\pi_X(C) = p$ is rigid, because the discrete set $\pi_X\inv(p)$ contains no nonconstant continuous paths.

As in~\cite{Abbott-2008}, we are especially interested in cases where $(\cL,X)$ draws both liftably and rigidly, as that creates a strong correlation between properties of the \emph{trace} and the linkage's own rigidity or flexibility. We also introduce a stronger notion:

\begin{definition}[Perfect Drawing]
  If the map $\pi_X$ is a \emph{homeomorphism} between $\Conf(\cL)$ and the trace, we say $(\cL,X)$ draws its trace \term{perfectly}.
\end{definition}

Perfect drawing easily implies liftable and rigid drawing, but it is even more restrictive. It is especially useful for \emph{parametrizing} a linkage's configuration space: if $(\cL,X)$ draws perfectly, each configuration of $\cL$ is uniquely and continuously determined solely from the locations of vertices in $X$.

Finally, we will often wish to create a linkage that ``behaves like'' a related linkage, which we formalize as follows:

\begin{definition}[Linkage Simulation]
  When $(\cL,X)$ draws precisely the full configuration space $\Conf(\cM)$ of another linkage~$\cM$, we say that $(\cL,X)$ \term{simulates $\cM$}. It may \term{liftably, rigidly, or perfectly simulate~$\cM$} if it draws $\Conf(\cM)$ in the corresponding manner.
\end{definition}

\subsection{Specification of the Main Theorem}
\label{sec:main-theorem-spec}

For a collection $F = \{f_1,\ldots,f_s\}$ of polynomials in $\bR[x_1,y_1,\ldots,x_m,y_m] = \bR[\widevec{xy}]$, the \term{algebraic set} (or \term{algebraic variety}) defined by $F$ is the set of common zeros,
\begin{equation*}
  Z(F) := \{\widevec{xy}\in\bR^{2m}\mid f_1(\widevec{xy}) = \cdots = f_s(\widevec{xy}) = 0\}.
\end{equation*}
The primary technical construction in this paper builds a globally noncrossing, partially rigidified linkage~$\cL(F)$ that draws the algebraic set $Z(f_1,\ldots,f_s)\subseteq \bR^{2m}$, or at least a large enough piece thereof, up to a translation of $\bR^{2m}$. This translation is necessary: without it, some algebraic sets would require the drawing vertices in $X$ to collocate in some configurations, precluding the possibility of global noncrossing. For example, two distinct vertices that draw the torus $S^1\times S^1$, where $S^1 = \{(x,y)\in\bR^2\mid x^2+y^2=1\}$, would be forced to intersect some of the time (e.g., at $(1,0)$). Allowing translation, this locus is easily drawable by two unit-length edges, each pinned by one endpoint at $(0,0)$ and $(3,0)$, respectively.

We are now prepared to precisely specify the properties of this construction, from which the results listed in Table~\ref{tab:results} follow as corollaries. We thoroughly detail these properties here, so that the corollaries may be derived solely from Theorem~\ref{thm:main-theorem}'s statement without referring to the specifics of its proof (with one small exception, discussed in Section~\ref{sec:matchstick-universality-full}). This also allows for maximal reuse: the commonalities in our arguments for our three linkage contexts---unconstrained globally noncrossing linkages in Section~\ref{sec:noncrossing}, unit-distance linkages in Section~\ref{sec:unit}, and matchstick linkages in Section~\ref{sec:matchstick}---have been unified and generalized into Theorem~\ref{thm:main-theorem}, so only features unique to each context need to be discussed in Sections~\ref{sec:noncrossing}--\ref{sec:matchstick}.

The Main Theorem is divided into three parts because it must be used in subtly different ways by the four types of results we seek. Hardness of realizability requires a polynomial-time construction of an abstract linkage that draws $Z(f_1,\ldots,f_s)$ \emph{without knowing whether the resulting configuration space is empty}, whereas proving hardness of rigidity and global rigidity requires the polynomial-time construction of a linkage \emph{together with a known configuration}. We thus separate these into different Parts of Theorem~\ref{thm:main-theorem} with slightly different assumptions about the input polynomials $f_j$ (Part~II for realizability, Part~III for rigidity and global rigidity). When proving universality, we must prove \emph{existence} of a linkage to draw any compact semialgebraic set, but the coefficients of the polynomials defining this set may be irrational or non-algebraic, as might the edge lengths and coordinates of the resulting linkage, so we isolate this in Part~I, away from algorithmic and efficiency concerns.

If the input set of polynomials $F$ has $|F|=s$ real polynomials, each of total degree $d$ in the $2m$ variables $(x_1,y_1,\ldots,x_m,y_m)$, then the number of coefficients is $s\cdot \binom{2m+d}{d}$ (in dense representation). The number of vertices and edges of the resulting linkage will be bounded by a polynomial in the related quantities $m^d$, $d^d$, and~$s$. We do not attempt bounds that are finer tuned depending on whether $F$ is sparse or somehow otherwise simpler than parameters $m,d,s$ might indicate. Indeed, sparse polynomial input does not lead to noticeable efficiency gains with our algorithm, due to a change of coordinates (detailed in Section~\ref{sec:detailed-overview}) that can turn sparse polynomials into dense ones.

Similarly, numerator bounds on edge lengths and coordinates are written in terms of~$M$, an upper bound on the absolute value of coefficients of the input polynomials. We wish to emphasize that $M$ bounds the number of \emph{unary} digits of coefficients, not binary digits. Nowhere in this paper do we measure the number of binary digits of an integer: the \emph{magnitude}, \emph{absolute value}, or \emph{size} of an integer always refers to its length in unary.

\begin{theorem}
  \label{thm:main-theorem}
  
  \begin{descriptionflush}
    
  \item[Part I.]
    Take as input a collection of polynomials $F = \{f_1,\ldots,f_s\}$, each in $\bR[x_1,y_1,\ldots,x_m,y_m]$ with total degree at most $d$. Then we may construct a partially rigidified linkage $\cL = \cL(F)$ that draws, up to translation, a bounded portion of the algebraic set $Z(F)$: specifically, there is a translation $T$ on $\bR^{2m}$ and a subset $X$ of $m$ vertices of $\cL$ such that
    \begin{equation*}
      T(Z(F)\cap[-1,1]^{2m}) \subseteq \pi_X(\Conf(\cL)) \subseteq T(Z(F)).
    \end{equation*}
    Furthermore, there is a constant integer $D$ depending only on $s,m,d$ (but independent of the coefficients in input~$F$) such that:
    \begin{enumerate}
    \item\label{thmpart:rigid-and-liftable}
      Vertex list $X$ draws this trace liftably and rigidly.
    \item\label{thmpart:num-vertices}
      The number of vertices and edges in $\cL$ is $O(\poly(m^d,d^d,s))$.
    \item\label{thmpart:gmfs}
      Each edge of $\cL$ has length at least $1/D$, and $\cL$ is globally noncrossing with global minimum feature size at least $1/D$.
    \item\label{thmpart:orthogonal-trees}
      For each constraint $\RigidCon_\cL(H,C_H)$ on $\cL$, $H$ is a tree that connects to $G\setminus H$ only at leaves of $H$, and configuration $C_H$ has all edges parallel to the $x$- or $y$-axes. Tree $H$ has at least three noncollinear vertices, and it does not flip: for each configuration $C$ of $\cL$, configuration $C_H$ and the induced configuration $C|_H$ have the same orientation. Each edge of $G$ is contained in at most one rigidified subtree~$(H,C_H)$.
    \item\label{thmpart:corners}
      There is a combinatorial embedding $\sigma$ of $G$ such that every configuration $C\in\Conf(\cL)$ is noncrossing and agrees with~$\sigma$. Furthermore, for any vertex $v$ of degree at least $2$ that is not internal to any rigidified tree, each corner $\Lambda$ at $v$ has angle $\angle C(\Lambda)\in(60^\circ,240^\circ)$.
    \item\label{thmpart:pins}
      Linkage $\cL$ has precisely $|P|=3$ pinned vertices, which belong to a single rigidified tree $(H,C_H)$ and are not collinear in $C_H$.
    \end{enumerate}
    
  \item[Part II.]
    Suppose polynomials $f_j$ have integer coefficients, each bounded in absolute value by~$M$. Then we may bound the complexity of $\cL$ as follows:
    \begin{enumerate}[resume]
    \item\label{thmpart:rational-edge-lengths}
      All edge lengths in $\cL$ belong to $\frac{1}{D}\cdot\bZ$ and have size $O(\poly(m^d,d^d,s,M))$. Edges not contained in a rigidified subtree have lengths at most $D$.
    \item\label{thmpart:poly-time}
      Constrained linkage~$\cL$, the set~$X$ of vertices, translation~$T$, and combinatorial embedding $\sigma$ may be constructed from~$F$ deterministically in time $O(\poly(m^d,d^d,s,M))$.
    \end{enumerate}

  \item[Part III.]
    Finally, if the polynomials $f_j$ each satisfy $f_j(\vec 0) = 0$, we may additionally compute an initial configuration $C_0$ satisfying:

    \begin{enumerate}[resume]
    \item\label{thmpart:rational-coords}
      All coordinates of $C_0$ belong to $\frac{1}{D}\cdot\bZ$ and have absolute value $O(\poly(m^d,d^d,s,M))$.
    \item\label{thmpart:uniqueness}
      $C_0$ is the \emph{only} configuration of $\cL$ that projects to $T(\vec 0) \in \pi_X(\Conf(\cL))$.
    \item\label{thmpart:orthogonal-induced-trees}
      For each rigidified subtree $(H,C_H)$, $C_0$ induces a configuration of $H$ in which all edges are parallel to the $x$- or $y$-axes. (Edges not in any rigidified subtree need not be axis-aligned.)
    \item\label{thmpart:poly-time-config}
      $C_0$ may also be computed deterministically in time $O(\poly(m^d,d^d,s,M))$.
    \end{enumerate}
    
  \end{descriptionflush}
\end{theorem}

\begin{convention}[Orientation]
  Property~\ref{thmpart:orthogonal-trees} above refers to the \emph{orientation} of a configured tree. In this paper, the word \emph{orientation} always refers to handedness, not to the angle or direction of an object. So a reflection reverses orientation; a rotation changes angle but preserves orientation.
\end{convention}

\subsection{Roadmap}
\label{sec:roadmap}

The rest of the paper is organized as follows: After gathering a few preliminary facts about real (semi)algebraic sets and the class $\exists\bR$ in Section~\ref{sec:preliminaries}, we use the Main Theorem (Theorem~\ref{thm:main-theorem}) to prove that graph realizability, rigidity, and global rigidity are $\exists\bR$-complete or $\forall\bR$-complete, and that linkages are universal at drawing semialgebraic sets, in each of three separate contexts: for globally noncrossing graphs/linkages in Section~\ref{sec:noncrossing}, for unit-distance (or $\{1,2\}$-distance) graphs/linkages in Section~\ref{sec:unit}, and for matchstick graphs/linkages in Section~\ref{sec:matchstick}. After all of that, in Section~\ref{sec:main-construction}, we finally provide the gory details of the Main Construction itself.

\headingboldmath
\section{Preliminaries on Semialgebraic Sets and $\exists\bR$}
\label{sec:preliminaries}

\subsection{Semialgebraic Sets as Projections of Algebraic Sets}
\label{sec:semialgebraic-sets}

Here we prove some elementary facts about representing certain semialgebraic sets as projections of (semi)algebraic sets having specific forms, which will be used in proofs of linkage universality (not in complexity results). For a general introduction to real semialgebraic geometry, see~\cite{bochnak-coste-roy}.

A \term{basic semialgebraic set} in $\bR^k$ is a set of the form $\{\vec x\in\bR^k\mid f_1(\vec x) = \cdots = f_s(\vec x) = 0, g_1(\vec x) > 0,\ldots, g_r(\vec x) > 0\}$ for polynomials $f_i$ and $g_j$. Any \term{semialgebraic set} can be written as a finite union of basic semialgebraic sets~\cite[Prop.~2.1.8]{bochnak-coste-roy}; in fact, this is one of several equivalent ways to define semialgebraic sets.

\begin{lemma}
  \label{lem:semialg-projection}
  Any bounded semialgebraic set $R\subset\bR^k$ can be expressed as the projection of some bounded basic semialgebraic set onto the first $k$ coordinates.
\end{lemma}

\begin{proof}
  We may assume $R$ is nonempty. Any semialgebraic set can be written as a finite union of basic semialgebraic sets~\cite[Prop.~2.1.8]{bochnak-coste-roy}, so write $R = \bigcup_{j=1}^t R_j$, where each $R_j$ is a nonempty basic semialgebraic set. Let $\vec x, \vec x_1,\ldots, \vec x_t$ each denote a variable point in $\bR^k$, and define $R'\subset(\bR^k)^{t+1}$ as the \emph{basic} semialgebraic set defined by the conditions $\vec x_j\in R_j$ for $1\le j\le t$, as well as
  \begin{equation*}
    f(\vec x, \vec x_1, \ldots, \vec x_t) := \prod_{j=1}^t |\vec x - \vec x_j|^2 = 0.
  \end{equation*}
  This last equation exactly stipulates that $\vec x = \vec x_j$ for some $1\le j\le t$, so $\pi(R') = \bigcup_{j=1}^t R_j = R$, where $\pi$ denotes the projection onto the coordinates of $\vec x$. Set $R'$ is bounded because each $R_j$ is bounded.
\end{proof}

\begin{lemma}
  \label{lem:semialg-projection-ne0}
  Any bounded semialgebraic set $R\subset\bR^k$ can be expressed as the projection of a bounded set of the form $\{\vec x\in\bR^m\mid f_1(\vec x) = \cdots = f_s(\vec x) = 0, g_1(\vec x) \ne 0, \ldots, g_r(\vec x)\ne 0\}$, for polynomials $f_i$ and $g_j$.
\end{lemma}

\begin{proof}
  By Lemma~\ref{lem:semialg-projection}, we may assume $R$ is a basic semialgebraic set, $R = \{\vec x\in\bR^k\mid p_1(\vec x) = \cdots = p_s(\vec x), q_1(\vec x) > 0, \ldots, q_r(\vec x) > 0\}$. Because $R$ is bounded, by scaling the polynomials if necessary, we may assume $|p_i(\vec x)|\le 1$ and $|q_j(\vec x)|\le 1$ for all $\vec x\in R$. Now introduce new real variables $\vec a = (a_1,\ldots,a_r)$ and define
  \begin{multline*}
    R' = \{(\vec x, \vec a)\in\bR^{k+r}\mid
    p_1(\vec x) = \cdots = p_s(\vec x) = 0,\\
    q_1(\vec x) = a_1^2, a_1\ne 0, \ldots,
    q_r(\vec x) = a_r^2, a_r\ne 0\}.
  \end{multline*}
  This set is contained in $R\times [-1,1]^r$ and is therefore bounded, and $R$ is the projection of $R'$ onto the coordinates of $\vec x$.
\end{proof}

A \term{basic closed semialgebraic set} has the form $\{\vec x\in\bR^k\mid f_1(\vec x) \ge 0, \ldots, f_s(\vec x)\ge 0\}$; if this set is also bounded, we call it a \term{basic compact semialgebraic set}.

\begin{lemma}
  \label{lem:compact-semialg-projection}
  Any compact semialgebraic set $R\subset\bR^k$ can be expressed as the projection of some compact algebraic set onto the first $k$ coordinates.
\end{lemma}

\begin{proof}
  Compact semialgebraic set $R$ can be written as a finite union of basic compact semialgebraic sets~\cite[Thm.~2.7.2]{bochnak-coste-roy}, and the same proof used in Lemma~\ref{lem:semialg-projection} shows that $R$ can be written as a coordinate projection of a basic compact semialgebraic set. (The condition $f = 0$ can be expressed as $f\ge 0$ and $-f\ge 0$.) So it suffices to show that a compact \emph{basic} semialgebraic set is a coordinate projection of a compact algebraic set.

  Assume $R$ is a compact basic semialgebraic set, $$R=\{\vec x\in\bR^n\mid f_1(\vec x) \ge 0,\ldots,f_s(\vec x)\ge 0\}.$$
  Choose new variables $\vec y\in\bR^s$,
  and define $R' = \{(\vec x,\vec y)\mid f_i(\vec x) = y_i^2\text{ for }1\le i\le s\}$. If $\pi$ is the projection onto coordinates of $\vec x$, then one may check that $\pi(R') = R$. Region $R'$ is bounded because $R$ itself is bounded, so this completes the proof. 
\end{proof}

Semialgebraic sets are closed under projection \cite[Thm.~2.2.1]{bochnak-coste-roy} (a form of the Tarski--Seidenberg principle), as are compact semialgebraic sets.  The lemmas above imply that algebraic sets and basic semialgebraic sets are not closed under projection.

\subsection{Existential Theory of the Reals}
\label{sec:exists-r}

The class $\exists\bR$ is the complexity class consisting of all problems polynomially (Karp) reducible to the Existential Theory of the Reals (\ETR), which is the language consisting of true formulas of the form $(\exists x_1,\ldots,x_n\in\bR)\phi(x_1,\ldots,x_n)$, where $\phi$ is a quantifier-free predicate over real variables $x_1,\ldots,x_n$ using arithmetic symbols $+$, $-$, $\times$, $0$, and $1$, logical predicates $<$, $\le$, and $=$, and boolean operators $\wedge$, $\vee$ and $\neg$. It is known that $\text{NP}\subseteq\exists\bR\subseteq\text{PSPACE}$, though neither inclusion is known to be strict. The first inclusion is simple, as \SAT{} may be encoded with polynomial constraints, so any problem complete for $\exists\bR$ is automatically NP-hard. The second inclusion is a nontrivial theorem of Canny~\cite{Canny-1988-pspace} and is the tightest known upper bound on the hardness of~$\exists\bR$.

We now prove the $\exists\bR$ hardness of two problems used in the forthcoming reductions. The problem $\HtwoN$ (probably short for \emph{Hilbert's homogeneous Nullstellensatz}~\cite{Schaefer-2013}) asks whether homogeneous polynomials $f_1,\ldots,f_s\in\bZ[x_1,\ldots,x_k]$ have a nonzero common root in $\bR^k$. It was introduced in~\cite{Koiran-2000} and was shown to be $\exists\bR$-complete in~\cite{Schaefer-2013}, even when all polynomials have degree $4$. We wish to reduce from a stronger version of this problem where, additionally, all coefficients are in $\{0,\pm 1,\pm 2\}$. To prove hardness of this problem, only a slight modification of Schaefer's original argument is necessary. Our starting point is a strengthening of the well-known $\CommonZero$ problem:

\begin{lemma}\label{lem:deg-2-hardness}
  The $\CommonZero$ problem---determining whether integer-coefficient polynomials $f_1,\ldots,f_s\in\bZ[x_1,\ldots,x_n]$ have a common solution in $\bR^n$---is $\exists\bR$-complete, even when the polynomials must have total degree at most $2$ and coefficients in $\{-1,0,1\}$.
\end{lemma}

\begin{proof}
  The problem is a subproblem of \ETR{}, so it is certainly in $\exists\bR$. For hardness, we reduce from the \PointConfiguration{} problem: given a choice of orientation (clockwise, counterclockwise, or collinear) for each triple $1\le i < j < k \le n$, determine whether there exists a configuration of points $p_1,\ldots,p_n$ in the plane so that each triangle $p_ip_jp_k$ has the specified configuration. This is equivalent (in fact, projectively dual) to the more well-known \Stretchability{} problem, shown to be $\exists\bR$-complete by Mn\"{e}v~\cite{Mnev-1988}; see also~\cite{shor}.

  If we write $p_i = (p_{i,1},p_{i,2})$ in coordinates, then the triangle's signed area is given by
  \begin{equation}\label{eqn:area-def}
    2\cdot\area(p_ip_jp_k) = \det(p_i,p_j,p_k) := \det\begin{pmatrix}p_{i,1} & p_{i,2} & 1\\p_{j,1} & p_{j,2} & 1\\p_{k,1} & p_{k,2} & 1\end{pmatrix},
  \end{equation}
  which is a homogeneous, degree $2$ polynomial in the coordinate variables with coefficients $\pm 1$. For the triples with collinear configuration, we simply include the polynomial $\det(p_i,p_j,p_k)=0$. For those with counterclockwise orientation, include the polynomials $\det(p_i,p_j,p_k)=a_{i,j,k}^2$ and $a_{i,j,k}b_{i,j,k}=1$, for new variables $a_{i,j,k}$ and $b_{i,j,k}$. Similarly, for clockwise triples, include $\det(p_i,p_j,p_k)=-a_{i,j,k}^2$ and $a_{i,j,k}b_{i,j,k}=1$. It is clear that the point configuration is realizable if and only if the resulting system of polynomials with coefficients in $\{-1,0,1\}$ and degree at most $2$ has a real solution.
\end{proof}

\begin{theorem}
  \label{thm:htwon'}
  The problem $\HtwoN$ remains $\exists\bR$-complete even when the input polynomials have degree~$4$ and all coefficients lie in $\{0,\pm1,\pm2\}$.
\end{theorem}

\begin{proof}
  The proof by Schaefer~\cite[Lem.~3.9, Cor.~3.10]{Schaefer-2013} for the hardness of $\HtwoN$ with degree $4$ starts with a collection of polynomials of degree at most $2$ and transforms them into a collection of homogeneous degree $4$ polynomials that have a nontrivial solution if and only if the original collection has any solution. These transformations, when applied without modification to polynomials with coefficients in $\{0,\pm1\}$, return a collection of homogeneous, degree $4$ polynomials whose coefficients lie in $\{0,\pm1,\pm2\}$. So this result follows from Schaefer's proof paired with Lemma~\ref{lem:deg-2-hardness}.
\end{proof}

We can also provide another strengthening of the $\CommonZero$ problem:

\begin{theorem} \label{thm:common-zero}
  The $\CommonZero$ problem is $\exists\bR$-complete, even when the given polynomials $f_1,\ldots,f_s\in\bZ[x_1,\ldots,x_m]$ have total degree at most~$4$, all coefficients are in $\{0,\pm1,\pm2\}$, and all common zeros, if any, are promised to lie in $[-1,1]^m$.
\end{theorem}

\begin{proof}
  We reduce from $\HtwoN$, as strengthened in Theorem~\ref{thm:htwon'}: we are given an instance $F = \{f_1,\ldots,f_s\}$ of $\HtwoN$ with degree $4$ polynomials whose coefficients lie in $\{0,\pm1,\pm2\}$. Return the instance $F' = F\cup\{g\}$ of $\CommonZero$, where $g(x_1,\ldots,x_m) = x_1^2 + \cdots + x_m^2 - 1$. Polynomial $g$ guarantees that $Z(F')\subset \overline B(\vec 0,1)\subset [-1,1]^m$, and all coefficients in $F'$ are still in $\{0,\pm1,\pm2\}$, so $F'$ has the required form.

  If $Z(F)$ contains only $\vec 0$, then $Z(F')$ is empty. On the other hand, if $\vec a\in Z(F)\setminus\{\vec 0\}$, then $\vec a / |\vec a|\in Z(F')$, so the reduction is correct.
\end{proof}

\section{Globally Noncrossing Graphs and Linkages}
\label{sec:noncrossing}

Using the statement of the Main Theorem (Theorem~\ref{thm:main-theorem}) and the hardness results from Section~\ref{sec:exists-r}, we show here that deciding realizability, rigidity, or global rigidity of a globally noncrossing graph/linkage is $\exists\bR$-complete or $\forall\bR$-complete, and that globally noncrossing linkages can draw any compact semialgebraic set in the plane.

\subsection{Rigidifying Polygons}

To simulate a rigidified tree, we will construct a globally rigid graph in the shape of a slight thickening of the tree. To that end, we provide here a general method that constructs a globally rigid triangulation, with Steiner points, of any simple polygon.

\begin{lemma}
  \label{lem:rigidify-quad}
  Any simple quadrilateral $A = A_1A_2A_3A_4$ has a globally rigid triangulation with four triangles and one Steiner point as in Figure~\ref{fig:rigidify-quad}.
\end{lemma}

\begin{proof}
  If any of the four angles of $A$ is $180^\circ$ or greater, relabel so that this angle is at $A_3$. This means any point in the interior of $A$ within some distance $d > 0$ from $A_1$ is visible to all four of $A$'s vertices inside the quadrilateral.

  Let $P$ be the point on side $A_1A_2$ at distance $d/2$ from $A_1$, so that the three angles $\angle A_2PA_3$, $\angle A_3PA_4$ and $\angle A_4PA_1$ lie strictly between $0^\circ$ and $180^\circ$. By continuity, there is a point $Q$ near $P$ in the interior of $A$ such that $\angle A_1QA_2$ is close to $180^\circ$ and the three angles $\angle A_2QA_3$, $\angle A_3QA_4$ and $\angle A_4QA_1$ lie strictly between $180^\circ - \angle A_1QA_2$ and $180^\circ$. We may also assume $Q$ is close to $A_1$ and is therefore visible to all four vertices of $A$.

  We claim that the triangulation $T$ using Steiner point~$Q$ is globally rigid. If not, then by Kawasaki's criterion~\cite{kawasaki1989relation} for flat-foldable single vertex crease patterns, it must be the case that some subset of the four angles at $Q$ sum to $180^\circ$. But $Q$ was chosen to ensure that this condition is false, since adding $\angle A_1QA_2$ to any of the other three angles results in more than $180^\circ$.
\end{proof}

\begin{figure}[hbt]
  \centering
  \mysubfigure
  {.48\textwidth}
  {\includegraphics{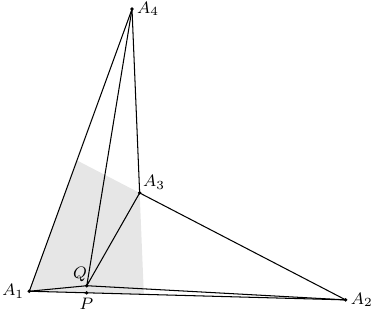}}
  {All points in the gray region are visible to the four vertices of quadrilateral $A$. Any point $Q$ close enough to $P$ renders this $4$-triangle triangulation of $A$ globally rigid.}
  {fig:rigidify-quad}
  \hfill
  \mysubfigure
  {.48\textwidth}
  {\includegraphics{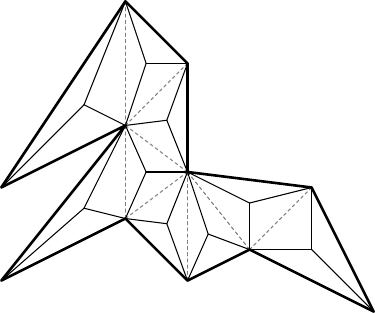}}
  {Rigidifying each quadrilateral in this decomposition using Lemma~\ref{lem:rigidify-quad} rigidifies the entire graph, as shown in Lemma~\ref{lem:rigidifying-polygons}.}
  {fig:rigidify-poly}
  \caption{Rigidifying a simple quadrilateral (left) and a general simple polygon (right) by a triangulation with Steiner points, as in Lemmas~\ref{lem:rigidify-quad} and~\ref{lem:rigidifying-polygons}.}
\end{figure}

\begin{lemma}
  \label{lem:rigidifying-polygons}
  If $A$ is a simple $n$-sided polygon, and points $P_1,\ldots,P_m$ interior to $A$ are specified, we may construct a triangulation of $A$ with $2n+6m-5$ Steiner points that is globally rigid as a configured graph and has a Steiner point at each of $P_1,\ldots,P_m$.
\end{lemma}

\begin{proof}
  Let $T_1$ be a triangulation of $A$ with Steiner points at $P_1,\ldots,P_m$; there are $n+2m-2$ triangles and $n+3m-3$ interior edges in this decomposition. Subdivide each triangle at its centroid to obtain triangulation $T_2$, and delete the original $n+3m-3$~interior edges from $T_2$ to obtain a subdivision $T_3$ of polygon $A$ into triangles and quadrilaterals, as illustrated in Figure~\ref{fig:rigidify-poly} (for the special case $m=0$). Finally, apply Lemma~\ref{lem:rigidify-quad} to each quadrilateral in $T_3$ to obtain the final triangulation $T$. The three triangles or quadrilaterals meeting at each centroid are individually globally rigid, so the union of these three pieces is globally rigid again by Kawasaki's criterion: three creases around a vertex, no two collinear, are insufficient for a nontrivial single-vertex flat folding. Applying this reasoning around each centroid shows that the entire triangulation $T$ is globally rigid, as desired.
\end{proof}

This Lemma will be used when simulating rigidified subtrees in Section~\ref{sec:noncrossing-simulation}.

\subsection{Rigidifying Polyominoes}

When the partially rigidified tree $(H,C_H)$ that we wish to simulate has \emph{integer} coordinates, we will use a more refined grid which uses only rational coordinates and small rational edge lengths. To accomplish this, we show in this section that any polyomino, after scaling up by a factor of $1440$, can be rigidified with only integer coordinates and constant-sized integer edge lengths. (A \term{polyomino} is a polygon with connected interior formed as the union of a finite set of squares in the standard unit-square tiling of~$\bR^2$.) We begin by rigidifying a single square:

\begin{figure}[hbt]
  \centering
  \begin{subfigure}[c]{.48\textwidth}
    \centering
    \includegraphics{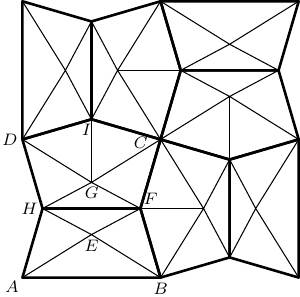}
    \caption{
      A globally rigid graph $\Gcell$ with integer coordinates and integer edge lengths in the shape of a $1440\times 1440$ square with indents.
      \label{fig:rigid-monomino}
    }
  \end{subfigure}
  \hfill
  \begin{subfigure}[c]{.48\textwidth}
    \centering
    \includegraphics{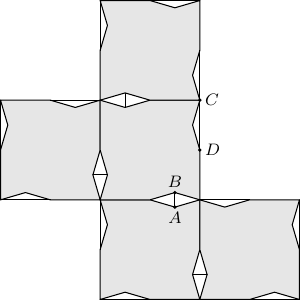}
    \caption{
      Multiple copies of $\Gcell$ can be joined into a globally rigid polyomino.
      \label{fig:thicken-with-polyomino}
    }
  \end{subfigure}
  \caption{Any polyomino made of $1440\times 1440$ squares can be turned into a globally rigid graph with integer coordinates and integer edge lengths.}
  \label{fig:ridigify-polyominoes}
\end{figure}

\begin{lemma}
  \label{lem:rigidify-square}
  The configured graph~$\Gcell$ shown in Figure~\ref{fig:rigid-monomino}, which has the shape of a $1440\times 1440$ square with small \term{indents} on the edges, is globally rigid. The vertex coordinates and edge lengths are all integers.
\end{lemma}

\begin{proof}
  To specify $\Gcell$ graph in more detail, the labeled vertices have coordinates
  \begin{gather*}
    A = (0,0), \qquad  B=(720,0), \qquad  C = (720,720), \qquad  D=(0,720), \qquad  E=(360,224),\\
    F=(615,360), \qquad  G=(360,496), \qquad  H=(105,360), \qquad  I=(360,825).
  \end{gather*}
  In particular, pentagon $HFCID$ and quadrilateral $ABFH$ have a vertical line of symmetry, and $AEB$ and $DGC$ are reflections of each other through $HF$. The edge lengths are
  \begin{gather*}
    AB=720,\qquad AE=424,\qquad AH=375,\qquad EF=289,\qquad GI=329,\qquad HF=510.
  \end{gather*}
  The rest of the coordinates may be computed by the $90^\circ$-degree rotational symmetry around $C$, and all distinct edge lengths are listed above. Note that $A,E,F$ are not quite collinear, and similarly for $H,G,C$.

  To show that $\Gcell$ is globally rigid, we again make repeated use of Kawasaki's criterion. First, the five triangles forming pentagon $HFCID$ (with Steiner point $G$) form a globally rigid subgraph: indeed it may be checked that $\angle HGF + \angle FGC = 180^\circ + \arcsin\frac{60}{901} > 180^\circ$ and $\angle HGF + \angle CGI = 180^\circ + \arcsin\frac{525}{15317} > 180^\circ$, so any subset of the five angles at $G$ that includes $\angle HGF$ cannot add to exactly $180^\circ$. The four triangles forming quadrilateral $ABFH$ form a globally rigid subgraph as well, because the crease pattern around $E$ is congruent to a subset of the crease pattern around $G$. The rotationally symmetric copies of pentagon $HFCID$ and $ABFG$ are likewise globally rigid. Finally, the quadrilateral and two pentagons meeting at~$F$ are globally rigid together because degree-$3$ crease patterns have no nontrivial flat foldings, and applying this reasoning four times around the square shows that all of $\Gcell$ is globally rigid.
\end{proof}

\begin{lemma}
  \label{lem:rigidifying-polyominoes}
  Any polyomino $P$ made of $1440\times 1440$ squares can be triangulated (allowing Steiner points and edge subdivision) into a globally rigid triangulation that has integer coordinates and constant-sized integer edge lengths.
\end{lemma}

\begin{proof}
  Place a copy of $\Gcell$ in each cell of polyomino $P$, swapping orientation for every other cell so that adjacent cells have aligned indents. A short edge of length $210$ within the indents---such as edge $A B$ in Figure~\ref{fig:thicken-with-polyomino}---renders each pair of adjacent cells globally rigid. Because $P$ is a polyomino which is connected via its edges, the whole assembly is thus globally rigid. Finally, indents along $P$'s boundary can be covered with edges of length $360$, as in Figure~\ref{fig:thicken-with-polyomino}.
\end{proof}

We show in the next section how to use this method to simulate integer-length rigidified subtrees.

\subsection{Simulation with Globally Noncrossing Linkages}
\label{sec:noncrossing-simulation}

\begin{construction}
  \label{con:noncrossing-main-theorem}
  Use notation as in Theorem~\ref{thm:main-theorem}: we are given a collection of polynomials $F = \{f_1,\ldots,f_s\}$ in $2m$ variables, each of total degree at most $d$, from which Theorem~\ref{thm:main-theorem} constructs a partially rigidified linkage $\cL = \cL(F)$.

  Under Part~I of Theorem~\ref{thm:main-theorem} (i.e., with no additional assumptions about $F$), we may construct a globally noncrossing linkage~$\cM = \cM(F)$ without constraints that perfectly simulates the constrained linkage~$\cL$.

  If the Part~II assumption holds, meaning the $f_j$ have integer coefficients with absolute value at most $M$, then $\cM$ may be constructed from $\cL$ in time $O(\poly(m^d,d^d,s,M))$, and furthermore, each edge length of $\cM$ is rational, with size at most $D$ and denominator dividing $28800D$.

  Finally, if the Part~III assumption also holds, i.e., each $f_j$ satisfies $f_j(\vec 0) = 0$ resulting in a configuration $C_0$ of $\cL$, then the corresponding configuration of $\cM$ has rational coordinates of magnitude $O(\poly(m^d,d^d,s,M))$ with denominators dividing $28800D$.
\end{construction}

\begin{figure}[h]
  \centering
  \mysubfigure
  {.48\textwidth}
  {\includegraphics{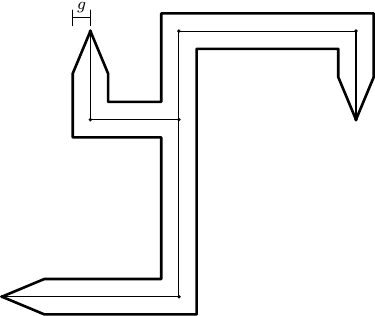}}
  {Thickening a partially rigidified subtree $(H,C_H)$ by radius $g$, with wedges at each leaf.}
  {fig:thicken-poly}%
  \hfill
  \mysubfigure
  {.48\textwidth}
  {\includegraphics{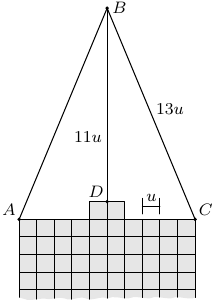}}
  {Rigidifying the thickened tree using a polyomino and $5$-$12$-$13$ triangles at the leaves.}
  {fig:thicken-poly-closeup}%
  \caption{Thickening a rigidified subtree $(H,C_H)$ into a globally rigid polygon.}
  \label{fig:simulate-globally-noncrossing}
\end{figure}

\begin{proof}

  Consider one of the rigid constraints $\RigidCon(H,C_H)$ on $\cL$, and draw a polygon $P$ that thickens $(H,C_H)$ by distance $g := 1/(4D)$ in each direction, with angled wedges smaller than $60^\circ$ at each leaf vertex, as shown in Figure~\ref{fig:thicken-poly}. Now apply the construction of Lemma~\ref{lem:rigidifying-polygons} to this polygon $P$ to obtain a globally rigid triangulation $T$ that has a Steiner point at each interior vertex of $C_H$, and replace subgraph $H$ of $\cL$ with this triangulation $T$. Because tree $(H,C_H)$ has at least three noncollinear vertices (Property~\ref{thmpart:orthogonal-trees}), globally rigid triangulation $T$ perfectly simulates the constraint $\RigidCon(H,C_H)$: configurations of $T$ are in bijection with configurations of rigidified tree $(H,C_H)$ when considered in isolation. (The Steiner points chosen for $T$ are there to satisfy the definition of \emph{simulate}: $T$ needs a vertex corresponding to each vertex of $H$.) Let $\cM$ be the linkage that results after performing this replacement for each rigid constraint on $\cL$. One rigidified tree was originally pinned in $\cL$ (by Property~\ref{thmpart:pins}); keep these three pins in the corresponding locations in the triangulation $T$ built from this tree. Because each tree is perfectly simulated, $\cM$ perfectly simulates $\cL$.

  To see that $\cM$ is globally noncrossing, note that Property~\ref{thmpart:corners} of Theorem~\ref{thm:main-theorem} guarantees that the $60^\circ$ wedges at leaf nodes do not intersect locally around their shared vertices, and because $g$ is less than half the minimum feature size of any configuration of $\cL$ (by Property~\ref{thmpart:gmfs}), the thickened trees do not intersect elsewhere.

  If Part~II holds, we construct globally rigid graph $T$ from each rigidified tree $(H,C_H)$ more carefully, by building a rigidified polyomino with Lemma~\ref{lem:rigidifying-polyominoes}, as follows. Each edge length in $\cL$ is an integer multiple of~$g$, by Property~\ref{thmpart:rational-edge-lengths}. Thicken each rigidified tree $(H,C_H)$ by distance $g$ into a polygon $P$ as above, where at each leaf of $(H,C_H)$, $P$ forms a wedge with angle $2\arcsin\frac{5}{13} < 60^\circ$. Now we can fill polygon $P$ with a polyomino made from cells of side-length $u = g/5$, where the wedge at each leaf is attached with three edges of length $13u$, $11u$, and $13u$ respectively, as shown in Figure~\ref{fig:thicken-with-polyomino}. By rigidifying the polyomino with Lemma~\ref{lem:rigidifying-polyominoes}, the resulting graph $T$ is a globally rigid thickening of tree $(H,C_H)$ that has edge lengths in $\frac{1}{t}\bZ\cap(0,\frac{1}{2}]$ and rational coordinates in $\frac{1}{t}\bZ$ of magnitude at most $O(\poly(m^d,d^d,s,M))$, where $t := 1440\cdot 5\cdot 4D = 28800D$. All edges of $\cL$ not contained in a rigidified tree, which appear unmodified in $\cM$, already have length at most $D$ by Property~\ref{thmpart:rational-edge-lengths}, proving the Part~II claim.

  Finally, for Part~III, first consider the nodes of $\cL$ that are not interior to a rigidified subtree, i.e., are either a leaf of one or more rigidified trees or are not incidient to any such tree (by Property~\ref{thmpart:orthogonal-trees}). These nodes of $\cL$ have corresponding nodes in $\cM$ configured in the same locations, so these nodes of $\cM$ are configured to have rational coordinates in $\frac{1}{D}\bZ$ by Property~\ref{thmpart:rational-coords}. Pairing this fact with Property~\ref{thmpart:orthogonal-induced-trees} implies that all remaining nodes of $\cM$, which are interior to rigidified thickened trees built from polyominoes, are indeed configured to have rational coordinates in $\frac{1}{t}\bZ$, proving the claim.
\end{proof}

\subsection{Hardness and Universality of Globally Noncrossing Linkages}
\label{sec:noncrossing-results}

We may finally prove the desired hardness and universality results about globally noncrossing graphs and linkages.

\begin{theorem}[Hardness of Globally Noncrossing Realizability]
  \label{thm:noncrossing-realizability}
  Deciding whether a given abstract weighted graph $\cG$ is realizable, even when $\cG$ is promised to be globally noncrossing and to have constant-sized integer edge lengths, is $\exists\bR$-complete.
\end{theorem}

\begin{proof}
  Membership in $\exists\bR$ is evident. For hardness, we reduce from \CommonZero\ as strengthened in Theorem~\ref{thm:common-zero}. Given an instance $F = \{f_1,\ldots,f_s\}$, which we may assume have $d=4$ and $M=2$, apply Construction~\ref{con:noncrossing-main-theorem} (Part II) to obtain a globally noncrossing linkage $\cM$ whose edge lengths are rational with size at most $D$ and denominators dividing $t := 28800D$. Let $\cG$ be the weighted graph that results by unpinning $\cM$'s three pins; the output of this reduction is the scaled graph $t\cdot\cG$, whose edge lengths are \emph{integers} bounded by $t\cdot D = O(1)$.

  It remains to show why $t\cdot \cG$ is realizable precisely when $Z(F)$ is nonempty. The linkage $\cM$ from Construction~\ref{con:noncrossing-main-theorem} liftably and rigidly draws a translation of $Z(F)$ (because $Z(F)\subseteq[-1,1]^{2m}$), so $F$ has a common root if and only if $\Conf(\cM)$ is nonempty, i.e., if and only if $\cM$ is realizable. By Property~\ref{thmpart:pins}, $\cM$'s pins serve only to prevent rigid transformations and do not affect realizability, so $\cM$ is realizable precisely when $\cG$ is realizable. Finally, scaling by~$t$ does not affect realizability, proving the result.
\end{proof}

\begin{theorem}[Hardness of Noncrossing Rigidity and Global Rigidity]
  \label{thm:noncrossing-rigidity}
  \label{thm:noncrossing-global-rigidity}
  Deciding whether a given configured weighted graph $(\cG,C_0)$ is rigid, when $\cG$ is promised to be globally noncrossing (so in particular, $C_0$ is noncrossing) and $C_0$ has integer coordinates and constant-sized integer edge lengths, is $\forall\bR$-complete. 
  It remains $\forall\bR$-complete if ``rigid'' is replaced by ``globally rigid''.
\end{theorem}

\begin{proof}
  Schaefer~\cite{Schaefer-2013} has shown that the general linkage rigidity problem is $\forall\bR$ complete and therefore belongs to $\forall\bR$. Linkage global rigidity likewise belongs to $\forall\bR$: it may be expressed in the form \emph{for all valid configurations $C$, $|C(u)-C(v)|=|C_0(u)-C_0(v)|$ for all pairs of (not necessarily adjacent) vertices $u$ and $v$}.
  
  For hardness, we reduce from $\HtwoN$ as strengthened in Theorem~\ref{thm:htwon'}, so suppose we are given a family of homogeneous polynomials $F = \{f_1,\ldots,f_s\}$ of degree $d=4$ in variables $\{x_1,y_1,\ldots,x_m,y_m\}$ with constant-sized integer coefficients. We may use Construction~\ref{con:noncrossing-main-theorem} (Part III) to build a globally noncrossing configured linkage $(\cM,C_0)$ that liftably and rigidly draws a trace $\pi_X(\Conf(\cM))$ satisfying $T(Z(F)\cap[-1,1]^{2m}) \subseteq \pi_X(\Conf(\cM)) \subseteq T(Z(F))$ for some translation $T$.
  The result of this reduction will be the configured graph $\cG$ formed by scaling $(\cM,C_0)$ by $28800D$ and removing the three pins. Configured graph $\cG$ indeed has polynomially-bounded integer coordinates and constant integer lengths.
  
  To verify the validity of this reduction, suppose first that $Z(F)$ contains some nonzero point $\vec a$. Then $Z(F)$ contains the entire path $p\mapsto p\cdot\vec a$ starting at $\vec{0}$, so $\vec 0$ is not isolated in $Z(F)$, i.e., $T(\vec 0)$ is not isolated in $\pi_X(\Conf(\cM))$. Because $\cM$ draws liftably, $(\cM,C_0)$ is not rigid (as a linkage), and therefore $\cG$ is not rigid (as a graph). On the other hand, if $Z(F) = \{\vec 0\}$ then $\pi_X(\Conf(\cM))$ contains only the single point $\pi_X(C_0) = T(\vec 0)$, and by Property~\ref{thmpart:uniqueness} of Theorem~\ref{thm:main-theorem} (uniqueness), it follows that $\Conf(\cM) = \{C_0\}$, i.e., $\cM$ is both rigid and globally rigid as a linkage. By Property~\ref{thmpart:pins}, all configurations of $\cG$ are Euclidean transformations of configurations of $28800D\cdot\cM$, and so $\cG$ is rigid and globally rigid as a graph.
\end{proof}

\begin{theorem}[Universality of Globally Noncrossing Linkages]
  \label{thm:noncrossing-universality}
  The proper subsets $R\subsetneq\bR^2$ that may be drawn by a globally noncrossing linkage are precisely the compact semialgebraic sets.
\end{theorem}

\begin{proof}
  Given any unconstrained linkage $\cL$ and vertex $v$, the trace of $v$ is either all of $\bR^2$ (if $\cL$ has no pins) or is the projection of compact algebraic set $\Conf(\cL)$ and is therefore compact and semialgebraic. (Recall that $\cL$ is assumed to be connected.)

  Conversely, by Lemma~\ref{lem:compact-semialg-projection}, any compact semialgebraic region $R\subset\bR^2$ may be written as the projection of some basic compact set $R' = Z(f_1,\ldots,f_s)\subset \bR^{2m}$ onto the first two variables, so it suffices (by ignoring all drawing vertices except the first) to show that some translation of $R'$ may be drawn by a globally noncrossing linkage. By scaling as necessary, we may further assume that the compact set~$R'$ lies in the box $[-1,1]^{2m}$. But this now follows directly from Construction~\ref{con:noncrossing-main-theorem}, Part~I.
\end{proof}

\headingboldmath
\section{Unit-Distance and $\{1,2\}$-Distance Graphs and Linkages}
\label{sec:unit}

\begin{definition}[Unit Distance Graphs/Linkages]
  Define an \term{abstract unit-distance graph (or linkage)} as an abstract weighted graph (or abstract linkage) where all edges have weight~$1$; a \term{configured unit-distance graph (or linkage)} additionally comes with such a configuration.
\end{definition}

With our terminology, an \term{abstract unit distance graph/linkage} does not necessarily have any valid configurations, in contrast to the more common usage of the term ``unit-distance graph''. To mitigate confusion with our overloading of this term, we will always refer to a unit-distance graph/linkage as ``abstract'' or ``configured''.

In this section we prove the strong $\exists\bR$-completeness or $\forall\bR$-completeness of realizability, rigidity, and global rigidity for unit-distance (or in the case of global rigidity, $\{1,2\}$-distance) graphs that allow crossings. We also show universality: any compact semialgebraic set in $\bR^2$ can be drawn by a unit-distance linkage. (Unit-distance graphs that do \emph{not} allow crossings, i.e., \emph{matchstick graphs}, are the topic of Section~\ref{sec:matchstick}.)

There are two noteworthy obstacles in these arguments that were not present in the previous section. First, the universality proof involves a new complication, namely, non-algebraic numbers. To illustrate, the circle $C = \{(x,y)\mid x^2+y^2=e^2\}$ (where $e$ is Euler's constant) can be drawn easily by a linkage (using a single edge of length $e$), but simulating such an edge with a unit-distance graph is impossible because $e$ is transcendental. As a workaround, we instead rely on \emph{pins} to introduce non-algebraic values. Indeed, we may slightly generalize curve $C$ by introducing new variables $(a,b)$ and considering the modified curve
\begin{equation*}
  C' = \{((x,y),(a,b))\in\bR^4 \mid x^2+y^2=a^2\}.
\end{equation*}
As $C'$ is now defined by polynomials with \emph{integer} coefficients, the Main Theorem (Part II) applies, and the resulting linkage may be simulated by a unit distance linkage using techniques to be presented below. Finally, with one pin, we may fix the values $a=e$ and $b=0$, which recovers the desired circle $C$. Suitably generalized, this argument can be made to work for arbitrary compact semialgebraic sets; see Theorem~\ref{thm:unit-universality} for details.

For the second obstacle, we were not able to prove $\forall\bR$-completeness of detecting global rigidity of unit-distance graphs. Indeed, we are not aware of the existence of \emph{any} globally rigid unit-distance graphs larger than a triangle!

\begin{question}
  Are there any globally rigid unit-distance graphs with more than $3$ edges?
\end{question}
\noindent If such a graph is found, it is likely that the methods of this paper can turn it into a proof of hardness.

As a consolation prize, we demonstrate $\forall\bR$-completeness of global rigidity for graphs with edge lengths in $\{1,2\}$, an appropriate strengthening of Saxe's result~\cite{Saxe-1979} that global rigidity is coNP-hard for graphs with edge lengths in $\{1,2\}$.

\headingboldmath
\subsection{Simulation with Unit- and $\{1,2\}$-Distance Linkages}
\label{sec:unit-simulation}

Here we show how to simulate Theorem~\ref{thm:main-theorem} using $\{1,2\}$-distance linkages.

\begin{lemma}[Reinforced Segment]
  \label{lem:reinforced-segment}
  A single edge of integer length $n$ is perfectly simulated by a \term{reinforced bar} graph formed by adjoining $n-1$ degenerate $\{1,1,2\}$-sided triangles along unit edges.
\end{lemma}

\begin{proof}
  This is a simple extension of a tool used in~\cite[Cor.~4.3]{Saxe-1979}.
\end{proof}

\noindent To rigidify orthogonal trees with $\{1,2\}$-graphs, it suffices to rigidify entire lattice grids:

\begin{lemma}[Reinforced Grid]
  \label{lem:reinforced-grid}
  Let $(G,C)$ be the configured graph whose vertices lie at all integer points in $[0,n]\times[0,n]$ and whose unit-length edges connect vertically and horizontally adjacent vertices in this grid. Then the rigidified graph $\cG := ((G,C),\RigidCon_G(G,C))$, where configuration $C$ is rigidified in its entirety, can be perfectly simulated by an unconstrained $\{1,2\}$-distance graph, called a \term{reinforced grid}.
\end{lemma}

\begin{proof}
  We may assume $n\ge 4$; otherwise, we apply the construction for $n=4$ and restrict attention to the smaller subgrid.
  For each $0\le j,k\le n$, let $v_{j,k}$ be the vertex of $G$ with $C(v_{j,k}) = (j,k)$. By Lemma~\ref{lem:reinforced-segment}, we may add length-$2$ edges $v_{j,k} v_{j+2,k}$ and $v_{j,k} v_{j,k+2}$ to force each row and column of vertices in $G$ to remain straight. Now add one more rigidified bar of length~$5$ connecting $v_{4,0}$ and $v_{0,3}$, which constrains row $k=0$ and row $j=0$ to remain at $90^\circ$ from each other. In fact, this resulting graph $G'$ is the desired graph. Indeed, suppose we have a configuration of $G'$; by a Euclidean motion, we may assume $v_{0,0}$, $v_{n,0}$, and $v_{0,n}$ are configured at $(0,0)$, $(n,0)$, and $(0,n)$ respectively. Because $|v_{n,0}-v_{n,n}| = n$ and $|v_{0,n}-v_{n,n}| = n$ in any configuration, $v_{n,n}$ must rest at $(n,n)$ or $(0,0)$. In the latter case, $v_{n,1}$ rests at $(n-1,0)$, which is not distance~$n$ away from $v_{0,1}$ at $(0,1)$, contradicting Lemma~\ref{lem:reinforced-segment}. So $v_{n,n}$ must indeed lie at $(n,n)$, and the rest of the vertices' locations are then fixed.
\end{proof}

\begin{construction}
  \label{con:unit-main-theorem}
  Use notation as in Theorem~\ref{thm:main-theorem}: we are given a collection of polynomials $F = \{f_1,\ldots,f_s\}$ in $2m$ variables, each of total degree at most $d$, from which Theorem~\ref{thm:main-theorem} constructs a partially rigidified linkage $\cL = \cL(F)$. We make no claims under Part~I alone.

  If the Part~II assumption holds, meaning the $f_j$ have integer coefficients with absolute value at most $M$, we may construct, in $O(\poly(m^d,d^d,s,M))$ time, an abstract linkage $\cM = \cM(F)$ with edge lengths in $\{1,2\}$ that perfectly simulates the scaled linkage $D\cdot\cL(F)$.
  In particular, there is a translation $T$ on $\bR^{2m}$ and
  a subset $X$ of $m$ vertices such that
  $$T(D \cdot (Z(F)\cap[-1,1]^{2m})) \subseteq \pi_X(\Conf(\cM)) \subseteq T(D \cdot Z(F)).$$

  If the Part~III assumption also holds, i.e., each $f_j$ satisfies $f_j(\vec 0) = 0$ which gives rise to a configuration $C_0$ of $\cL$, then the configuration of $\cM$ corresponding to $C_0\in\Conf(\cL)$ has rational coordinates whose numerators have magnitude $O(\poly(m^d,d^d,s,M))$ and whose denominators are at most $D^2$.
\end{construction}

\begin{proof}
  By Property~\ref{thmpart:rational-coords}, all edge lengths of $D\cdot\cL$ are integers. By Properties~\ref{thmpart:num-vertices}, \ref{thmpart:orthogonal-trees} and~\ref{thmpart:rational-edge-lengths}, each rigidified subtree $(H,C_H)$ in $D\cdot \cL$ has integer coordinates whose sizes are bounded by $O(\poly(m^d,d^d,s,M))$.

  Each edge of $D\cdot \cL$ not belonging to any rigidified tree gets replaced with a reinforced segment of appropriate length as in Lemma~\ref{lem:reinforced-segment}. 
  For each rigidified tree $(H,C_H)$ of $D\cdot \cL$, build a reinforced grid of unit squares as in Lemma~\ref{lem:reinforced-grid} large enough to include the coordinates of the vertices of $(H,C_H)$, and replace the tree by this grid;  neighboring edges or trees are attached at the corresponding grid point. The three pins of $D\cdot\cL$ are likewise transferred to their corresponding grid points. Call the resulting linkage $\cM = \cM(F)$. Each reinforced segment of $\cM$ perfectly simulates its edge by Lemma~\ref{lem:reinforced-segment}. Likewise, because each rigidified tree in $D\cdot \cL$ has at least three noncollinear vertices (Property~\ref{thmpart:orthogonal-trees} of the Main Theorem) each reinforced grid perfectly simulates its rigidified tree by Lemma~\ref{lem:reinforced-grid}. So $\cM$ perfectly simulates $\cL$.

  If Part~III holds, then scaled configuration $D\cdot C_0$ of $D\cdot \cL$ has integer coordinates bounded by $O(\poly(m^d,d^d,s,M))$ in magnitude. By Property~\ref{thmpart:orthogonal-induced-trees} and Lemma~\ref{lem:reinforced-grid} (Reinforced Grid), all nodes of $\cM$ in reinforced grids are configured with coordinates in $\frac{1}{5}\cdot\bZ$. It remains to look at the coordinates of the remaining nodes: those on reinforced segments. Consider an edge $e$ of $D\cdot \cL$ \emph{not} contained in a rigidified tree. Edge $e$ has integer length $q \le D^2$, and its endpoints are initially configured at integer coordinates $(a_1,b_1)$ and $(a_2,b_2)$ (by Property~\ref{thmpart:rational-coords}). Then the reinforced bar corresponding to $e$ in $\cM$ has nodes configured at $(a_1,b_1) + \frac{h}{q}(a_2-a_1, b_2-b_1)$ for integers $0\le h\le q$, and these coordinates are rationals of the required form.
\end{proof}

We rely on the full strength of \emph{perfect} simulation in the proof of global rigidity below, but for the other hardness results, liftable and rigid simulation is sufficient. The latter may be achieved with only unit-length edges:

\begin{lemma}
  \label{lem:moser-spindle}
  A single edge of length $2$ can be liftably and rigidly simulated by a unit-distance graph with $19$ edges, formed by joining two copies of Moser's Spindle along a common equilateral triangle. If the edge has its endpoints configured at $(0,0)$ and $(2,0)$, then the configured unit-distance graph simulating this edge has coordinates of the form $a + b\sqrt{3} + c\sqrt{11} + d\sqrt{33}$ for rational numbers $a,b,c,d$.
\end{lemma}

\begin{proof}
  The first claim follows from \cite[Lemma~3.4]{Schaefer-2013}, while the second may be verified by direct computation.
\end{proof}

\headingboldmath
\subsection{ Hardness and Universality of Unit-Distance and $\{1,2\}$-Distance Graphs}

\begin{theorem}[Hardness of Unit-Distance Realization]
  \label{thm:unit-realizability}
  Deciding whether a given abstract unit-distance graph is realizable is  $\exists\bR$-complete.
\end{theorem}

\begin{note}
  This was shown by Schaefer~\cite{Schaefer-2013} with a simpler, specialized construction, but we include it here for completeness.
\end{note}

\begin{proof}
  Membership in $\exists\bR$ is evident. Hardness follows by reduction from \CommonZero{} exactly as in the proof of Theorem~\ref{thm:noncrossing-realizability}, using Construction~\ref{con:unit-main-theorem} in place of Construction~\ref{con:noncrossing-main-theorem}.
\end{proof}

\begin{theorem}[Hardness of Unit-Distance Rigidity]
  \label{thm:unit-rigidity}
  The problem of determining whether a configured unit-distance graph with coordinates in $\bQ[\sqrt{3},\sqrt{11}]$ is rigid is $\forall\bR$-complete, even when all coordinates have the form $(a+b\sqrt{3}+c\sqrt{11}+d\sqrt{33})/n$ for integers $a,b,c,d,n$ where $b$, $c$, $d$, and $n$ have size $O(1)$ (but $a$ need not have constant size).
\end{theorem}

\begin{proof}
  As in Theorem~\ref{thm:noncrossing-rigidity}, this problem lies in $\exists\bR$. Hardness follows by reduction from the complement of $\HtwoN$: given an instance $F = \{f_1,\ldots,f_s\}$ of this problem (which we may assume consists of polynomials of degree~$4$ with constant-sized coefficients by Theorem~\ref{thm:htwon'}), use Construction~\ref{con:unit-main-theorem} and Lemma~\ref{lem:moser-spindle} to build a configured unit-distance linkage~$\cM$ that liftably and rigidly draws a translation of $D\cdot(Z(F)\cap[-1,1]^{2m})$. Then this linkage is rigid if and only if $\vec 0$ is the only common zero of $F$, i.e., $F$ is a ``no'' instance of $\HtwoN$. As in the proof of Theorem~\ref{thm:noncrossing-rigidity}, removing the three pins of $\cM$ results in a unit-distance graph that is rigid if and only if $F$ is a ``no'' instance of \HtwoN{}.

  We must also show that the coordinates of~$\cM$ have the required form.  Construction~\ref{con:unit-main-theorem} guarantees that Lemma~\ref{lem:moser-spindle} is applied only to length-$2$ edges whose endpoints have rational coordinates with denominators bounded by $D^2$. For each such edge, say with endpoints at $p=(p_1,p_2)\in\bQ^2$ and $q=(q_1,q_2)\in\bQ^2$, the gadget drawn to connect $p$ and $q$ may be computed by starting with the gadget of Lemma~\ref{lem:moser-spindle} (which connects $(0,0)$ to $(2,0)$), applying the rotation matrix
  \begin{equation*}
    \frac{1}{2}\begin{pmatrix}
      q_1-p_1 & -(q_2-p_2) \\
      q_2-p_2 & q_1-p_1
    \end{pmatrix},
  \end{equation*}
  and then translating by $(p_1,p_2)$. The entries of the rotation matrix are rationals with denominators at most $2D^2$ and magnitudes at most $1$ (because $|q-p| = 2$), while $p_1$ and $p_2$ have denominators bounded by $D^2$, so the result follows.
\end{proof}

\begin{theorem}[Hardness of $\{1,2\}$-Distance Global Rigidity]
  \label{thm:unit-global-rigidity}
  The problem of deciding whether a given configured $\{1,2\}$-distance graph with coordinates in $\bQ$ is globally rigid is $\forall\bR$-complete, even when all coordinates have denominators of size $O(1)$.
\end{theorem}

\begin{proof}
  Membership in $\forall\bR$ follows as in Theorem~\ref{thm:noncrossing-global-rigidity}. Hardness follows by reduction from the complement of \HtwoN{} just as in the proof of Theorem~\ref{thm:noncrossing-rigidity}, using Construction~\ref{con:unit-main-theorem} instead of Construction~\ref{con:noncrossing-main-theorem}. As in the proof of Theorem~\ref{thm:noncrossing-rigidity}, this makes essential use of the fact that $\cM$ from Construction~\ref{con:unit-main-theorem} simulates the linkage~$\cL$ from the Main Theorem \emph{perfectly}, not just liftably and rigidly.
\end{proof}

As discussed at the start of Section~\ref{sec:unit}, Construction~\ref{con:unit-main-theorem} requires the input polynomials to have integer coefficients, but some compact semialgebraic sets cannot be expressed in this way. We now formalize the workaround described there to prove universality of unit-distance linkages.

\begin{theorem}[Universality of Unit-Distance Linkages]
  \label{thm:unit-universality}
  Any compact semialgebraic set $R\subset\bR^2$ may be drawn by a unit-distance linkage.
\end{theorem}

\begin{proof}
  As in the proof of Theorem~\ref{thm:noncrossing-universality} we may write $R$ as the projection onto coordinates $x_1,y_1$ of some compact basic algebraic set
  \begin{equation*}
    R' = Z(f_1,\ldots,f_s), f_j\in\bR[x_1,y_1,\ldots,x_m,y_m],
  \end{equation*}
  and it suffices to show that some translation of $R'$ may be drawn with a unit-distance linkage. In fact, it suffices to show that the scaled set $\frac{1}{n}\cdot R'$ may be drawn (up to translation) by a unit-distance linkage $\cL$, for some $n\in\bN$: indeed, if a unit-distance linkage $\cL$ draws a translation of $\frac{1}{n}\cdot R'$, then $n\cdot\cL$ draws a translation of $R'$ and has integer edge lengths, so by Lemmas~\ref{lem:reinforced-segment} and~\ref{lem:moser-spindle} it may be simulated by a unit-distance linkage, as required. By this reasoning, we may replace the compact set $R'$ by some small-enough $\frac{1}{n}\cdot R'$ and thereby assume that $R'$ lies in the box $[-1,1]^{2m}$.

  Write the scaled set $\frac{1}{D}\cdot R'$ as $Z(g_1,\ldots,g_s)$, where $g_j(\widevec{xy}) = f_j(D\cdot\widevec{xy})$. We wish to apply Construction~\ref{con:unit-main-theorem} to polynomials $g_1,\ldots,g_s$ so the resulting linkage draws precisely $R'$ (up to translation), but these coefficients may not be integers (or even algebraic numbers, as described above), so we will temporarily replace these coefficients with variables, as follows.
  By scaling the coefficients of each $g_j$ (which does not affect $Z(g_1,\dots,g_s)$), we may assume that the coefficients are in $[-1,1]$. For each nonzero monomial $c_{j,J}\widevec{xy}^J$ in each $g_j$ (where $J$ is a vector of exponents), create new variables $a_{j,J}$ and $b_{j,J}$, gather all of these new variables into a vector $\widevec{ab}$ with length $2r$, and define the new polynomials
  \begin{equation*}
    h_j(\widevec{xy},\widevec{ab}) := \sum_{J\text{ such that }c_{j,J}\ne 0} a_{j,J}\widevec{xy}^J,
  \end{equation*}
  for $1\le j\le s$. Polynomials $h_j$ now have integer coefficients: in fact, all coefficients equal~$1$.
  It remains to implement the equations
  $h_j(\widevec{xy},\widevec{ab})=0$ and $(a_{j,J},b_{j,J})=(c_{j,J},0)$,
  for each $1\le j\le s$ and $c_{j,J}\ne 0$,
  which exactly recover the solutions $\widevec{xy} \in Z(g_1,\dots,g_s)$.

  Applying Construction~\ref{con:unit-main-theorem} and Lemma~\ref{lem:moser-spindle} to the polynomials $h_j$, we may construct some unit-distance linkage $\cL$, a set of drawing vertices
  \begin{equation*}
    X = \{v_j\mid 1\le j\le s\}
    \cup\{v_{j,J} \mid 1\le j\le s\text{ and }c_{j,J}\ne 0\}
  \end{equation*}
  corresponding to variables $(x_j,y_j)$ and $(a_{j,J},b_{j,J})$ respectively, and some translation $T$ on $\bR^{2m+2r}$ such that $X$ draws a set between $T(D\cdot (Z(h_1,\ldots,h_s)\cap[-1,1]^{2m+2r}))$ and $T(D\cdot Z(h_1,\ldots,h_s))$. Finally, we pin all the vertices $v_{j,J}$ in the plane to force variables $(a_{j,J},b_{j,J})$ to take the values $(c_{j,J},0)$: specifically, if $T_{j,J}$ denotes the translation $T$ restricted to coordinates $(a_{j,J},b_{j,J})$, we pin $v_{j,J}$ to the point $T_{j,J}(c_{j,J},0)$. The trace of vertices $v_1,\ldots,v_s$ in this pinned linkage is a translation of $D\cdot Z(g_1,\ldots,g_s) = R'$, as required.  
\end{proof}

\section{Matchstick Graphs and Linkages}
\label{sec:matchstick}

In Section~\ref{sec:noncrossing} we discussed globally noncrossing linkages, which have no crossing configurations because they are required to be carefully designed to enforce this stringent property. By contrast, in this section we look at \emph{NX-constrained} linkages (short for non-crossing-constrained), which have no crossing configurations for a very different reason: we simply declare that crossing configurations should be ignored. \emph{Any} linkage can be made into an NX-constrained linkage by simply attaching a constraint (cf. Definition~\ref{def:constraint}), specifically an \emph{NX-constraint}, that redefines the set of valid configurations as follows:

\begin{definition}[NX-Constrained Linkages and Matchstick Linkages]
  If $\cL$ is a linkage, we define the \term{NX-constraint} on $\cL$, denoted $\NXCon_\cL$, as the set of configurations of $\cL$ that do not cross. In other words, if $\cL' := (\cL,\NXCon_\cL)$ is an \term{NX-constrained linkage}, then its configuration space is defined by $\Conf(\cL') := \NXConf(\cL)$.

  If $\cL$ is an unconstrained linkage all of whose edges have length~$1$, then the NX-constrained linkage $\cL' := (\cL,\NXCon_\cL)$ is called a \term{matchstick linkage}.
\end{definition}

In this section we prove analogous hardness results about matchstick linkages and graphs: matchstick graph realization is $\exists\bR$-complete, and matchstick graph rigidity and global rigidity are $\forall\bR$-complete. An analogous ``universality'' result---that matchstick linkages can draw all compact semialgebraic sets---can also be proved with arguments similar to the prior universality results. However, such a theorem would be incomplete, because matchstick linkages (more generally, NX-constrained linkages) can draw more than just compact semialgebraic sets! As a simple example, the NX-constrained two-bar linkage $\cA$ of Figure~\ref{fig:nx-conf-open} draws the half-open annulus $\{(x,y)\in\bR^2\mid 1 < x^2+y^2 \le 9\}$, because configurations of (the unconstrained linkage underlying) $\cA$ with $v$ on circle $x^2+y^2=1$ are crossing and are thus considered invalid by the constraint. In general, the traces of NX-constrained linkages are \emph{bounded and semialgebraic} sets (assuming at least one pin). NX-constrained linkages---in fact, matchstick linkages---can indeed draw all such sets, but our proof of this stronger result (Theorem~\ref{thm:matchstick-universality-full}) subtly breaks the abstraction barrier set up by the Main Theorem, so we postpone this proof until Section~\ref{sec:matchstick-universality-full}.

\begin{figure}[h]
  \centering
  \includegraphics{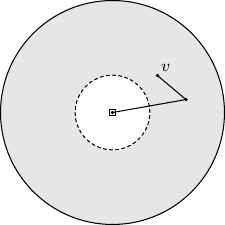}
  \caption{The trace of an NX-constrained linkage need not be closed.
    \label{fig:nx-conf-open}
  }
\end{figure}

\subsection{Simulation with Matchstick Linkages}

For an integer $a \ge 5$, we will simulate length-$a$ edges by \term{edge polyiamonds} as shown in Figure~\ref{fig:edge-polyiamonds}. Note that any polyiamond, considered as a matchstick graph, is globally rigid. As proven in the following lemma, two simple but effective \term{wing edges} suffice to fix the relative orientation and position of two adjacent edge polyiamonds (Figure~\ref{fig:edge-polyiamonds}). We may more easily describe the \emph{relative} positions of edge polyiamonds $P$ and $Q$ by temporarily pinning $P$ in place:

\begin{lemma}[Wing Edges]
  \label{lem:matchstick-wing-edges}
  Consider the matchstick linkage $\cL$ drawn in Figure~\ref{fig:edge-polyiamonds}, with two edge polyiamonds $P$ and $Q$ sharing a vertex, three pins in $P$, and two extra \term{wing edges} attached at vertices $a$ and $b$ as shown. Then in every configuration of this linkage, $Q$ has the same orientation as $P$, and its central axis is rotated from $P$'s axis by an angle $60^\circ < \theta < 240^\circ$. Every such $\theta$ corresponds to a unique configuration of $\cL$.
\end{lemma}

\begin{proof}
  Polyiamond $Q$ must be drawn with the same orientation as $P$, or else the wing edges will create crossings. Likewise, $a$, $o$, $b$, and $c$ must form a rhombus to prevent these four edges from intersecting each other. So $Q$'s and $c$'s location are determined by the single angle~$\theta$, and it may be seen that no crossings occur precisely when $60^\circ < \theta < 240^\circ$.
\end{proof}

\begin{figure}[hbtp]
  \centering
  \includegraphics{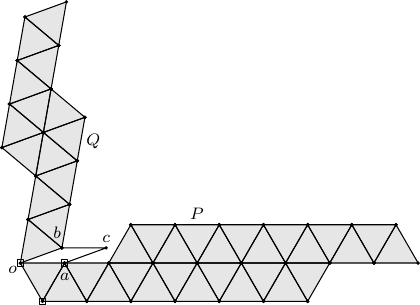}
  \hfill
  \includegraphics{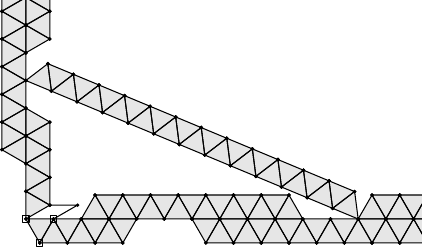}
  \caption{Left: Edge polyiamonds used to simulate edges of integer length. Right: Edge polyiamonds braced at $90^\circ$ based on a $5$-$12$-$13$ right triangle.}
  \label{fig:edge-polyiamonds}
\end{figure}

\begin{lemma}[Orthogonal Braces]
  \label{lem:matchstick-right-angles}
  By modifying the edge polyiamonds and adding a \term{hypotenuse polyiamond} as shown in Figure~\ref{fig:edge-polyiamonds}, the resulting matchstick linkage~$\cM$ is globally rigid, and its unique configuration has $\theta = 90^\circ$.
\end{lemma}

\begin{proof}
  The previous lemma and the $5$-$12$-$13$ right triangle force $Q$ to be drawn at a $+90^\circ$~angle from~$P$ and with the same orientation. The hypotenuse polyiamond then has only one crossing-free position.
\end{proof}

\begin{construction}
  \label{con:matchstick-main-theorem}
  Use notation as in Theorem~\ref{thm:main-theorem}: we are given a collection of polynomials $F = \{f_1,\ldots,f_s\}$ in $2m$ variables, each of total degree at most $d$, from which Theorem~\ref{thm:main-theorem} constructs a partially rigidified linkage $\cL = \cL(F)$.
  
  If the hypotheses of Part~II hold, so polynomials $f_j$ have integer coordinates with absolute values at most $M$, we may construct, in $O(\poly(m^d,d^d,s,M))$ time, an abstract matchstick linkage $\cM$ that perfectly simulates the scaled linkage $D\cdot\cL$.

  If the Part~III assumption also holds, i.e., each $f_j$ satisfies $f_j(\vec 0) = 0$ giving rise to configuration $C_0$ of $\cL$, then the corresponding configuration of $\cM$ has coordinates of the form $(a+b\sqrt{3})/c$, where $a$ is an integer of magnitude $O(\poly(m^d,d^d,s,M))$, and $b$ and $c$ are integers with magnitude at most $O(1)$.
\end{construction}

\begin{proof}
  We may assume that in each rigidified subtree $(H,C_H)$ of $\cL$, no internal node has degree exactly~$2$ with edges at~$180^\circ$ from each other in $C_H$, because such a node can be erased by merging its two edges together. For the remainder of this proof, replace $\cL$ with $40D\cdot\cL$: each rigidified tree has integer coordinates, each edge has integer length at least $40$, and each configuration has feature size at least $40$, by Theorem~\ref{thm:main-theorem}.

  We create a matchstick linkage~$\cM$ simulating $\cL$ using Lemmas~\ref{lem:matchstick-wing-edges} and~\ref{lem:matchstick-right-angles}, as follows. First, $\cM$ has a vertex corresponding to each vertex of~$\cL$. Each edge of~$\cL$ is replaced by an edge polyiamond of the appropriate length connecting the corresponding vertices.

  At any vertex $v$ internal to some (necessarily unique, by Property~\ref{thmpart:orthogonal-trees}) rigidified tree $(H,C_H)$, brace any right angles at~$v$ in $(H,C_H)$ as in Lemma~\ref{lem:matchstick-right-angles}, by modifying the corresponding edge polyiamonds and adding wings and a hypotenuse. Any $180^\circ$ or $270^\circ$ angles at $v$ in $(H,C_H)$ can be left alone.
  At any vertex $v$ of $\cL$ that is \emph{not} internal to any rigidified tree (and has degree at least $2$), simply attach wings to the incident edge polyiamonds as in Lemma~\ref{lem:matchstick-wing-edges} according to the cyclic order~$\sigma_v$ at~$v$.
  Finally, $\cL$ has exactly three pins in one of the rigidified trees $(H,C_H)$; in $\cM$, we instead place three non-collinear pins in the edge polyiamond corresponding to \emph{one} of $H$'s edges.

  We must show that $\cM$ perfectly simulates $\cL$. First, we claim no edge polyiamond may reverse its orientation: they all have the same orientation as in Figure~\ref{fig:edge-polyiamonds} in every configuration of $\cM$. This is certainly true for the pinned edge polyiamond. Lemma~\ref{lem:matchstick-wing-edges} shows that edge polyiamonds connected by wing edges must maintain their relative orientations. At each vertex $v$ of $\cL$, \emph{all} of the edge polyiamonds of $\cM$ incident to $v$ are connected to each other by wing edges: if $v$ is internal to a rigidified tree then at most one of its corners is missing wing edges (by the assumption above), and otherwise all of $v$'s corners have wing edges. So by connectivity, all edge polyiamonds maintain the same orientation, as claimed. This shows that each configuration $C$ of $\cL$ induces \emph{at most} one configuration of $\cM$: the edge polyiamond built for edge $u v$ must be rotated and translated (not reflected) to connect points $C(u)$ and $C(v)$ in the plane.

  Similarly, each vertex interior to a rigidified tree $(H,C_H)$ has hypotenuse polyiamonds on all or all but one of its corners, which enforce the constraint $\RigidCon(H,C_H)$ by Lemma~\ref{lem:matchstick-right-angles}. So each configuration of $\cM$ comes from a configuration of $\cL$, i.e., $\cM$ draws a \emph{subset} of $\Conf(\cL)$.

  It remains to show that each configuration $C$ of~$\cL$ induces a valid (i.e., noncrossing) configuration of $\cM$. For each rigidified tree $(H,C_H)$, the edge and hypotenuse polyiamonds simulating $(H,C_H)$ in $\cM$ may be configured to match $C$ without intersecting each other because the induced configuration $C|_H$ has the same orientation as $C_H$, by Property~\ref{thmpart:orthogonal-trees} of the Main Theorem. For each vertex $v$ of $\cL$ not interior to a rigidified tree, the angles of $v$'s corners in configuration $C$ lie strictly between $60^\circ$ and $240^\circ$ by Property~\ref{thmpart:corners} of the Main Theorem, so locally, wing edges around $v$ do not create crossings by Lemma~\ref{lem:matchstick-wing-edges}. Globally, configuration $C$ of~$\cL$ has feature size at least $40$, while the linkage~$\cM$ extends less than $20$~units away from the vertices and edges it simulates, so $C$ indeed induces a (noncrossing) configuration of $\cM$, as desired.

  Finally, if Part~III applies, let~$C_0'$ be the configuration of $\cM$ induced by initial configuration $C_0$ (after scaling $C_0$ by $40D$ as above), and note that (the scaled) $C_0$ has integer coordinates and integer edge lengths, both of polynomial magnitude. We must show that the coordinates of $C_0'$ have the required form $(a+b \sqrt{3})/c$ for integers $a,b,c$ bounded as claimed. For any edge $e$ of $C_0$ with integer endpoints $(a_1,b_1)$ and $(a_2,b_2)$ and integer length $q$, the vertices along the central axis of the corresponding edge polyiamond have rational coordinates of the form $p = (a_1,b_1) + \frac{h}{q}(a_2-a_1,b_2-b_1)$ for integers $0\le h\le q$. If $e$ is part of a rigidified tree then this edge is axis-aligned (by Property~\ref{thmpart:orthogonal-induced-trees}) and so $p$ has \emph{integer} coordinates; otherwise, $e$ has length $q \le 40D^2$ by Property~\ref{thmpart:rational-edge-lengths}, so the denominators of $p$ have magnitude $q = O(1)$. The rest of the vertices in the polyiamond are offset from these coordinates by the unit vector $\vec w = (a_2-a_1,b_2-b_1)/q$ rotated by some multiple of $60^\circ$; for the same reasons, denominators in $w$'s entries have constant size, so the same holds for the vertices in the edge polyiamond. The same computations holds for hypotenuse polyiamonds. The only vertices not yet accounted for are the wing vertices, and these have the form $A+B-O$ where $A$, $B$, and $O$ are rational-coordinate points along edge-polyiamond axes as described above.   
\end{proof}

\begin{note}
  A single edge on its own is \emph{not} perfectly simulated by an edge polyiamond, because the latter may reflect across its central axis. In the proof above, we were careful to argue that all edge polyiamonds maintain a fixed orientation (due to pins and wing edges) so that this ambiguity is impossible.
\end{note}

\subsection{Hardness of Matchstick Linkages}

\begin{theorem}[Hardness of Matchstick Graph Realizability]
  \label{thm:matchstick-realizability}
  Deciding whether a given abstract matchstick graph is realizable is $\forall\bR$-complete.
\end{theorem}

\begin{theorem}[Hardness of Matchstick Graph Rigidity]
  \label{thm:matchstick-rigidity}
  Deciding whether a given configured matchstick graph with coordinates in $\bQ[\sqrt{3}]$ is rigid is $\forall\bR$-complete, even when all coordinates have the form $(a+b\sqrt{3})/c$ for integers $a,b,c$ where $b$ and $c$ have size $O(1)$.
\end{theorem}

\begin{theorem}[Hardness of Matchstick Graph Global Rigidity]
  \label{thm:matchstick-global-rigidity}
  Deciding whether a given configured matchstick graph with coordinates in $\bQ[\sqrt{3}]$ is globally rigid is $\forall\bR$-complete, even when all coordinates have the form $(a+b\sqrt{3})/c$ for integers $a,b,c$ where $b$ and $c$ have size $O(1)$.
\end{theorem}

These hardness proofs are perfectly analogous to those in Section~\ref{sec:unit}, using Construction~\ref{con:matchstick-main-theorem} instead of Construction~\ref{con:unit-main-theorem} and Lemma~\ref{lem:moser-spindle}, so we omit their proofs.

As described at the beginning of this section, \emph{universality} of drawing with matchstick linkages is more subtle than prior universality results, because matchstick linkages can draw traces that are bounded semialgebraic sets which need not be closed.  We will prove that matchstick linkages can indeed draw \emph{all} such sets, but our argument utilizes details from the \emph{proof} of the Main Theorem. The universality statement and proof may therefore be found at Theorem~\ref{thm:matchstick-universality-full}, at the end of Section~\ref{sec:main-construction}.

\section{Extended Linkages and the Main Construction}
\label{sec:main-construction}

In this section we finally prove the Main Theorem. We first define \emph{extended linkages}, a special type of constrained linkage designed specifically for this proof, in Section~\ref{sec:defining-extended-linkages}. Section~\ref{sec:detailed-overview} then presents a detailed outline of our proof strategy, and Section~\ref{sec:parameters} details our choices of parameters. The remainder of this section then covers the proof in full.

\subsection{Defining Extended Linkages}
\label{sec:defining-extended-linkages}

We define and use \term{extended linkages}, which are constrained linkages whose constraints are tailored for the specifics of our construction.
(More general constraints are possible but not explored in this paper.)

The first of these constraints, the angle constraint, specifies a preferred arrangement of edges around each vertex, and specifies that the angle at each corner may not change very much. Recall that a \term{combinatorial embedding} consists of a cyclic counterclockwise ordering $\sigma_v$ of the edges incident to each vertex $v$ (Definition~\ref{def:combinatorial-embedding}).

\begin{definition}[Angle Constraint]
  For a linkage $\cL$ and a combinatorial embedding $\sigma$ of $\cL$, an \term{angle constraint}, $\AngleCon_\cL(\sigma,A,\Delta)$, is specified by an assignment of an angle $0 \le A(\Lambda) \le 2\pi$ and an angle tolerance $\Delta(\Lambda)\ge 0$ to each corner $\Lambda$ of $\sigma$, with the condition that $A$ assigns a total of~$2\pi$ to the corners around each vertex.

  A configuration $C\in\Conf(\cL)$ satisfies the angle constraint $\AngleCon_\cL(A,\Delta)$ if, for each corner $\Lambda$, angle $\angle C(\Lambda)$ lies in the closed interval
  \begin{equation*}
    [A(\Lambda)-\Delta(\Lambda),A(\Lambda)+\Delta(\Lambda)].
  \end{equation*}
  In particular, any corner with $\Delta(\Lambda) = 0$ is rigid: its angle in $C$ must be exactly $A(\Lambda)$. We call such corners \term{frozen}.
\end{definition}

Note that edges incident to a vertex $v$ might not be configured in the cyclic order that $\sigma$ prefers, if tolerances $\Delta(\Lambda)$ are large enough to allow otherwise. But even if each vertex locally agrees with $\sigma$'s cyclic ordering in some configuration $C$, we need $C$ to be noncrossing before we can say that it \emph{agrees with} $\sigma$ as in Definition~\ref{def:combinatorial-embedding}. In an angle constraint, $\sigma$ has only one role: to identify which pairs of edges should be considered corners.

The other constraint allows for sliceforms:

\begin{definition}[Sliceform Constraint]
  For a constrained abstract linkage~$\cL$ and a combinatorial embedding $\sigma$, a \term{Sliceform Constraint}, $\SliceCon_\cL(\sigma,S)$, is specified by a subset $S\subset V(G)$ of (some or all of the) vertices of degree~$4$. A configuration $C\in\Conf(\cL)$ satisfies the sliceform constraint $\SliceCon_\cL(\sigma,S)$ if, for each \term{sliceform vertex} $v\in S$ with neighboring vertices $w,x,y,z$ in cyclic order according to $\sigma$, points $C(w),C(v),C(y)$ are collinear in this order, and $C(x),C(v),C(z)$ are collinear in this order.
\end{definition}

Sliceforms permit a limited form of ``nonplanar'' interaction while still being simulatable without crossings (cf.\ Lemma~\ref{lem:sliceform-gadget} on page~\pageref{lem:sliceform-gadget}), so they are our primary tool in circumventing the difficulties of planarity.

An extended linkage is simply a linkage with each type of constraint listed above, with a few convenient restrictions:

\begin{definition}[Extended Linkage]
  An \term{$(\eps,\delta)$-extended linkage} where $0 < \delta < \eps < \pi/4$ is defined as a constrained linkage $\cL$ whose constraints $K$ have the form  $$K = \{\AngleCon_\cL(\sigma,A,\Delta),\SliceCon_\cL(\sigma,S)\}$$ with the same combinatorial embedding $\sigma$, where at each corner $\Lambda$ of $\cL$,
  $A(\Lambda)$ is chosen from $\{90^\circ,180^\circ,270^\circ,360^\circ\}$, and
  $\Delta(\Lambda)$ lies in $\{0,\delta,\eps\}$.
  We will call~$\cL$ simply an~\term{extended linkage} when~$\eps$ and~$\delta$ are the global constants given later by Equations~\ref{eq:eps}--\ref{eq:delta}.
\end{definition}

\begin{convention}[Drawing Extended Linkages]
  \label{con:drawing-extended-linkages}
  In drawings of extended linkages, all corners $\Lambda$ are drawn at their ``base'' angle, $A(\Lambda)$. Corners marked with a solid gray sector are frozen ($\Delta(\Lambda) = 0$), and the rest have $\Delta(\Lambda) = \eps$ unless otherwise specified. Vertices surrounded by small squares are pinned, and those marked with an ``x'' are sliceform vertices.
  See Figure~\ref{fig:scale-invariant-gadgets} for examples.
\end{convention}

It will often be useful to describe extended linkages not by the positions of their vertices, as with the projection $\pi_X$ (cf.\ Definition~\ref{def:trace} of \emph{trace}), but by the angles of a chosen set of corners. We therefore define a function $\Offset$ that measures how these angles differ from their ``neutral'' values given by $A$ in the angle constraint $\AngleCon(\sigma,A,\Delta)$:
\begin{definition}[The $\Offset$ Function]
  If $\cL$ is an extended linkage and $Y = (\Lambda_1,\ldots,\Lambda_k)$ is a tuple of (some or all of) the corners of $\cL$, define the function $\Offset_Y:\Conf(\cL)\to \bR^{|X|}$ by 
  \begin{equation*}
    \Offset_Y(C) = (\angle C(\Lambda_1) - A(\Lambda_1), \ldots, \angle C(\Lambda_k) - A(\Lambda_k)).
  \end{equation*}
\end{definition}

We conclude this section with some useful facts about extended linkages.

\begin{lemma}
  If $\cL$ is any $(\eps,\delta)$-extended linkage, then each vertex of $\cL$ has degree at most $4$. If a configuration $C\in\Conf(\cL)$ has no crossings, then $C$ agrees with combinatorial embedding~$\sigma$.
\end{lemma}

\begin{proof}
  The $A(\Lambda)$ values around each vertex must sum to $2\pi$, and each $A(\Lambda)$ is at least $\pi/2$, so each vertex has degree at most $4$, verifying the first claim.

  For the second claim, pick a vertex $v$, which we may assume has degree at least $2$, for otherwise there is nothing to prove at $v$. Let $\Lambda_1,\ldots,\Lambda_{\deg(v)}$ be the corners around $v$, and let $C$ be \emph{any} configuration of $\cL$. Because each corner begins where the previous corner ends, the angles $\angle C(\Lambda_j)$ around $C(v)$ will always add to an integer multiple of $2\pi$; the edges surrounding $v$ are configured in the cyclic order specified by $\sigma$ precisely when each $\angle C(\Lambda_j)$ is positive and they all add to exactly $2\pi$.

  For each $\Lambda_j$ ($1\le j\le \deg(v)$), we have $0 < \angle C(\Lambda_j) < A(\Lambda_j) + \pi/4$ by our assumptions on $A$ and $\Delta$. Thus,
  \begin{equation*}
    0 < \sum_{j=1}^{\deg(v)} \angle C(\Lambda_j) <
    \deg(v)\cdot\frac{\pi}{4} + \sum_{j=1}^{\deg(v)} A(\Lambda_j) =
    \deg(v)\cdot\frac{\pi}{4} + 2\pi
    < 4\pi,
  \end{equation*}
  so the angles $\angle C(\Lambda_j)$ must indeed add to $2\pi$, as claimed. Thus, if $C$ has no crossings, it agrees with~$\sigma$.
\end{proof}

\begin{lemma}
  \label{lem:extended-linkage-conf}
  For any $(\eps,\delta)$-extended linkage $\cL$, the configuration space $\Conf(\cL)$ is closed and semialgebraic. If $\cL$ is connected and has at least one pin, then $\Conf(\cL)$ is compact.
\end{lemma}
\begin{proof}
  Let $\cL'$ be the underlying, unconstrained linkage. Then $\Conf(\cL')$ is closed and semialgebraic (in fact algebraic), so it is enough to show that each of the constraints specifies a closed, algebraic subset of $\Conf(\cL')$.

  For the angle constraint, suppose $\Lambda = u v w$ is a corner with $A(\Lambda) = \alpha$ and $\Delta(\Lambda) = \theta$. If $u = w$ then $v$ has degree~$1$ and there is no constraint, so we may assume $u$ and $w$ are distinct vertices. We claim the angle constraint at $\Lambda$ can be expressed by the dot product inequality $p'\cdot q \ge \cos\theta$, where $p = \frac{u-v}{\ell(u v)}$, $p' = \Rot_\alpha p$ is the rotation of $p$ by angle $\alpha$ around the origin, and $q = \frac{w-v}{\ell(w v)}$.
  To see this, first observe that in any configuration of $\cL'$, $p$ and $q$ (and hence also $p'$) will be unit vectors. The dot product then computes the cosine of the angle between $p'$ and $q$. The inequality bounds this angle to the interval $[-\theta,\theta]$, as desired. So the angle constraint is indeed a closed, semialgebraic subset of $\Conf(\cL')$.

  The sliceform constraint is simpler. If $v$ is a sliceform vertex with neighboring vertices $w$, $x$, $y$, $z$ in cyclic order, then the constraint may be expressed by the vector equalities
  \begin{equation*}
    \frac{w-v}{\ell(v w)} = -\frac{y-v}{\ell(v y)}
    \qquad\text{and}\qquad
    \frac{x-v}{\ell(v x)} = -\frac{z-v}{\ell(v z)},
  \end{equation*}
  which is an algebraic (and therefore closed semialgebraic) constraint on $\Conf(\cL')$. The configuration space $\Conf(\cL)$ is therefore closed and semialgebraic, as claimed.

  If $\cL$ is connected and has at least one pin, then $\Conf(\cL)$ is closed and \emph{bounded} and is therefore compact.
\end{proof}

\subsection{Detailed Overview of Strategy}
\label{sec:detailed-overview}

Our Main Construction is primarily concerned with showing, for a finite set~$F$ of polynomials in $\bR[x_1,y_1,\ldots,x_m,y_m]$, how to construct an extended linkage that draws a bounded portion of the common zero set $Z(F)$, i.e., a set between $Z(F)\cap[-1,1]^{2m}$ and $Z(F)$, up to a translation. This subsection outlines the key points of our approach, with full details to follow in subsequent subsections.

\begin{convention}
  For convenience, points in $\bR^2$ will be specified by cartesian coordinates or by complex numbers interchangeably, so $(3\cos\theta,3\sin\theta)$ and $3\exp(i\theta)$ are identical.
\end{convention}

Our construction will be able to more easily manipulate \emph{angles} than \emph{vectors}, so wherever possible we will encode a point using two angles $\alpha$ and $\beta$, which relate to the point's Cartesian coordinates as follows:
\begin{align}\label{eqn:angular-coordinates}
  \Rect(\alpha,\beta) := 
  (\cos\alpha,\sin\alpha) + (-\sin\beta,\cos\beta) - (1,1)
  = e^{i\alpha}+ i\cdot e^{i\beta} - (1+i).
\end{align}
Usually $\alpha$ and $\beta$ will be small angles, so geometrically, $\Rect(\alpha,\beta)$ is the result of starting at $(-1,-1)$, stepping one unit at angle $\alpha$ (nearly horizontal), and then one unit at angle $\beta+\pi/2$ (nearly vertical). In particular, $\Rect(0,0) = (0,0)$. We refer to these angles $\alpha$ and $\beta$ as \term{angular coordinates} for the corresponding point. Similar strategies are employed in~\cite{Kempe-1876,Abbott-2008}.

\begin{lemma}
  \label{lem:rect-contains-box}
  For any angle $0 < \theta < \pi/4$, the function $\Rect$ is a homeomorphism from the region $[-\theta,\theta]^2$ onto its image, which is a compact region (specifically, the Minkowski sum of two circular arcs) containing the box $[-\theta/2,\theta/2]^2$. %
\end{lemma}

\begin{proof}
  The region $M := \Rect([-\theta,\theta]^2)$ is the Minkowski sum of two orthogonal circular arcs, each with central angle $2\theta$, as shown in Figure~\ref{fig:rect-facts}.

  $\Rect$ is a local homeomorphism on $[-\theta,\theta]^2$, because the Jacobian determinant at $(\alpha,\beta)$ is $\cos(\alpha-\beta)\ne 0$. It may also be seen geometrically that $\Rect$ is injective on the boundary of $[-\theta,\theta]^2$, since the four circular arcs bounding~$M$ in Figure~\ref{fig:rect-facts} are disjoint away from their endpoints. This is enough to ensure that $\Rect$ is bijective and in fact homeomorphic with its image, by~\cite[Corollary 2.6]{massey1992sufficient}.

  Region $M$ contains the square $[-\sin\theta,\sin\theta+\cos\theta-1]^2$, as shown in Figure~\ref{fig:rect-facts}. If we could show that $p(\theta) := \sin\theta - \theta/2$ and $q(\theta) := (\sin\theta+\cos\theta-1) - \theta/2$ are both nonnegative, we would conclude that this square, and hence $M$ itself, contains $[-\theta/2,\theta/2]^2$. Both $p$ and $q$ are concave down functions of $\theta\in[0,\pi/4]$, so it is enough to verify nonnegativity only at the endpoints $0$ and $\pi/4$, which is straightforward: $p(0)=q(0)=0$, $p(\pi/4) = (4\sqrt{2}-\pi)/8 > 0.314 > 0$, and $q(\pi/4) = \sqrt{2}-1-\pi/8 > 0.0215 > 0$.
\end{proof}

\begin{figure}[h]
  \centering
  \includegraphics{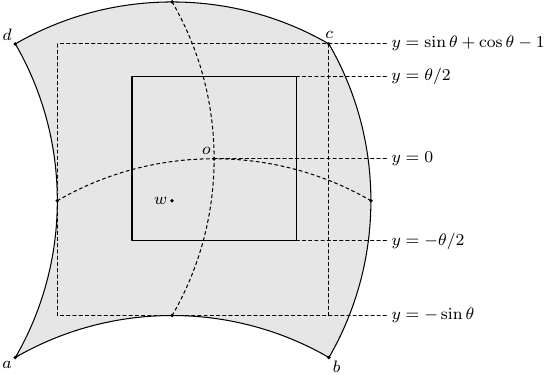}
  \caption{The region $M$ described in Lemma~\ref{lem:rect-contains-box} (gray) is bounded by four circular arcs and contains the square $[-\theta/2,\theta/2]^2$.}
  \label{fig:rect-facts}
\end{figure}

For the constructions below, in place of rectangular coordinates $(x_j,y_j)$ we will use angles $(\alpha_j,\beta_j)$ related by $(x_j,y_j) = 2r\cdot\Rect(\alpha_j,\beta_j)$, where the radius~$2r$ will be carefully chosen later. We may write this equivalently as
\begin{align}
  x_j &= r\left(e^{i\alpha_j} + e^{-i\alpha_j} + i e^{i\beta_j} - i e^{-i\beta_j} - 2\right),\nonumber \\
  y_j &= r\left(-i e^{i\alpha_j} + i e^{-i\alpha_j} + e^{i\beta_j} + e^{-i\beta_j} - 2\right).
        \label{eqn:change-of-variables-explicit}
\end{align}
Applying (\ref{eqn:change-of-variables-explicit}) for each $1 \leq j \leq m$,
write $\widevec{xy} := (x_1,y_1,\ldots,x_m,y_m)$ as a function
$\widevec{xy}(\widevec{\alpha\beta})$
where $\widevec{\alpha\beta} := (\alpha_1,\beta_1,\ldots,\allowbreak\alpha_m,\allowbreak\beta_m)$.
Make this substitution into each polynomial $f\in F$, expand fully, and combine like terms, resulting in an expression of the form
\begin{equation*}
  f(\widevec{xy}(\widevec{\alpha\beta})) = \sum_{u=0}^3 \sum_{I\in\Coeffs(2m,d)} i^u\cdot d_{u,I}\cdot e^{i\cdot(I\cdot\widevec{\alpha\beta})},
\end{equation*}
where each $d_{u,I}$ is a nonnegative real number, and
\begin{equation}
  \Coeffs(2m,d) := \{(a_1,b_1,\ldots,a_m,b_m)\in\bZ^{2m}\mid |a_1| + |b_1| + \cdots + |a_{m}| + |b_m| \le d\}.
  \label{eqn:polynomials-in-angular-intermediate}
\end{equation}
(The factors of $r$ from~\eqref{eqn:change-of-variables-explicit} have been incorporated into coefficients $d_{u,I}$.)
Only $d_{u,I}-d_{u+2,I}$ (with indices taken modulo $4$) affects the total, so by further cancellation we may assume one or both of these coefficients is $0$ (for each $u,I$ pair).
It will prove useful to change the form of this expression a bit: replace each term $e^{i\cdot(I\cdot\widevec{\alpha\beta})}$ from~\eqref{eqn:polynomials-in-angular-intermediate} with $e^{i\cdot(I\cdot\widevec{\alpha\beta})}-1$, and add a constant term to compensate for this change. The result is
\begin{equation}
  f(\widevec{xy}(\widevec{\alpha\beta})) = f(\vec 0) + \sum_{u=0}^3\sum_{I\in\Coeffs(2m,d)} i^u\cdot d_{u,I}\cdot \left(e^{i\cdot(I\cdot\widevec{\alpha\beta})} - 1\right);
  \label{eqn:polynomials-in-angular}
\end{equation}
by evaluating at $\widevec{\alpha\beta} = \vec 0$ (equivalently, $\widevec{xy} = 0$), it may be verified that the required constant term is in fact $f(\vec 0)$. Note that the four coefficients $d_{u,I}$ with $I=\vec 0$ can be discarded, as these terms in~\eqref{eqn:polynomials-in-angular} are identically $0$.

Even though $f(\widevec{xy}(\widevec{\alpha\beta}))$ computes real values, these complex representations of $f$ will prove more useful for turning polynomials into linkages. In particular, the modified form~\eqref{eqn:polynomials-in-angular} has the advantage that each term other than the $f(\vec 0)$ term evaluates to $0+0i$ at $\widevec{\alpha\beta} = \vec 0$, which is crucial for ensuring that our linkages come with nice, rational coordinates when $f(\vec 0) = 0$.

\begin{lemma}
  \label{lem:angular-coefficients}
  If~$f$ has $2m$ variables of total degree at most~$d$ and coefficients with magnitude at most~$M$, and~$r > 0$, then in representation~\eqref{eqn:polynomials-in-angular} there are at most $2(2m)^d(2d+1)^d$ nonzero coefficients $d_{u,I}$, and the absolute values of all of these coefficients add to at most $6^d \cdot r^d\cdot M\cdot \binom{2m+d}{d}$. If $f$ has \emph{integer} coefficients and $r$ is a positive integer, then entries $d_{u,I}$ are also integers, and they may be computed from $f$ in deterministic time~$O(\poly(m^d,d^d,r^d,M))$.
\end{lemma}

\begin{proof}
  We may bound the size of $\Coeffs(2m,d)$ by overcounting as follows: construct an integer vector $I\in\bZ^{2m}$ by choosing $d$~entries (or choosing all of them if $d > 2m$), assigning each chosen entry a value from $\{-d,-d+1,\ldots,d\}$, and assigning $0$ for the remaining entries. The number of such $I$ is $\max(\binom{2m}{d},1)\cdot (2d+1)^d\le (2m)^d(2d+1)^d$. There are therefore at most $2\cdot (2m)^d(2d+1)^d$ nonzero coefficients $d_{u,I}$.

  For each monomial $c\cdot \widevec{xy}^J$ in $f$ of total degree $k\le d$, substitute each $x_j$ and $y_j$ as in Equation~\eqref{eqn:change-of-variables-explicit} and expand fully \emph{without collecting like terms}. There are at most $6^k$ terms in this expansion (interpreting $2$ as two separate terms $1+1$), each of the form $c\cdot r^k\cdot i^u\cdot e^{I\cdot\widevec{\alpha\beta}}$, so the sum of the magnitudes of these coefficients is at most $6^k\cdot r^k\cdot c \le 6^d\cdot r^d\cdot M$. Adding this quantity across all monomials in $f$, we obtain a sum no larger than $6^d \cdot r^d\cdot M\cdot \binom{2m+d}{d}$. Finally, collecting like terms (including cancelling $d_{u,I}$ and $d_{u+2,I}$ pairs as much as possible) can only decrease the sum of the magnitudes of the coefficients by the triangle inequality. The coefficients $d_{u,I}$ do not change when modifying Equation~\eqref{eqn:polynomials-in-angular-intermediate} to Equation~\eqref{eqn:polynomials-in-angular} except that the coefficients with $I=\vec 0$ can be discarded, as the corresponding terms are identically zero.

  If the coefficients of $f$ are integers, then every term in the above expansion also has integer coefficients. This expansion may be performed in time $O(\poly(m^d,d^d,r^d,M))$ with straightforward methods.
\end{proof}

We will use this angular representation as a template to compute each polynomial~$f$ in the linkage. Indeed, much like in Kempe's original strategy~\cite{Kempe-1876} and especially in Abbott et al.'s correction thereof~\cite{Abbott-2008}, we provide gadgets for the following tasks.
\begin{itemize}
\item The Start Gadget (Figure~\ref{fig:start-gadget}) converts from rectangular position $(x_j,y_j)$ to angles $(\alpha_j,\beta_j)$.
\item The Angle Sum Gadget (Figure~\ref{fig:angle-sum-gadget}), built from the Angle Average Gadget (Figure~\ref{fig:angle-average-gadget}), allows adding and subtracting angles to construct all of the $I\cdot\widevec{\alpha\beta}$ values.
\item The Vector Creation Gadget (Figure~\ref{fig:vector-create-gadget}) and Vector Rotation Gadget (Figure~\ref{fig:vector-rotate-gadget}) construct the vectors $i^u\cdot d_{u,I}\cdot e^{i(I\cdot\widevec{\alpha\beta})}$.
\item The Vector Sum Gadget (Figure~\ref{fig:vector-sum-gadget}), built from the Vector Average Gadget (Figure~\ref{fig:vector-average-gadget}), allows adding vectors to construct the values $f(\widevec{xy}(\widevec{\alpha\beta})) - f(\vec 0)$ for each $f\in F$.
\item The End Gadget (Figure~\ref{fig:end-gadget}) constrains these values to equal $-f(\vec 0)$.
\end{itemize}
We employ several new ideas to ensure the resulting extended linkage $\cE(F)$ is noncrossing.
First, we construct a rigid grid of large square cells. Each gadget is isolated in one or $O(1)$ of these cells, and information is passed between gadgets/cells only using sliceform vertices along grid edges. In this way, these modular gadgets may be analyzed individually, as there is no way for distinct gadgets to intersect each other.
Second, the linkage $\cE(F)$ is an $(\eps,\delta)$-extended linkage, where~$\eps$ and~$\delta$ (the angle tolerances in $\AngleCon_{\cE(F)}(\sigma,A,\Delta)$) are global constants that restrict movement enough to prevent crossings within each separate gadget. The values of constants $\eps$ and $\delta$ are discussed in the next subsection.
We rely on the Copy Gadget (Figure~\ref{fig:copy-gadget}) to copy angles and propagate them along paths of cells to distant gadgets in the grid. The Crossover Gadget (Figure~\ref{fig:crossover-gadget}) allows these paths to cross, so we are not restricted to planar communication between gadgets. All of our gadgets make frequent use of the Parallel Gadget in Figure~\ref{fig:parallel-gadget}, which (with pins removed) keeps segments parallel without otherwise restricting motion (within some neighborhood). Figure~\ref{fig:example-angle-sum} shows an example of the gadgets working together.

Finally, we simulate extended linkage $\cE(F)$ with a partially rigidified linkage $\cL(F)$, in two steps. First, by replacing a vicinity of each sliceform vertex in $\cE(F)$ with the Sliceform Gadget (Figure~\ref{fig:sliceform-gadget}), we construct an extended linkage $\cE'(F)$ that perfectly simulates $\cE(F)$ but has no sliceforms. Then, we replace each edge of $\cE'(F)$ with a rigidified orthogonal tree, connected to neighboring edges with the Angle Restrictor Gadget (Figure~\ref{fig:angle-restrictor-gadget}), which enforces the precise bounds in the angle constraints.

\subsection{Constants and Parameters}
\label{sec:parameters}

Our construction uses a few carefully chosen and interconnected parameters, which we explain and collect here for ease of reference.

Linkage $\cE(F)$ is an $(\eps,\delta)$-extended linkage, where the angle tolerances $\eps$ and $\delta$ are global constants tuned for specific purposes. The parameter~$\eps$ has a single, simple goal: to constrain bar movement enough to ensure uniqueness (but not necessarily existence) of configurations. For example, Kempe's fatal flaw was failing to prevent parallelograms from ``flipping'' into contraparallelograms, and vice versa. Abbott et al.~\cite{Abbott-2008} resolved this problem by using an additional \emph{brace} along the midline of the contraparallelogram. Our angle constraints provide an alternative solution\footnote{Angle constraints \emph{per se} only apply to \emph{constrained} linkages. It is more correct to say that we resolve parallelogram flipping via angle constraints \emph{combined} with a method of \emph{simulating} these constraints using classical linkages.}: if a rectangular linkage is given angle constraints with tolerance $\eps < \pi/2$, then it cannot flip into a contraparallelogram, since one of its angles, constrained to the interval $\pi/2\pm\eps$, would be forced to lie in the disjoint interval $3\pi/2 \pm \eps$.

The smaller parameter $\delta$ has a double purpose. First, $\delta$ further restricts movement enough to prevent intersections and to control the global minimum feature size. But simply upper bounding how much each gadget can move is not enough: at some point, we must be able to say that our linkage can move \emph{at least} enough to draw the loci we are interested in. This is $\delta$'s second purpose: in each of our gadgets, we ensure that every corner with tolerance~$\delta$ can in fact realize any angle offset in the entire interval $[-\delta,\delta]$. In other words, while $\eps$ ensures uniqueness, $\delta$ measures \emph{existence}.

The precise values we have chosen are
\begin{align}
  \label{eq:eps}
  \eps := &\cos\inv\left(1-\frac{3}{10} \cdot \frac{2n_\eps}{n_\eps^2+1}\right),
  \quad \text{where } n_\eps = 5\cdot 10^{3}; \quad \text{and}
  \\
  \label{eq:delta}
  \delta := &\cos\inv\left(1-\frac{3}{10} \cdot \frac{2n_\delta}{n_\delta^2+1}\right),
  \quad \text{where } n_\delta =  4\cdot 10^{14}.
\end{align}
Numerically, these come out to approximately $\eps \approx 0.0155$ and $\delta\approx 5.48\cdot 10^{-8}$. Most uses of $\eps$ and $\delta$ rely only on the facts that $\eps \le \pi/16$ and $\delta\cdot 2^{18}\le \eps$, which may be numerically verified for the choices above. Only our final gadget, the Angle Restrictor Gadget in Section~\ref{sec:sliceform-and-angle-restrictor-gadgets}, places additional constraints on these values, in terms of both size and form: the trigonometric forms allow the Angle Restrictor Gadget to have rational coordinates.

Henceforth, all $(\eps,\delta)$-extended linkages will use these global constants for $\eps$ and $\delta$, so we will refer to them simply as ``extended linkages''.

Recall that the input $F$ is a set of polynomials in variables $\widevec{xy} = (x_1,y_1,\ldots,x_m,y_m)$ each with total degree at most $d$ and with integer coefficients bounded in magnitude by $M$. There are $|F| = s$ polynomials. The recurring parameters used in our construction
are as follows:
\begin{itemize}
\item $r = \lceil d/\delta\rceil$ is half the radius used in the angular change of coordinates $(x,y) = 2r\cdot\Rect(\alpha,\beta)$ applied to each polynomial.
\item $Q = 40\lceil 6^d r^d M \binom{2m+d}{d}/(6\delta)\rceil$ is the side-length of each cell in the background grid, chosen large enough to accommodate the necessary constructions. It is a positive integer divisible by $40$.
\item $R = 3Q/10$ is the radius used in angular coordinate representations in most of the gadgets. It is a positive integer, and $R > 4r$ may be readily verified.
\end{itemize}
It may seem unnecessary to use two different radii ($2r$ and $R$) for angular coordinates, but there is a good reason to do so. We evaluate the angular representation~\eqref{eqn:polynomials-in-angular} of each polynomial term-by-term, constructing partial sums along the way. The size of these partial sums grows as a function of the radius $2r$ used in the polynomial change of coordinates, and they may be much larger than $2r$. We thus make $r$ as small as possible to control the size of the partial sums, while $R$ must be large enough to be able to compute the partial sums. As consolation, choosing a constant ratio $R/Q$ allows most gadgets to scale uniformly with $Q$, making them simpler to analyze.

We have not attempted to optimize the parameters $\eps,\delta,r,Q,R$ or the gadgets they depend on, instead choosing the most expedient values that allow the proofs to succeed.
In fact, for convenience, we round many constant factors to powers of~$2$.
Significant improvements are likely possible.

\subsection{Hook Linkage and Parallel Linkage}

To streamline the analysis of the gadgets in Sections~\ref{sec:gadgets} and~\ref{sec:sliceform-and-angle-restrictor-gadgets}, we collect here a few tools to provide explicit bounds on the displacement and rotation of certain linkages under small perturbation. We show that the \emph{Parallel Gadget} in Figure~\ref{fig:parallel-gadget}, which is used in every gadget in Section~\ref{sec:gadgets}, constrains two edges $a b$ and $e f$ to remain parallel without otherwise restricting their movement, so long as neither tries to move or rotate too far from how it started (Lemma~\ref{lem:parallel-gadget}). Along the way we analyze the \emph{hook linkage}, i.e., a path of two edges $a b$ and $b c$: if vertices $a$ and $c$ are each displaced by a small amount, we provide explicit bounds on how far $b$ must move to preserve the lengths of edges $a b$ and $b c$ (Lemma~\ref{lem:hook-linkage}). This hook linkage bound is helpful for the parallel linkage, but it is also directly invoked by a few of the gadgets in Sections~\ref{sec:gadgets} and~\ref{sec:sliceform-and-angle-restrictor-gadgets}.

\begin{lemma}\label{lem:sin-and-arcsin}
  For $0 < t < 1/2$, we have $\sin t > 0.9 t$ and $\sin\inv t < 1.1 t$.
\end{lemma}
\begin{proof}
  For $0 < t < 1/2$ we may compute
  \begin{equation*}
    \sin t - \frac{9}{10}t
    > \left(t-\frac{1}{6}t^3\right)-\frac{9}{10}t
    = \frac{1}{30}t(3-5t^2)
    > 0
  \end{equation*}
  and
  \begin{equation*}
    \frac{11}{10} t - \sin\inv t
    \ge \frac{11}{10}t - \left(t + \frac{1}{6}t^3\right)
    = \frac{1}{30}t(3-5t^2)
    > 0.\qedhere
  \end{equation*}
\end{proof}

\begin{lemma}\label{lem:arg-of-quotient}
  If $A$ and $B$ are complex numbers with $|B| \le \frac{1}{2}|A|$, and if $\theta = \arg \frac{A+B}{A}$, then $|\theta|\le\sin\inv\frac{|B|}{|A|} \le 1.1\frac{|B|}{|A|}$.
\end{lemma}
\begin{proof}
  By replacing $(A,B)$ with $(1,B/A)$, we may assume $A=1$. Let $|B| = t$. Geometrically, $1+B$ lies on a circle with center $1$ and radius $t$, and the points $p$ on this circle where $\arg p$ is maximized or minimized are the points of contact of the lines through the origin that are \emph{tangent} to this circle. These angles are $\sin\inv t = \sin\inv\frac{|B|}{|A|}$, as claimed. Lemma~\ref{lem:sin-and-arcsin} shows that $\sin\inv t \le 1.1 t$ so long as $t\le 1/2$.
\end{proof}

\begin{lemma}\label{lem:rotating-bar-displacement}
  If an edge of length $\ell$ rotates around an endpoint by angle $\theta$, the other endpoint lands no farther than $\ell\cdot|\theta|$ from its initial position. More generally, if a \emph{path} of edges $v_1 v_2, v_2 v_3,\ldots,v_{k-1} v_k$ of total length $\ell$ reconfigures itself so that $v_1$ stays fixed while each edge changes its angle by at most $\pm\theta$, then vertex $v_k$ is displaced by at most $\ell\cdot|\theta|$.
\end{lemma}
\begin{proof}
  In the case of a single edge of length $\ell$, the distance in question is
  \begin{equation*}
    \ell\cdot \left|e^{i\theta} - 1\right|
    = \ell\cdot \left|2\sin\frac{\theta}{2}\right|
    \le \ell\cdot \left|2\cdot\frac{\theta}{2}\right|
    = \ell\cdot|\theta|,
  \end{equation*}
  as claimed. The more general claim follows by applying the previous fact to each edge separately, noting that the displacements simply add together, and applying the triangle inequality to their sum.
\end{proof}

\begin{lemma}\label{lem:feature-size}
  Let $C$ be a noncrossing configuration of a (possibly constrained) linkage $\cL$ that is noncrossing with minimum feature size $f$. If $C'$ is another configuration such that $|C'(v) - C(v)| \le h$ for each vertex $v$, and if $h < f/2$, then $C'$ is noncrossing with minimum feature size $f-2h$.
\end{lemma}
\begin{proof}
  Let $p$ be a general point on $\cL$, by which we mean either a vertex $v$ or a fixed ratio $k$ along an edge $v w$; call the latter an \emph{edge point}. For an edge point $p$, $C(p) := (1-k)\cdot C(v) + k\cdot C(w)$, and likewise for $C'(p)$. We claim that $|C'(p) - C(p)| \le h$ for any general point $p$. The result is immediate if $p$ is a vertex, so consider the edge point case:
  \begin{equation*}
    \left|C'(p) - C(p)\right|
    \le (1-k)\cdot |C'(v) - C(v)| + k\cdot |C'(w) - C(w)|
    \le (1-k)\cdot h + k\cdot h = h,
  \end{equation*}
  as claimed. It follows that if $p$ and $q$ are two general points on $\cL$ that start with $|C(p)-C(q)| \ge f$, then $|C'(p) - C'(q)| \ge f-2h$ by the triangle inequality. This means $C'$ has minimum feature size $f-2h$, so long as we can show that $C'$ is noncrossing.

  So suppose $C'$ has a crossing. This means either a vertex $p=u$ coincides with a point $q$ on edge $v w$ not incident with $u$, or two edge points $p$ and $q$ from non-incident edges $u_1 u_2$ and $v_1 v_2$ coincide.
  Then $C'(p) = C'(q)$.
  On the other hand, $|C(p) - C(q)| \ge f$ because the smallest distance
  between non-incident edges is lower bounded by the smallest distance
  between a vertex and a non-incident edge,
  which is the minimum feature size~$f$ (Definition~\ref{def:noncrossing}).
  As noted above, $|C'(p)-C'(q)| \ge f-2h$,
  which is strictly positive because $h < f/2$, contradicting that
  $C'(p) = C'(q)$.
  For completeness, we include a proof of the second case.
  Consider triangle $C(p)C(v_1)C(v_2)$. If one of the triangle's angles at $C(v_1)$ or $C(v_2)$ is at least $\pi/2$, then $C(p)C(q)$ is at least as long as one of the features $C(p)C(v_j)$. So we may assume both these angles are $\le\pi/2$, meaning the projection of $C(p)$ onto line $C(v_1)C(v_2)$ lands on segment $C(v_1)C(v_2)$. Redefine $q$ as this projection, which does not lengthen $C(p)C(q)$. Now slide $p$ and $q$ while keeping $C(p)C(q)$ perpendicular to $C(v_1)C(v_2)$; there is some direction of sliding that does not lengthen $C(p)C(q)$, so continue until one of $p$ or $q$ hits an endpoint of its segment. Thus we obtain a distance between a vertex and a non-incident edge that is no longer than the original $C(p)C(q)$, as desired.
\end{proof}

\begin{lemma}[Hook Linkage]\label{lem:hook-linkage}
  Define $\cLhook = \cLhook(\ell_1,\ell_2,\theta)$ as the unconstrained linkage shown in Figure~\ref{fig:hook-linkage}, initially configured at $a_0 := C_0(a)=0$, $b_0 := C_0(b) = \ell_1$, and $c_0 := C_0(c) = \ell_1+\ell_2e^{i\theta}$, where $\theta\in[\pi/3,2\pi/3]$.

  \begin{enumerate}
  \item\label{thmpart:hook-pinned}
    For any point $c'$ in the plane, there is at most one configuration with vertex $c$ resting at point $c'$ such that $\angle cba \in (0,\pi)$. If $h\le\min(\ell_1,\ell_2)/2^4$ and $|c'-c_0|\le h$, this configuration indeed exists, and it has the additional properties that $b$ has been displaced at most $|b-b_0| \le 1.5h$ from its initial position, the angle of bar $ab$ has changed by at most $\pm 1.5h/\ell_1$, and likewise the angle of $bc$ has changed by at most $\pm 1.5h/\ell_2$.
    
  \item\label{thmpart:hook-unpinned}
    If the pin in $\cLhook$ is removed and any points $a'$ and $c'$ in the plane are chosen, there is at most one configuration placing vertex $a$ at $a'$ and vertex $c$ at $c'$ having $\angle cba\in(0,\pi)$. If  $h \le\min(\ell_1,\ell_2)/2^5$ and $|a'-a_0|,|c'-c_0|\le h$, then this configuration exists, and it has the additional properties that $b$ has moved at most $|b-b_0| \le 4h$ from its initial position, the angle of bar $ab$ has changed by at most $\pm 4h/\ell_1$, and the angle of $bc$ has changed by at most $\pm 4h/\ell_2$.
  \end{enumerate}
\end{lemma}

\begin{proof}
  For Part~\ref{thmpart:hook-pinned}, we first prove uniqueness. Suppose $C$ some configuration has $C(a) = a_0$ and $C(c) = c'$. Observe that $C(b)$ must lie on the circle centered at $a_0$ with radius $\ell_1$, as well as the circle centered at $c'$ with radius $\ell_2$. These circles cannot be concentric: indeed, $C(a)=C(c)$ would force $\angle C(c)C(b)C(a) = 0$, which is outside the allowed bounds.
  It follows that there are at most two positions for $b$, namely at $C(b)$ or at the reflection of $C(b)$ across $C(a)C(c)$. If $C'$ is the configuration that uses the latter option, then $\angle C'(c)C'(b)C'(a) = 2\pi - \angle C(c)C(b)C(a) \notin (0,\pi)$. So $C$ is indeed unique, as claimed.

  We now prove existence, so suppose $|c'-c_0| \le h$ as above. If bars $ab$ and $bc$ rotate by angles $\alpha$ and $\beta$ respectively, point $c$ ends up at
  \begin{equation*}
    H(\alpha,\beta) := \ell_1 e^{i\alpha} + \ell_2 e^{i(\theta+\beta)}.
  \end{equation*}
  (This function $H$ is similar to the $\Rect$ function defined in Section~\ref{sec:detailed-overview}, except that (1) $H(0,0)$ has not been shifted to the origin, and (2) $H$ will not be used beyond this proof.)
  Let $U = [-1.5h/\ell_1, 1.5h/\ell_1]\times [-1.5h/\ell_2, 1.5h/\ell_2]$. We will first show that there are some angles $(\alpha,\beta)$ in $U$ with $H(\alpha,\beta) = c'$. To that end, we define a related function
  \begin{equation*}
    G(\alpha,\beta) := \ell_1 (1 + i\sin\alpha) + e^{i\theta}\cdot\ell_2(1+i\sin\beta).
  \end{equation*}
  When $|\alpha| \le 1.5h/\ell_1$ we may verify that
  \begin{gather*}
    \left|\ell_1 (1 + i\sin\alpha) - \ell_1 e^{i\alpha}\right|
    = \ell_1\cdot\left|1-\cos\alpha\right|
    = \ell_1\cdot 2\sin^2\frac{\alpha}{2}\\
    \le \ell_1\cdot\frac{1}{2}\alpha^2
    \le \ell_1\cdot \frac{1}{2}\left(\frac{3h}{2\ell_1}\right)^2
    = \frac{9h^2}{8\ell_1}
    \le \frac{h}{12},
  \end{gather*}
  where the last inequality follows from $h/\ell_1 \le 1/16$. By similar comparison between the other two terms of $H$ and $G$, we conclude that when $(\alpha,\beta)\in U$, $|H(\alpha,\beta) - G(\alpha,\beta)| \le h/6$.

\begin{figure}[hbt]
  \begin{minipage}[t]{0.42\textwidth}
    \centering
    \includegraphics{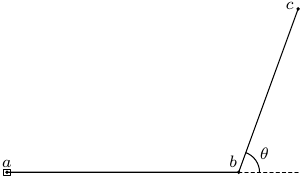}
    \caption{The Hook Linkage from Lemma~\ref{lem:hook-linkage} consists of two incident bars. }
    \label{fig:hook-linkage}
  \end{minipage}
  \hfill
  \begin{minipage}[t]{0.53\textwidth}
    \centering
    \includegraphics{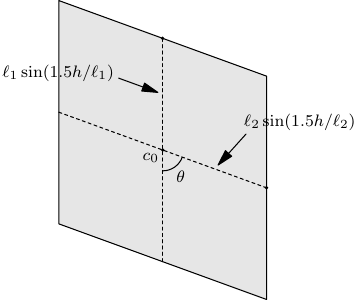}
    \caption{The region $G(U)$ in Lemma~\ref{lem:hook-linkage} is a parallelogram centered at $c_0$ with dimensions as shown.}
    \label{fig:hook-linkage-region}
  \end{minipage}
\end{figure}

  The linear map $G$ is easy to describe: region $G(U)$ is a parallelogram centered on $c_0$ with dimensions as shown in Figure~\ref{fig:hook-linkage-region}. The horizontal distance from the center to (the line containing) either vertical edge is
  \begin{equation*}
    \ell_1\cdot \sin\frac{3h}{2\ell_1}\cdot \sin\theta
    \ge \ell_1\cdot \frac{9}{10}\frac{3h}{2\ell_1}\cdot \frac{\sqrt{3}}{2}
    > \frac{7}{6} h,
  \end{equation*}
  where our use of Lemma~\ref{lem:sin-and-arcsin} is valid because $1.5h/\ell_1 \le 3/32 < 1/2$. The distance to the (line containing) each diagonal edge is likewise greater than $7h/6$, so the interior of $G(U)$ contains the disk centered at $c_0$ of radius $7h/6$. Restricting attention to the boundary of $U$, we find that the loop $G(\partial U)$ surrounds this disk and is disjoint from it. In particular, this loop has winding number $1$ around $c'$ and does not come within distance $h/6$ of $c'$.

  By our comparison of $H$ and $G$ above, the straight-line homotopy between loops $G(\partial U)$ and $H(\partial U)$ (which moves corresponding points along straight line segments) does not meet $c'$, so the latter loop also has winding number $1$ around $c'$. If $c'$ were not contained in region $H(U)$, then continuous function $H$ would map the contractible set $U$ into a non-contractible subset of $\bR^2\setminus c'$, which is impossible. So there indeed exists a pair of angles $(\alpha,\beta)\in U$ with $H(\alpha,\beta) = c'$.

  Using these angles $\alpha$ and $\beta$, we claim that the configuration $C$ that places $a$ at $a_0=0$, $b$ at $\ell_1 e^{i\alpha}$ and $c$ at $H(\alpha,\beta) = c'$ will exhibit all of the desired properties. The angle at $b$ will have changed by at most $1.5h/\ell_1 + 1.5h/\ell_2 \le 3/16 < \pi/3$, so it will still lie in $(0,\pi)$. Also, $b$ has moved by at most
  \begin{equation*}
    \left|b - b_0\right|
    = \ell_1\cdot\left|e^{i\alpha}-1\right|
    = 2\ell_1\cdot\left|\sin\frac{\alpha}{2}\right|
    \le \ell_1\cdot\left|\alpha\right|
    \le 1.5h,
  \end{equation*}
  as required.

  For Part~\ref{thmpart:hook-unpinned}, apply Part~\ref{thmpart:hook-pinned} with target position $\hat c := c'-a'$ and distance $2h$.
  Let $\hat a = 0, \hat b, \hat c$ be the resulting positions
  for vertices $a,b,c$.
  This configuration satisfies $|\hat b - b_0| \leq 1.5 \cdot 2h = 3h$,
  the angle of bar $ab$ is in the interval $0 \pm 1.5\cdot 2h/\ell_1 = 0 \pm3h/\ell_1$,
  and likewise the angle of bar $bc$ is in the interval
  $\theta \pm 1.5\cdot 2h/\ell_2 = \theta \pm 3h/\ell_2$.
  Now we translate this configuration by $a'$ to obtain
  vertices at $a' = \hat a + a', b' = \hat b + a', c' = \hat c + a'$.
  By the triangle inequality, $|b'-b_0| = |\hat b - b_0 + a'| \leq
  |\hat b - b_0| + |a'| \leq 3h + h = 4h$.
  The bar angles remain the same under translation;
  we increase the lead constant from $3$ to $4=2^2$ for future ease of use.
\end{proof}

We rely heavily on parallelograms to force edges to remain parallel. As described in Section~\ref{sec:parameters}, angle constraints can be used to ensure that parallelograms remain parallelograms, forbidding the contraparallelogram ``flip'' that Kempe failed to account for. We formalize this as follows:

\begin{figure}[hbtp]
  \centering
  \begin{subfigure}{.48\textwidth}
    \centering
    \includegraphics{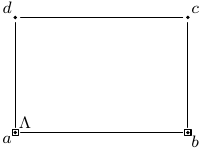}%
    \caption{Parallelogram linkage $\cP_1$, initially with opposite corners at $a=(0,0)$ and $c=(x,y)$ where $x,y\ge 1$.}
    \label{fig:parallelogram-linkage-rectangle}
  \end{subfigure}
  \hfill
  \begin{subfigure}{.48\textwidth}
    \centering
    \includegraphics{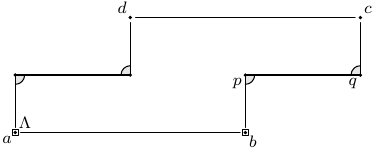}
    \caption{Parallelogram linkage $\cP_2$, with corners initially at $a=(0,0)$, $b=(4,0)$, $c=(6,2)$, and $d=(2,2)$.}
    \label{fig:parallelogram-linkage}
  \end{subfigure}
  \caption{Creating parallelograms with extended linkages. Corners marked with gray sectors are frozen, and vertices enclosed in squares are pinned, as in Convention~\ref{con:drawing-extended-linkages}.}
  \label{fig:parallelograms}
\end{figure}

\begin{lemma}[Parallelogram Linkage]
  \label{lem:parallelogram-linkage}
  Extended linkages $\cP_1$ and $\cP_2$ of Figures~\ref{fig:parallelogram-linkage-rectangle} and~\ref{fig:parallelogram-linkage}, whose corners are given $\eps$-tolerance angle constraints, are globally noncrossing with global minimum feature size at least $1/2$, and every configuration of either linkage has $\widevec{dc} = \widevec{ab}$. Furthermore, the configuration spaces $\Conf(\cP_1)$ and $\Conf(\cP_2)$ are perfectly described by the map $\Offset_{\Lambda}$, which provides homeomorphisms
  \begin{equation*}
    \Conf(\cP_1) \simeq [-\eps,\eps] \qquad\text{and}\qquad
    \Conf(\cP_2) \simeq [-\eps,\eps].
  \end{equation*}
  In every configuration of $\cP_2$, all vertices other than $a$ and $b$ are configured at least $1/2$ units away from the line through $a$ and $b$.
\end{lemma}

Note that the pins in $\cP_1$ and $\cP_2$ merely serve to fix the reference coordinate frame. Without pins, the angle displacement at corner $\Lambda$ uniquely determines the locations of $c$ and $d$ \emph{relative to $a$ and $b$}.

\begin{proof}
  Begin with linkage $\cP_1$ in Figure~\ref{fig:parallelogram-linkage-rectangle}, which has initial configuration $a=(0,0)$ and $c=(x,y)$. In any configuration $C$, $$d = y\cdot(\cos(\pi/2+\theta),\sin(\pi/2+\theta))$$ where $\theta = \Offset_{\Lambda}(C) \in[-\eps,\eps]$. Based on edge lengths, vertex $c$ has two potential positions: $b+d-a$ (forming a parallelogram), or the reflection of this point across diagonal $bd$ (forming a contraparallelogram). The former case indeed satisfies all of its constraints and is noncrossing with feature size at least $\min(x,y)/2$ because $\eps < \pi/4$. In the latter case, corner $dcb$ would have angle $3\pi/2-\theta$, which is well outside the allowable range of $[\pi/2-\eps,\pi/2+\eps]$. So the parallelogram configuration exists uniquely.

  In linkage $\cP_2$, the vertices have initial positions $a=(0,0)$, $b=(4,0)$, $c=(6,2)$, and $d=(2,2)$. By similar arguments, each configuration $C$ has $c-d = b-a$ and is uniquely determined by $\theta = \Offset_{\Lambda}(C)$. If the minimum feature size were to drop below $1/2$, this would happen first when $p$ came that close to edge $a b$ or $c d$. We may calculate that these events correspond to respective angles $\theta = \pi/3 = 60^\circ$ and
  \begin{equation*}
    \theta = \sin\inv\frac{1/2}{\sqrt{5}} -  \sin\inv \frac{1}{\sqrt{5}} \approx -13.64^\circ.
  \end{equation*}
  Because $\eps < \pi/16 = 11.25^\circ$, the interval $[-\eps,\eps]$ keeps the minimum feature size above $1/2$, as required. We may also observe that, in this interval, no vertex comes closer than $1/2$ to the $x$-axis, as required.
\end{proof}

Parallelograms are especially useful in pairs, as with the linkage in Figure~\ref{fig:parallel-gadget}, which forces bars $ab$ and $ef$ to remain parallel while letting them move freely relative to each other (within some neighborhood). The classical counterpart of this linkage \cite{Kempe-1876} suffers from nonuniqueness: When positions for bars $ab$ and $ef$ are chosen (necessarily parallel), usually there are two possible locations for bar $cd$. Angle constraints again improve the situation: one of these two positions violates the angle constraints (at all six named vertices!) and is therefore invalid.

\begin{figure}[hbtp]
  \centering
  \includegraphics{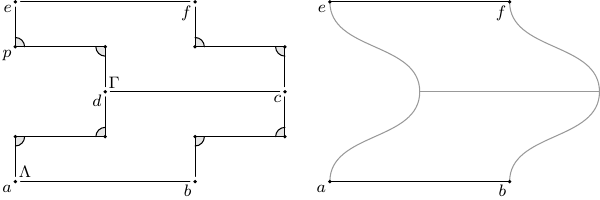}
  \caption{
    Left: The Parallel Gadget allows $e$ to move freely in a neighborhood of its initial position while forcing $ef$ to remain parallel to $ab$. Corners marked with gray sectors are frozen, as in Convention~\ref{con:drawing-extended-linkages}.
    Edge lengths are all in $\{1,2,4\}$ and are drawn to scale.
    Right: a schematic representation of the same gadget.
    \label{fig:parallel-gadget}
  }
\end{figure}

\begin{lemma}[Parallel Gadget]
  \label{lem:parallel-gadget}
  \ %
  \begin{enumerate}
  \item\label{thmpart:parallel-gadget-unique} The extended linkage $\cLparallel$ drawn in Figure~\ref{fig:parallel-gadget}, in which all corners are given an $\eps$ angle constraint, is globally noncrossing with minimum feature size at least $\frac12$. In every configuration, vectors $ab$, $dc$, and $ef$ are equal. Any choice of position for $a$, $b$, and $e$ can be completed to \emph{at most one} configuration of $\cLparallel$.

  \item\label{thmpart:parallel-gadget-exists} If three points $a'$, $b'$, and $e'$ are chosen at distance at most $\eps/2^7$ from their initial positions of $a_0 = (0,0)$, $b_0=(4,0)$, and $e_0=(0,4)$ respectively such that $|b'-a'| = 4$, then a unique configuration $C$ exists with $C(a) = a'$, $C(b) = b'$, and $C(e) = e'$. This configuration $C$ moves each vertex no farther than~$\eps/2$ from its initial position and rotates each edge by at most $\pm\eps/2$.
  \end{enumerate}

\end{lemma}

\begin{note}
  In this Parallel Gadget, once $a$ and $b$ have been configured, angles $\Offset_{(\Lambda,\Gamma)}(C) = (\alpha,\beta)\in[-\eps,\eps]^2$ uniquely determine the rest of the configuration $C$, but it is worth mentioning that not all such pairs $(\alpha,\beta)$ actually give rise to valid configurations. Indeed, the angle to the left of vertex $d$ becomes $\pi+\alpha-\beta$, so the $\eps$-tolerance forces $|\alpha-\beta|\le\eps$ as an additional constraint (in fact, the only other constraint). This lemma does not need to worry about this fact because it does not attempt to explore the entire configuration space of $\cLparallel$, instead restricting attention to $|\alpha|,|\beta| \le \eps/2$.
\end{note}

\begin{proof}
  Let $C_0$ be the initial configuration of $\cLparallel$ shown in Figure~\ref{fig:parallel-gadget}, and for each vertex $v$, we use the shorthand $v_0 = C_0(v)$.
  
  This linkage is made from two instances of $\cP_2$ (one reflected) attached along a bar. In any configuration $C$ of $\cLparallel$, by Lemma~\ref{lem:parallelogram-linkage}, vectors $ab$ and $dc$ must be equal, as must $dc$ and $ef$. Each vertex other than $c$ and $d$ has distance at least $1/2$ from the line containing $c$ and $d$, so the smallest feature of $C$ comes from one of the two instances of $\cP_2$ and is thus at least $1/2$.

  For uniqueness, suppose positions $a'$, $b'$, and $e'$ have been chosen for $a$, $b$, and $e$. First observe that the position of $d$ is forced: treating the paths from $a$ to $d$ and from $d$ to $e$ as bars $a d, d e$ of fixed lengths $\ell(a d)=\ell(d e)=2\sqrt{2}$, and thus $a d e$ as a \emph{hook linkage}, Lemma~\ref{lem:hook-linkage} shows that there exists a unique point $d'$ at distance $2\sqrt{2}$ from each of $a'$ and $e'$ with $\angle e'd'a' < \pi$, which is required by one of the angle constraints at vertex $d$. The previous paragraph then shows that $c$ must be placed at $c' = d'+(b'-a')$ and $f$ must be placed at $e'+(b'-a')$. Finally, the rigidified paths (e.g., from $a$ to $d$) are uniquely determined by their endpoints, so all vertex locations are forced, as claimed.

  Define $k := \eps/2^7$. To prove Part~\ref{thmpart:parallel-gadget-exists}, first consider the easier case where $a'=a_0$, $b'=b_0$, and $|e'-e_0| \le 6k$. We will show that a (necessarily unique) configuration $C$ with $C(a)=a'=a_0$, $C(b) = b'=b_0$, and $C(e) = e'$ exists, moves each vertex by at most $28k$, and rotates each edge by at most $9k$.

  As described above, Lemma~\ref{lem:hook-linkage} (Hook Linkage) allows us to construct the unique positions $b'$, $c'$, and $f'$ where vertices $b,c,f$ must lie, which then determines the entire configuration; we must show that this forced configuration $C$ is actually valid, i.e., satisfies its angle constraints. Because $6k < \min(\ell(a d),\ell(d e))/32 = (2\sqrt{2})/32$, Lemma~\ref{lem:hook-linkage} proves that $a'd'$ and $d'e'$ are each rotated from $a_0d_0$ and $d_0e_0$ by angle at most $\pm 4\cdot 6k/(2\sqrt{2}) = 6k\sqrt{2}\le 9k$, so the angles at $C$'s corners differ from those of $C_0$ by at most $18k < \eps$ and therefore satisfy their $\eps$-tolerance angle constraints. The same lemma also shows that $|c'-c_0| = |d'-d_0| \le 24k$. The bound $|f'-f_0| = |e'-e_0| \le 6k\le 24k$ follows by assumption.

  It remains to show that $|C(v)-v_0| \le 28k$ for each vertex $v$ of $\cLparallel$ other than $a,b,c,d,e,f$.
  Focus on vertex $p$ first, and let $p' = C(p)$.
  Because the path from $e$ to $d$ through $p$ is frozen
  in Figure~\ref{fig:parallel-gadget},
  triangles $e_0 p_0 d_0$ and $e' p' d'$ are congruent.
  Thus $p_0$ and $p'$ can be written as the same affine linear combination of
  $e_0,d_0$ and $e',d'$ respectively.
  Because $d_0 = (2,2) = 2+2i$, $e_0 = (0,4) = 4i$, and $p_0 = (0,3) = 3i$,
  we have the (unique) solution
  $$p_0 = \left(\frac{1-i}{4}\right) d_0 + \left(\frac{3+i}{4}\right) e_0,$$
  and thus the same for $p'$.  Therefore
  \begin{equation*}
    p'-p_0 = \left(\frac{1-i}{4}\right)(d'-d_0) + \left(\frac{3+i}{4}\right)(e'-e_0).
  \end{equation*}
  Because $|d'-d_0|$ and $|e'-e_0|$ are upper bounded by $24k$
  (as argued above), the triangle inequality gives
  \begin{equation*}
    \left| p' - p_0 \right| \le \left|\frac{1-i}{4}\right|\cdot 24k + \left|\frac{3+i}{4}\right|\cdot 24k
    \le 28k.
  \end{equation*}
  The remaining vertices other than $a,b,c,d,e,f$ follow the same analysis
  as $p$
  because the four frozen paths are isometric, leading to the same affine
  weights or their conjugates, and all six points $a,b,c,d,e,f$ are displaced
  by at most $24k$ (as argued above).

  Now we tackle Part~\ref{thmpart:parallel-gadget-exists} in general; as was done for the Hook Linkage in Lemma~\ref{lem:hook-linkage}, we will reduce this general case to the special case where $a$ and $b$ are fixed. So suppose $a'$, $b'$, and $e'$ are chosen at distance at most $k$ from $a_0,b_0,e_0$ respectively and with $|a'-b'| = 4$. Uniqueness is guaranteed as above, so we focus on existence.

  Let $\theta = \arg (b'-a')$, i.e., $b'-a' = 4\exp(i\theta)$, and observe that $|\theta| \le 1.1\cdot \frac{2k}{4} \le k$ by Lemma~\ref{lem:arg-of-quotient}, using $A = 4$ and $B = b'-b_0 + a'-a_0$, which satisfy $|B| \le 2k < \frac{1}{2}|A|$.
  Define $a'' = a_0$, $b'' = b_0$, and $e'' = (e'-a')\cdot \exp(-i\theta)$, which may be obtained from $a',b',e'$ by translating by $-a'$ and then rotating by angle $-\theta$. We have
  \begin{align*}
    \left|e'' - e_0\right|
    &= \left|(e'-a')\cdot \exp(-i\theta) - e_0\right|\\
    &\le \left|(e'-a')\cdot \exp(-i\theta) - e_0\cdot \exp(-i\theta)\right| + \left|e_0\cdot \exp(-i\theta) - e_0\right|\\
    &= \left|(e'-a') - e_0\right| + \left|e_0\cdot \exp(-i\theta) - e_0\right|\\
    &\le |e'-e_0| + |a'| + \left|e_0\cdot \exp(-i\theta) - e_0\right|\\
    &\le k + k + 4|\theta|\\
    &\le 6k.
  \end{align*}
  By the special case discussed above, there exists a unique configuration $C$ with $C(a) = a_0$, $C(b) = b_0$, and $C(e) = e''$ where each vertex is at most $28k$ away from its initial position, and where each bar is rotated by at most $9k$ from its initial direction. Now define $C'$ as the configuration formed by rotating configuration $C$ through angle $\theta$ (around the origin) and then translating the result by $a'$. This $C'$ is a valid configuration of $\cLparallel$ because $C$ itself was, and furthermore, $C'(a) = a'$, $C'(b) = b'$, and $C'(e) = e'$. So $C'$ is the promised unique configuration.

  Finally, we measure how much $C'$ perturbs the vertices and edge rotations. For each vertex $v$, we have $|C(v)| \le |C_0(v)| + |C(v) - C_0(v)| \le 3\sqrt{5}+6k\le 7k$, so rotating by angle $\theta$ moves $C(v)$ by at most $7|\theta| \le 7k$ by Lemma~\ref{lem:rotating-bar-displacement}. Translating by $a'$ moves each vertex by another $|a'|\le k$, so in total,
  \begin{equation*}
    \left|C'(v) - v_0\right| \le 28k + 7k + k = 35k < \frac{\eps}{2}
  \end{equation*}
  for each vertex $v$. Likewise, $C$ has rotated each bar by at most $9k$, so $C'$ rotates each bar by at most $9k+|\theta|\le 10k < \eps/2$. This finishes the proof.
\end{proof}

\subsection{The Cell Gadgets}
\label{sec:gadgets}

With the above tools in place, we now present the modular gadgets themselves. Each will inhabit one or a constant number of $Q\times Q$ square cells, where $Q$ was defined in Section~\ref{sec:parameters}. Each cell has pairs of \term{transmission edges} attached at \term{transmission vertices}, which are sliceform vertices at the midpoints of cell edges (usually labelled $b_j$), and the corners around each transmission vertex will have tolerance~$\delta$ or~$0$---never~$\eps$---as required by the gadget. All other vertices will have tolerance~$\eps$ or~$0$, never $\delta$. If $v$ is a transmission vertex with transmission edge $v w$, the corner $u v w$ is the corresponding \term{transmission corner}; the angles at these transmission corners are the gadgets' only means of communicating with each other.

We use a few naming conventions for corners and their angle offsets. Corners are usually labelled by $\Lambda_j$ or $\Gamma_j$, chosen because these letters visually resemble corners, and their corresponding angle offsets are $\alpha_j = \alpha_j(C) :=  \Offset_{\Lambda_j}(C)$ and $\beta_j = \beta_j(C) := \Offset_{\Gamma_j}(C) = \beta_j$, chosen because `$\Lambda$` resembles a capital `A' and `$\Gamma$' resembles part of a capital `B'. These corners should be considered a pair: $\alpha_j$ and $\beta_j$ together will form the angular coordinates for some meaningful vector. Corners that are \emph{not} paired in this way will be labelled $\Theta_j$, corresponding to angle offsets $\theta_j = \theta_j(C) := \Offset_{\Theta_j}(C)$.

The analyses in this section do not depend on the precise values of $\eps$ and $\delta$; only on the facts that $\eps \le \pi/16$ and $\delta \le \eps/2^{18}$. For each gadget $\cL$ below, we will prove a statement of the following form: $\phi:\Conf(\cL)\to U$ is a homeomorphism, where $\phi$ is a function defined in terms of angle displacements $\Offset_Y$ and possibly vertex positions $\pi_X$, and $U$ is some subset of some $\bR^k$. This provides an explicit parametrization of $\cL$'s configuration space: each configuration is uniquely and continuously determined by just the corners in $Y$ and vertices in $X$, which are related by the shape of $U$. For example, the Angle Average Gadget $\cLangleavg$ is exactly parametrized by
\begin{equation*}
  \Uangleavg = \left\{(\theta_1,\theta_2,\theta_3)\in[-\delta,\delta]^3\;\middle|\; \theta_2 = \frac{\theta_1+\theta_3}{2}\right\}.
\end{equation*}
This means the angle offsets $\theta_j$ at corners $\Theta_j$ must satisfy this average condition, and any triple of angles in $[-\delta,\delta]$ that does satisfy it gives rise to a unique configuration of $\cLangleavg$. Furthermore, this configuration depends continuously on the angles $\theta_j$. Because $\Conf(\cL)$ is compact by Lemma~\ref{lem:extended-linkage-conf}, the following well-known fact says that this continuity comes for free once we know $\phi$ to be a bijection:

\begin{lemma}[{\cite[Theorem 4.17]{rudin}}]
  If $\phi: A\to B$ is a continuous bijection of metric spaces and $A$ is compact, then $\phi\inv$ is continuous, i.e., $\phi$ is a homeomorphism.
\end{lemma}

When using some of these gadgets later on,
we will not need to work with the exact configuration space $U$,
and it will instead suffice to know that the gadgets have sufficient mobility,
i.e., $U$ contains a large enough neighborhood of the initial configuration.
The lemma below also specifies these \term{useful neighborhoods} when appropriate.

For clarity, Figures \ref{fig:scale-invariant-gadgets},
\ref{fig:angle-average-gadget}, \ref{fig:unscaled-gadgets},
\ref{fig:vector-average-gadget}, and \ref{fig:crossing-end-gadget}
draw Parallel Gadgets using the schematic shorthand
from Figure~\ref{fig:parallel-gadget} (right).
All other vertices in these figures are explicitly marked by a dot or an `x',
and we call these vertices \term{marked}.
We will analyze marked vertices separately from the vertices internal to Parallel Gadgets.

We now analyze each gadget in turn. They are organized by similarity of \emph{analysis}, not of \emph{function}. This first set of gadgets all exhibit \emph{scale invariance}: all coordinates are constant multiples of the cell edge length $Q$, so increasing $Q$ serves only to increase the global minimum feature size of each gadget without otherwise affecting the gadget. This greatly simplifies the analysis.

\def\gadgetWidth{.48\textwidth}

\begin{figure}[p]
  \begin{subfigure}{\gadgetWidth}
    \centering
    \includegraphics{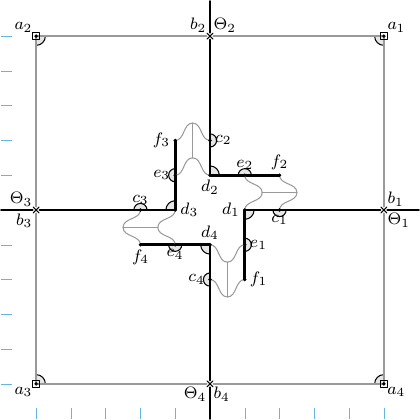}
    \caption{The Copy Gadget forces $\theta_1=\theta_2=\theta_3=\theta_4$. See this gadget's movement in Figure~\ref{fig:copy-gadget-snapshots}.}
    \label{fig:copy-gadget}
  \end{subfigure}
  \hfill
  \begin{subfigure}{\gadgetWidth}
    \centering
    \includegraphics{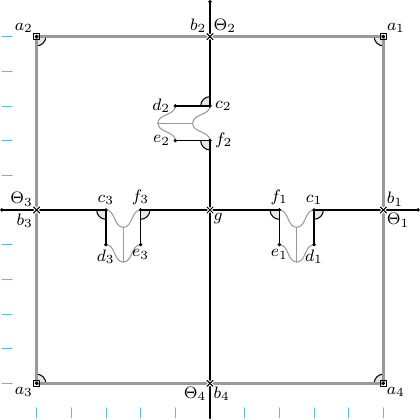}
    \caption{The Crossover Gadget forces $\theta_1=\theta_3$ and $\theta_2=\theta_4$.}
    \label{fig:crossover-gadget}
  \end{subfigure}
  \\
  \begin{subfigure}{\gadgetWidth}
    \centering
    \includegraphics{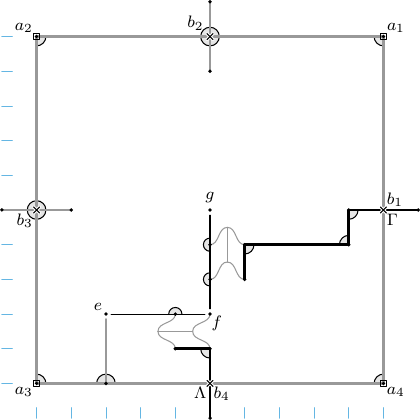}
    \caption{The Angular Gadget forces $g = (Q/2,Q/2)+R\cdot\Rect(\alpha,\beta)$.}
    \label{fig:angular-gadget}
  \end{subfigure}
  \hfill
  \begin{subfigure}{\gadgetWidth} %
    \centering
    \includegraphics{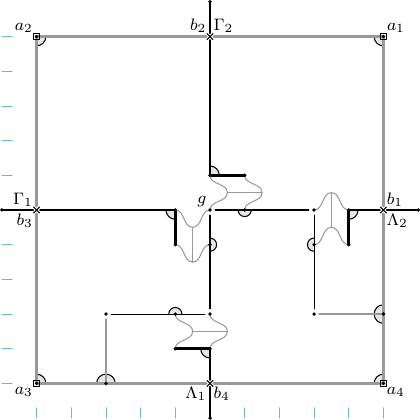}
    \caption{The Vector Rotation Gadget forces $ \Rect(\alpha_1,\beta_1)={i\cdot\Rect(\alpha_2,\beta_2)}$.}
    \label{fig:vector-rotate-gadget}
  \end{subfigure}
  \caption{
    A selection of extended linkage gadgets for manipulating angles and vectors. These scale-invariant gadgets are analyzed in Lemma~\ref{lem:scale-invariant-gadgets}. In these figures, $\Offset_{\Lambda}(C) = \alpha$, $\Offset_\Gamma(C) = \beta$, $\Offset_\Theta(C) = \theta$, and similarly with subscripts.
    As described in Convention~\ref{con:drawing-extended-linkages}, vertices surrounded by squares are pinned; those marked with an ``x'' are sliceform vertices; and corners marked with a solid gray sector are frozen, i.e., have tolerance $0$.
    Unfrozen sliceform vertices $b_j$ are assigned tolerance~$\delta$, while all remaining unfrozen corners have tolerance~$\eps$. The pins at the vertices $a_j$ are for clarification only; in the overall construction, these nodes $a_j$ are forbidden from moving by other means, so these explicit pins are unnecessary.
    The thick edges emphasize the rigidified background grid (in gray) and
    the transmission edges (in black), along with everything rigidly attached to those edges.
    Each cell has dimensions $Q \times Q$, and all marked vertices have coordinates at integer multiples of $Q/10$, indicated by cyan axis notches.}
  \label{fig:scale-invariant-gadgets}
\end{figure}

\begin{figure}[!b]
  \includegraphics{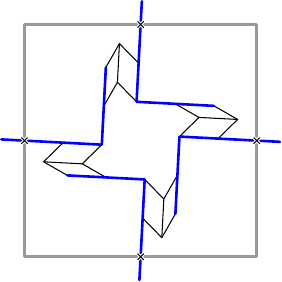}%
  \hfill%
  \includegraphics{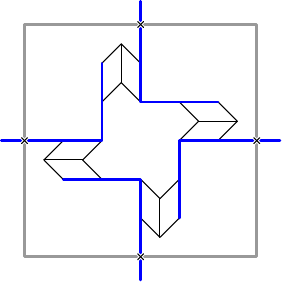}%
  \hfill%
  \includegraphics{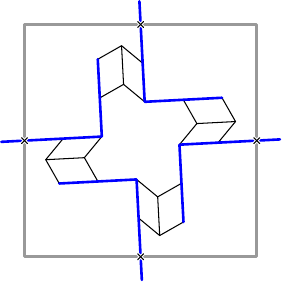}
  \caption{Snapshots of the configuration space of the Copy Gadget (cf.\ Figure~\ref{fig:copy-gadget}, Lemma~\ref{lem:scale-invariant-gadgets}), not drawn to scale. Angle offsets shown are $-3^\circ$, $0^\circ$ and $3^\circ$ respectively. The gadget has been slightly modified and simplified to emphasize movement.
  The thick blue edges emphasize the transmission edges.}
  \label{fig:copy-gadget-snapshots}
\end{figure}

\begin{lemma}[Scale Invariant Cell Gadgets]
  \label{lem:scale-invariant-gadgets}
  The abstract extended linkages described below are globally noncrossing with global minimum feature size at least $1/2$, and with their configuration spaces $U$ exactly parameterized as specified in each case.  Each gadget is described by an initial configuration where $a_3=(0,0)$ and $a_1=(Q,Q)$, all marked vertices have coordinates that are integer multiples of $Q/10$, and so all vertices have coordinates that are integer multiples of $Q/40$ (an integer).
  For some gadgets, we prove that the configuration space $U$ contains a conveniently shaped subset, identified below as a \term{useful neighborhood}.

  \begin{descriptionflush}
  \item[Copy Gadget.]
    The Copy Gadget $\cLcopy$ drawn in Figures~\ref{fig:copy-gadget} and~\ref{fig:copy-gadget-snapshots} constrains its four sliceform angles to remain equal. Specifically, $\Offset_{(\Theta_1,\Theta_2,\Theta_3,\Theta_4)}$ is a homeomorphism from $\Conf(\cLcopy)$ to $$\Ucopy = \{(\theta,\theta,\theta,\theta)\mid\theta\in[-\delta,\delta]\}.$$

  \item[Crossover Gadget.]
    The Crossover Gadget $\cLcross$ drawn in Figure~\ref{fig:crossover-gadget} constrains opposite sliceform angles to remain equal. Specifically, $\Offset_{(\Theta_1,\Theta_2,\Theta_3,\Theta_4)}$ is a homeomorphism from $\Conf(\cLcross)$ to
    $$
    \Ucross = \left\{
      (\theta_1,\theta_2,\theta_1,\theta_2)
      \;\middle|\;
      \theta_1,\theta_2 \in [-\delta,\delta]
    \right\}.
    $$

  \item[Angular Gadget.]
    The Angular Gadget $\cLangular$ drawn in Figure~\ref{fig:angular-gadget}
    relates the rectangular coordinates of vertex $g$ with the angles $\alpha,\beta$ in its angular representation. Specifically, the function $\pi_{g}\times \Offset_{\Lambda,\Gamma}$ is a homeomorphism from $\Conf(\cLangular)$ to
    \begin{equation*}
      \Uangular = \left\{((x,y),(\alpha,\beta))\mid \alpha,\beta\in[-\delta,\delta]\text{ and }(x,y) = (Q/2,Q/2)+R\cdot\Rect(\alpha,\beta)\right\}.
    \end{equation*}
    Useful Neighborhood: if $(x,y)$ is any point at most $R\delta/2$ away from the center of the cell, a unique configuration exists placing $g$ at $(x,y)$.

  \item[Vector Rotation Gadget.]
    The Vector Rotation Gadget $\cLvecrot$ (Figure~\ref{fig:vector-rotate-gadget}) enables constructing the $90^\circ$ rotation of a given vector.  Specifically, $\Offset_{(\Lambda_1,\Gamma_1,\Lambda_2,\Gamma_2)}$ is a homeomorphism from $\Conf(\cLvecrot)$ to
    \begin{equation*}
      \Uvecrot = \left\{
        (\alpha_1,\beta_1, \alpha_2,\beta_2)\in[-\delta,\delta]^4
        \;\middle|\;
        \Rect(\alpha_1,\beta_1) =
        i\cdot\Rect(\alpha_2,\beta_2)
      \right\}.
    \end{equation*}
    Useful Neighborhood: if $\alpha_1$ and $\beta_1$ are such that $|R\cdot\Rect(\alpha_1,\beta_1)| \le R\delta/2$, then a (necessarily unique) configuration $C$ exists with $\Offset_{(\Lambda_1,\Gamma_1)}(C) = (\alpha_1,\beta_1)$.

  \item[Angle Average Gadget.]
    The Angle Average Gadget $\cLangleavg$ drawn in Figures~\ref{fig:angle-average-gadget} and~\ref{fig:angle-average-gadget-snapshots} constrains one sliceform angle to equal the average of the other two. Specifically, $\Offset_{(\Theta_1,\Theta_2,\Theta_3)}$ is a homeomorphism from $\Conf(\cLangleavg)$ to
    $$\Uangleavg = \left\{
      (\theta_1,\theta_2,\theta_3)
      \;\middle|\;
      \theta_1,\theta_3 \in [-\delta,\delta]
      \text{ and } \theta_2 = (\theta_1+\theta_3)/2
    \right\}.
    $$

  \end{descriptionflush}
\end{lemma}

\newcommand\proofPar[1]{\noindent\emph{#1.}}
\newcommand\proofParNo[1]{\emph{#1.}}

\begin{proof}
  Each of these gadgets is \emph{scale invariant}: scaling down by a factor of $40/Q$ results in the same gadgets with the value of $Q$ replaced by $40$. Showing that each of these scaled down gadgets has global minimum feature size at least $1/2$ will show that the original gadgets have global minimum feature size at least $\frac{Q}{40}\cdot\frac{1}{2} \ge \frac{1}{2}$, so we may proceed as if $Q=40$ instead of being defined as in Section~\ref{sec:parameters}.
  
  \begin{descriptionflush}
      
  \item[Copy Gadget.]
    Consider first the Copy Gadget $\cLcopy$ of Figure~\ref{fig:copy-gadget}, whose movement is illustrated in Figure~\ref{fig:copy-gadget-snapshots}. The idea of the proof is as follows: $\eps$ enforces uniqueness and the claimed angle relationships, while $\delta$ ensures existence.

    \proofPar{Constraint and Uniqueness}
    Relying only on our choice of $\eps$, we first verify that angles $\theta_1,\theta_2,\theta_3,\theta_4$ must be equal, and that there is \emph{at most} one configuration having $\theta_1 = \theta_2 = \theta_3 = \theta_4 = \theta$ for each value $\theta$. Path $b_1 c_1 d_1 e_1 f_1$ must move as a rigid subassembly because of its frozen corners. So angle offset $\theta_1$ determines the locations of vertices $c_1,d_1,e_1,f_1$: they are simply rotated around $b_1$ from their initial positions by angle $\theta_1$. Positions for the rest of the marked vertices in Figure~\ref{fig:copy-gadget} are likewise determined by $\theta_2$, $\theta_3$, and $\theta_4$. By Lemma~\ref{lem:parallel-gadget}, the Parallel Gadgets force $d_jc_j \parallel e_{j+1}f_{j+1}$ for $1\le j\le 4$ (with indices taken cyclically) in any configuration of $\cLcopy$. It follows that $\theta_1 = \theta_2 = \theta_3 = \theta_4$ in any configuration; let $\theta$ be this common value. By the same lemma, the positions of $d_j$, $c_j$, and $e_{j+1}$ uniquely determine the positions of the other vertices within their Parallel Gadget, whenever such a configuration exists. So each configuration is indeed uniquely characterized by its single angle $\theta$. Said differently, the map $\phi = \Offset_{(\Theta_1,\Theta_2,\Theta_3,\Theta_4)}$ injectively maps $\Conf(\cLcopy)$ to~$\Ucopy$.

    \proofPar{Existence}
    On its own, the above paragraph does not prove that \emph{every} angle $\theta\in[-\delta,\delta]$ gives rise to a configuration of $\cLcopy$, or said differently, that $\phi$ is a bijection. Indeed, too large a $\theta$ would stretch or compress the Parallel Gadgets too far, resulting in broken angle constraints or worse. However, we next show that our choice of $\delta$ is small enough to guarantee $\phi$ is indeed a bijection. By Lemma~\ref{lem:rotating-bar-displacement}, because the path $b_j,c_j,d_j,e_j,f_j$ has length $\frac{6}{10} \cdot 40 = 24$, the rotations by $\theta$ around each $b_j$ move each vertex $c_j,d_j,e_j,f_j$ at most $24\delta < \eps/2^7$ from its starting position, and guarantee that $e_jc_j\parallel e_{j+1}f_{j+1}$ for each $j=1,2,3,4$. Then Lemma~\ref{lem:parallel-gadget} (Parallel Gadget) shows that each Parallel Gadget can configure itself at these new endpoints, making sure that each vertex moves by no more that $\eps/2$ from its starting point and each edge rotates by no more than $\eps/2$ from its original direction. This is enough to guarantee that all angle constraints are satisfied, so the configuration indeed exists. By Lemma~\ref{lem:feature-size}, this configuration is noncrossing with minimum feature size at least $1-2\cdot\eps/2 > 1/2$.    

  \item[Crossover Gadget.]
    The proof proceeds similarly to that of $\cLcopy$ above.

    \proofPar{Constraint and Uniqueness}
    In any configuration of $\cLcross$, the Parallel Gadgets and the sliceform vertex~$g$ enforce $b_3c_3 \parallel f_3g\parallel gf_1\parallel c_1b_1$ and likewise $b_4g\parallel gf_2\parallel c_2b_2$. This shows that $\theta_1=\theta_3$ and $\theta_2=\theta_4$ throughout $\Conf(\cLcross)$, and that these two angles uniquely determine the entire configuration (when such a configuration exists): each bold rigid assembly $b_jc_jd_j$ ($1\le j\le 3$) has rotated by angle $\theta_j$ around its (stationary) pivot $b_j$, edge $b_4 g$ has rotated by angle $\theta_4$ around $b_4$, each smaller rigid assembly $gf_je_j$ has likewise rotated by angle $\theta_j$ around the \emph{new} location of $g$, and by Lemma~\ref{lem:parallel-gadget} each Parallel Gadget can be drawn to connect its endpoints in at most one way.

    \proofPar{Existence}
    Consider a given pair of angles $\theta_1,\theta_2\in[-\delta,\delta]$.
    Because every marked vertex has a path of length at most $\frac{8}{10} \cdot 40 = 32$ to a transmission vertex $b_i$ (avoiding Parallel Gadgets),by Lemma~\ref{lem:rotating-bar-displacement}, the forced location of each marked vertex is at most $32\delta < \eps/2^7$ units away from its initial position. The opposite ends of each Parallel Gadget are indeed parallel, allowing the Parallel Gadgets to configure themselves by Lemma~\ref{lem:parallel-gadget}. As before, vertices have been displaced by at most $\eps/2$ and edges have been rotated by at most $\eps/2$, guaranteeing the configuration's validity and minimum feature size of at least $1-\eps \ge 1/2$. 

  \item[Angular Gadget.] \proofParNo{Constraint and Uniqueness}
    This is straightforward: the Parallel Gadgets force $ef$ and assembly $b_4c_4d_4$ to have rotated by $\alpha$, and likewise $fg$ and the assembly containing $b_1,c_1,d_1$ to have rotated by $\beta$ around the new position of $f$, so $g = (Q/2,Q/2) + R\cdot\Rect(\alpha,\beta)$, so everything is uniquely specified.

    \proofPar{Existence}
    This is also easy: because every marked vertex has a path of length at most $\frac{8}{10} \cdot 40 = 32$ to a vertex of the rigidified background grid (avoiding Parallel Gadgets), Lemma~\ref{lem:rotating-bar-displacement} shows that each marked vertex has moved at most $32\delta\le \eps/2^7$ from its initial position, so the Parallel Gadgets can be configured to connect their opposite edges, which can be checked to be parallel. As above, we may also conclude that each edge has rotated at most $\pm\eps/2$ from its initial direction. We conclude the validity, noncrossing, and minimum feature size of the resulting configuration as before.

    \proofPar{Useful Neighborhood}
    Suppose $(x,y)=g'$ is some point at distance at most $R\delta/2$ from the center of the cell. By Lemma~\ref{lem:rect-contains-box} with $\theta=\delta$, there is exactly one pair $(\alpha,\beta)\in[-\delta,\delta]^2$ such that $R\cdot\Rect(\alpha,\beta) = g'-(Q/2,Q/2)$, i.e., such that $(g',(\alpha,\beta))\in\Uangular$.
    
  \item[Vector Rotation Gadget.]
    \proofParNo{Constraint and Uniqueness}
    This is perfectly analogous to previous analyses---in fact, this gadget is essentially two conjoined Angular Gadgets---so we have omitted this analysis.

    \proofPar{Existence}
    This is also analogous and therefore omitted.

    \proofPar{Useful Neighborhood}
    Vertex $g$ must land in the region $(Q/2,Q/2) + R\cdot\Rect([-\delta,\delta]^2)$ (because of $\Lambda_1,\Gamma_1$) as well as in this same region rotated by $\pi/2$ (because of $\Lambda_2,\Gamma_2$). These two regions are not identical, but by Lemma~\ref{lem:rect-contains-box} with $\theta = \delta$, they both contain the square $(Q/2,Q/2) + [-R\delta/2,R\delta/2]^2$. Any point within this square therefore gives rise to a valid configuration of $\cLvecrot$.

  \item[Angle Average Gadget.]
    This analysis is more involved than the previous ones, but it follows the same outline. The gadget $\cLangleavg$ is depicted in Figure~\ref{fig:angle-average-gadget} (drawn to scale). We choose $f_1$, $f_3$, and $e_1$ such that $f_3e_3he_2$ and $he_1f_1e_2$ are kites which are initially similar. We will show that these kites remain similar throughout $\Conf(\cLangleavg)$, which is pivotal (so to speak) to how the gadget works. This gadget is based on the \emph{Reflector Gadget} from~\cite{popups}, but as a proof of the Reflector Gadget's movement was not included in that extended abstract, we provide a self-contained proof of our Angle Average Gadget here. The motion of this gadget has been illustrated in Figure~\ref{fig:angle-average-gadget-snapshots}, which is \emph{not} drawn to scale to emphasize the movements.

    \proofPar{Constraint and Uniqueness}
    We first argue that angles $\theta_1,\theta_2,\theta_3$ uniquely determine their configuration, when one exists. Arguments like those above show that all marked vertices \emph{other} than $e_2$, $k_1$, and $k_2$ depend straightforwardly on angles $\theta_1,\theta_2,\theta_3$, by rotations and translations. Treating $he_2$ as a single fixed-length bar, Lemma~\ref{lem:hook-linkage} (Hook Linkage) on $he_2f_3$ (which has angle $\cos\inv(-{1}/{\sqrt{5}}) < 2\pi/3$) shows that $e_2$ has at most one valid position, and the same lemma shows that $k_1$ and $k_2$ are likewise uniquely determined.

    We next argue that, for a configuration to exist, $\theta_2 = (\theta_1+\theta_3)/2$ must hold. Focus for now on kites $he_1f_1e_2$ and $f_3e_3he_2$. Define $\alpha_1 = \angle e_1he_2$ and $\alpha_2 = \angle f_1e_1h$, and similarly, $\beta_1 = \angle e_3f_3e_2$ and $\beta_2 = h e_3f_3$. Considering path $m_1e_1he_2f_3m_3$, segments $m_1e_1$ and $f_3m_3$ being parallel (which is guaranteed by the two upper Parallel Gadgets) implies that
    \begin{equation*}
      \alpha_2 + \alpha_1 + (2\pi-\beta_2) + (\pi-\beta_1) = 3\pi,
    \end{equation*}
    i.e., $\alpha_1+\alpha_2 = \beta_1+\beta_2$. We will show that this forces the kites to be similar.

    Kite $f_3e_3he_2$ has edge lengths $4$ and $4\sqrt{5}$. We may compute that $|e_2-e_3| = 8\sin\frac{\beta_1}{2}$, $\angle e_2e_3f_3 = (\pi-\beta_1)/2$, and
    \begin{equation*}
      \angle he_3e_2 = \cos\inv\frac{|e_2-e_3|/2}{4\sqrt{5}}
      = \cos\inv\frac{t}{\sqrt{5}},
    \end{equation*}
    where $t = \sin\frac{\beta_1}{2}$. So, in terms of $t$,
    \begin{equation*}
      \beta_1 + \beta_2 =
      \frac{\pi}{2} + \frac{1}{2}\beta_1 + \angle e_2e_3h
      = \frac{\pi}{2} + \sin\inv t + \cos\inv\frac{t}{\sqrt{5}}.
    \end{equation*}
    Analogously, if $s = \sin\frac{\alpha_1}{2}$, we have
    \begin{equation*}
      \alpha_1 + \alpha_2 =
      \frac{\pi}{2} + \sin\inv s + \cos\inv\frac{s}{\sqrt{5}}.
    \end{equation*}
    But the function $f(x) := \sin\inv x + \cos\inv \frac{x}{\sqrt{5}}$ is strictly monotonically increasing on $x\in[0,1]$, because its derivative is $(1-x^2)^{-1/2} - (5-x^2)^{-1/2} > 0$. So $f(s) = f(t)$ implies $s = t$, meaning $\alpha_1=\beta_1$ and $\alpha_2 = \beta_2$.

    Revisiting Figure~\ref{fig:angle-average-gadget}, let $\kappa = \cos\inv\frac{-1}{\sqrt{5}}$ be the initial value of angles $\alpha_2 = \angle f_1e_1h$ and $\beta_2 = \angle h e_3f_3$. The new values for these angles are $\kappa + \theta_1 - \theta_2$ and $\kappa + \theta_2 - \theta_3$ respectively. But these values must be equal, so indeed $\theta_2 = (\theta_1+\theta_3)/2$.

    \proofPar{Existence}    
    Conversely, we must show that if angles $\theta_j\in[-\delta,\delta]$ are chosen from the set $\Uangleavg$ (i.e., satisfying $\theta_2 = (\theta_1+\theta_3)/2$), then a configuration indeed exists. Use the $\theta_j$ to configure all marked vertices other than $e_2$, $k_1$, and $k_2$ by rotations as above.
    Because every marked vertex has a path of length at most $\frac{11}{10} \cdot 40 = 44$ to a transmission vertex $b_i$ (avoiding Parallel Gadgets),by Lemma~\ref{lem:rotating-bar-displacement}, each marked vertex is displaced by at most $44\delta < 2^6\delta$.

    Let $p$ be the reflection of $e_1$ across $hf_1$, and $q$ the reflection of $e_3$ across $hf_3$; we will verify that $p=q$, so that $e_2$ has a well-defined position at the correct distances from $h$, $f_1$, and $f_2$ simultaneously. Equality $\theta_2 = (\theta_1+\theta_3)/2$ implies that $\alpha_2 = \beta_2$ and therefore kites $he_1f_1p$ and $f_3e_3hq$ are similar. Label their angles as $\alpha_1' := \angle e_1hp = \angle e_3f_3q$ and $\alpha_3' := \angle pf_1e_1 = \angle qhe_3$. We have $2\pi = \angle f_1e_1h + \angle e_1he_3 + \angle h e_3f_3 = \angle e_1he_3 + 2\alpha_2$ because $f_1e_1$ and $e_3f_3$ are parallel, so we may compute
    \begin{align*}
      \angle phq
      &= \angle e_1he_3 - \angle e_1hp - \angle qhe_3\\
      &= (2\pi - 2\alpha_2) - \alpha_1' - \alpha_3'\\
      &= 0,
    \end{align*}
    meaning $p = q$, as claimed. 

    Each of $h$ and $f_3$ has been displaced by at most $2^6\delta$ as above, which is less than $\frac{1}{32}\min(|h-e_2|,|e_2-f_3|) = 1/8$, so the Hook Linkage Lemma shows that $e_2$'s forced location of $p=q$ is at most $2^8\delta$ from its initial position and that $he_2$ and $e_2f_3$ have each rotated by at most $2^6\delta$ from their original directions (since $|h-e_2|$ and $|e_2-f_3|$ are at least $4$). Triangle $hk_2e_2$ must be rotated by this same amount to match its hypotenuse, so $k_2$ is displaced by at most $|h-k_2|\cdot 2^6\delta = 2^9\delta$.

    Similarly, the Hook Linkage Lemma applied to $he_2f_1$ shows that edge $e_2 f_1$ has been rotated by at most $2^6\delta$, so the same is true for edges $k_1 f_1$ and $k_1 e_2$. So $k_1$'s total displacement may be bounded by $e_2$'s displacement plus $|e_2-k_1|\cdot 2^6\delta$, namely $(2^{8}+2^{10})\delta\le 2^{11}\delta$.

    We have thus displaced all marked vertices no farther than $2^{11}\delta \le \eps/2^7$ and rotated each edge (other than those in Parallel Gadgets) by no more than $2^6\delta < \eps/2$. We may thus configure the Parallel Gadgets and then verify angle constraints and minimum feature size exactly as for prior gadgets, which completes the analysis.\qedhere

  \end{descriptionflush}
  
\end{proof}

\begin{figure}[hbt]
  \centering
  \includegraphics{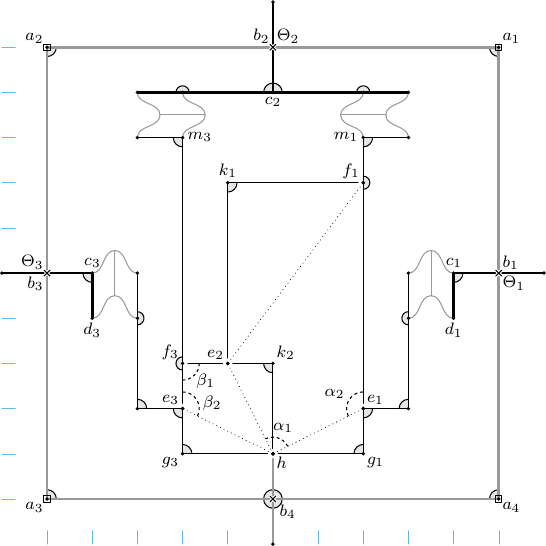}
  \caption{The Angle Average Gadget forces $\theta_2=(\theta_1+\theta_3)/2$.
  This gadget is analyzed in Lemma~\ref{lem:scale-invariant-gadgets}. See this gadget's movement in Figure~\ref{fig:angle-average-gadget-snapshots}.
  The drawing follows Convention~\ref{con:drawing-extended-linkages}.
  Unfrozen sliceform vertices $b_j$ are assigned tolerance~$\delta$, while all remaining unfrozen corners have tolerance~$\eps$.
  The pins at the vertices $a_j$ are for clarification only; in the overall construction, these nodes $a_j$ are forbidden from moving by other means, so these explicit pins are unnecessary.
  The thick edges emphasize the rigidified background grid (in gray) and
  the transmission edges (in black), along with everything rigidly attached to those edges.
  Dotted segments and dashed angles are referred to in the proof.
  The cell has dimensions $Q \times Q$, and all marked vertices have coordinates at integer multiples of $Q/10$, indicated by cyan axis notches.
  }
  \label{fig:angle-average-gadget}
\end{figure}

\begin{figure}
  \includegraphics{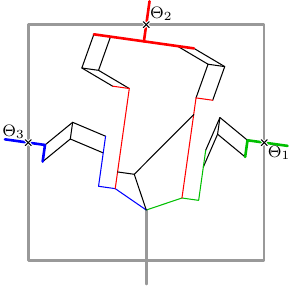}%
  \hfill%
  \includegraphics{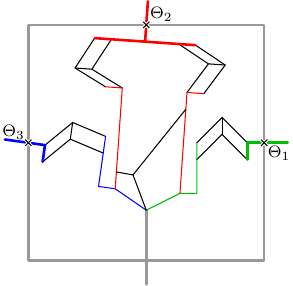}%
  \hfill%
  \includegraphics{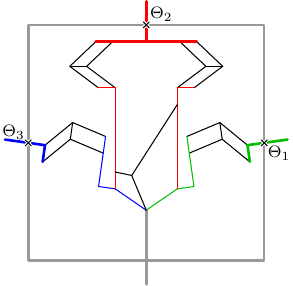}
  \\[.1in]
  \includegraphics{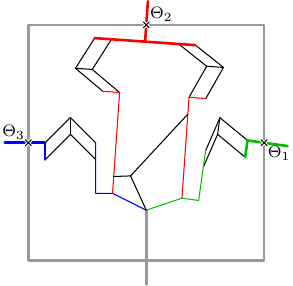}%
  \hfill%
  \includegraphics{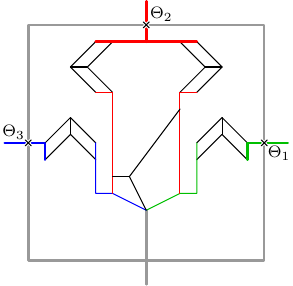}%
  \hfill%
  \includegraphics{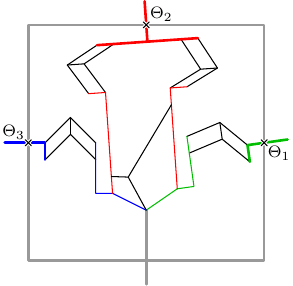}
  \\[.1in]
  \bigskip
  \includegraphics{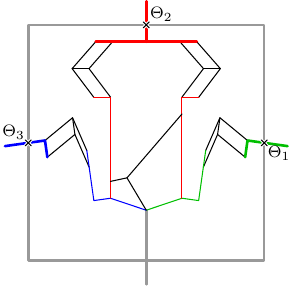}%
  \hfill%
  \includegraphics{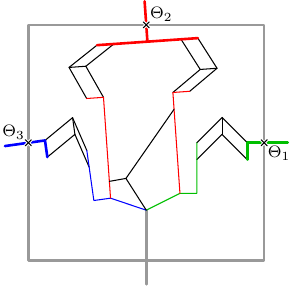}%
  \hfill%
  \includegraphics{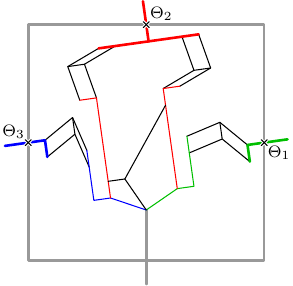}%
  \caption{Snapshots of the Angle Average Gadget's configuration space (cf.\ Lemma~\ref{lem:scale-invariant-gadgets}, Figure~\ref{fig:angle-average-gadget}), shown for each pair of values $\theta_1,\theta_3\in\{-8^\circ, 0^\circ, 8^\circ\}$ (where $\theta_j = \Offset_{\Theta_j}$). In each configuration, it may be observed that $\theta_2 = (\theta_1+\theta_3)/2$. Edge lengths have been altered from those in Figure~\ref{fig:angle-average-gadget} to exaggerate the gadget's movement.
  The thick colored edges emphasize the transmission edges, and non-black thin edges are colored to indicate which transmission edge controls their angle.}
  \label{fig:angle-average-gadget-snapshots}
\end{figure}

The next batch of cell gadgets are not scale invariant: they may have minimum feature size less than $1/2$ when scaled down by $40/Q$, unlike the previous gadgets, so we cannot assume $Q=40$ as above. We must therefore skip the simplifying scaling step, work instead with the value of $Q$ given in Section~\ref{sec:parameters}, and be more careful when arguing about feature size.

\begin{lemma}[Cell Gadgets with Unscaled Measurements]
  \label{lem:unscaled-gadgets}
  The abstract extended linkages described below are globally noncrossing with global minimum feature size at least $1/2$, and their configuration spaces are exactly parameterized as described.  Each gadget is described by an initial configuration having $a_3=(0,0)$ and $a_1=(Q,Q)$, and where all marked vertices have coordinates that are integer multiples of $Q/10$, \emph{with a few exceptions} listed below.
  
  \begin{descriptionflush}
  \item[Start Gadget.]
    The Start Gadget $\cLstart$ (Figure~\ref{fig:start-gadget})
    relates rectangular and angular coordinates just like the Angular Gadget, but this time the angular coordinates use radius $2r$ instead of $R$. Recall $r = \lceil d/\delta\rceil$ is half the radius used in the angular change of coordinates $(x,y) = 2r\cdot\Rect(\alpha,\beta)$ applied to each polynomial and that $r \leq Q/20$. The function $\pi_v\times\Offset_{\Lambda,\Gamma}$ is a homeomorphism from $\Conf(\cLstart)$ to
    \begin{equation*}
      \Ustart = \left\{((x,y),(\alpha,\beta)) \mid \alpha,\beta\in[-\delta,\delta]\text{ and } (x,y) = (Q/5+2r,Q/5+2r) + 2r\Rect(\alpha,\beta)\right\}.
    \end{equation*}
    Exceptional vertices are initially configured at points $u = (Q/5,Q/5) + (2r,0)$, $v = (Q/5,Q/5) + (2r,2r)$, and $w = (Q/2,Q/5) + (0,2r)$, which all have integer coordinates.

  \item[Vector Creation Gadget.]
    For a positive number $1\le w\le R\delta/2$ (not necessarily an integer), the Vector Creation Gadget $\cLveccreate(w)$ (Figure~\ref{fig:vector-create-gadget}) enables constructing angular coordinates for the vector with complex representation $w\cdot(e^{i\theta}-1)$ for all $\theta \in [-\delta,\delta]$. Specifically, $\Offset_{(\Theta,\Lambda,\Gamma)}$ provides a homeomorphism from $\Conf(\cLveccreate(c))$ to
    \begin{equation*}
      \Uveccreate = \left\{
        (\theta,\alpha,\beta)\in[-\delta,\delta]^3
        \;\middle|\;
        w\cdot (\cos\theta,\sin\theta) - (w,0) =
        R\cdot \Rect(\alpha,\beta)
      \right\}.
    \end{equation*}
    In the initial configuration shown in Figure~\ref{fig:vector-create-gadget}, the only exceptional vertex has position $c_3=(Q/2-w,Q/2)$, which has integer coordinates when $w$ itself is an integer.

  \item[End Gadget.]
    For any chosen real number $w$ with magnitude at most $R\delta/2$,
    the End Gadget $\cLend = \cLend(w)$ (Figure~\ref{fig:end-gadget}) constrains a vector to the position $(Q/2+w,Q/2)$. Specifically, $\Conf(\cL)$ has exactly one configuration, and this configuration has $R \cdot \Rect(\alpha,\beta) = (w,0)$. When $w=0$, there are no exceptional vertices. (When $w\ne 0$, we make no promises about vertices having integer coordinates other than those frozen to the background grid.)

  \end{descriptionflush}

\end{lemma}

\begin{figure}[phbt]
  \begin{subfigure}[t]{\gadgetWidth}
    \centering
    \includegraphics{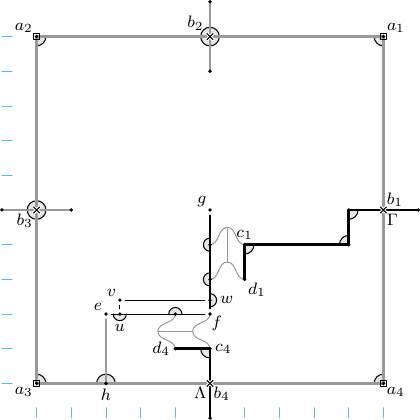}
    \caption{The Start Gadget forces $v=2r\cdot\Rect(\alpha,\beta)$.}
    \label{fig:start-gadget}
  \end{subfigure}
  \hfill
  \begin{subfigure}[t]{\gadgetWidth}
    \centering
    \includegraphics{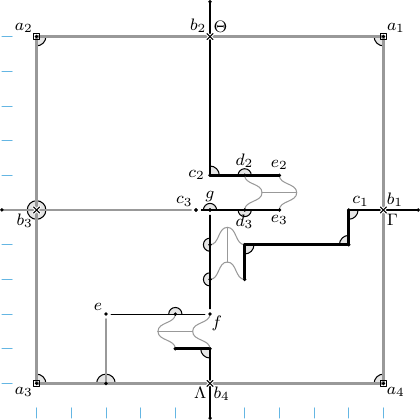}
    \caption{The Vector Creation Gadget forces $R\cdot\Rect(\alpha,\beta) = {w\cdot (e^{i\theta}-1)}$, where $\ell(c_3 g) = w$.}
    \label{fig:vector-create-gadget}
  \end{subfigure}
  \\
  \hfill
  \begin{subfigure}{\textwidth} %
    \centering
    \includegraphics{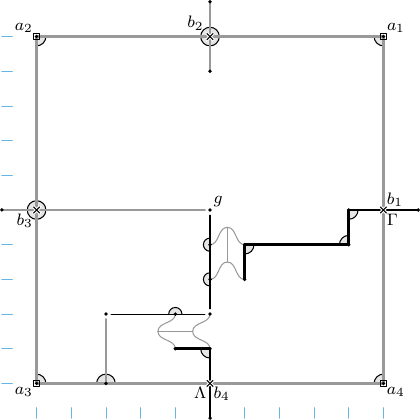}
    \caption{The End Gadget forces $R\cdot\Rect(\alpha,\beta) = (w,0)$.
      The edge length $\ell(b_3 g)$ is $Q/2+w$.}
    \label{fig:end-gadget}
  \end{subfigure}
  \hfill
  \caption{
    A selection of extended linkage gadgets for manipulating angles and vectors. These gadgets do not scale uniformly with $Q$, and they are analyzed in Lemma~\ref{lem:unscaled-gadgets}.
    In these figures, $\Offset_{\Lambda}(C) = \alpha$ and $\Offset_\Gamma(C) = \beta$.
    As described in Convention~\ref{con:drawing-extended-linkages}, vertices surrounded by squares are pinned; those marked with an ``x'' are sliceform vertices; and corners marked with a solid gray sector are frozen, i.e., have tolerance $0$.
    However, unfrozen sliceform vertices $b_j$ are assigned tolerance~$\delta$, while all remaining corners have tolerance~$\eps$. The pins at the vertices $a_j$ are for clarification only; in the overall construction, these nodes $a_j$ are forbidden from moving by other means, so these explicit pins are unnecessary.
    Each cell has dimensions $Q \times Q$, and all vertices except those specified in Lemma~\ref{fig:unscaled-gadgets} have coordinates at integer multiples of $Q/10$, indicated by cyan axis notches.}
  \label{fig:unscaled-gadgets}
\end{figure}

\begin{proof}
  We consider each gadget in turn.
  \begin{descriptionflush}
      
  \item[Start Gadget.]
    This gadget is simply the Angular Gadget from Lemma~\ref{lem:scale-invariant-gadgets} amended with three vertices $u,v,w$ forming a parallelogram with $f$.

    \proofPar{Constraint and Uniqueness} Vertices $u$ and $w$ must move rigidly with assemblies $ef$ and $fg$, so their positions are determined by rotations by $\alpha$ and $\beta$. By Lemma~\ref{lem:parallelogram-linkage}, vertex $v$ must complete the parallelogram, so its position is likewise uniquely determined as $e + 2r\exp(i\alpha) + 2r\cdot i\cdot \exp(i\beta) = e+(2r,2r)+2r\cdot \Rect(\alpha,\beta)$.

    \proofPar{Existence} As in the Angular Gadget, each bar is rotated by at most $\pm\eps/2$, so all angle constraints are satisfied. It remains to show that the linkage is noncrossing with minimum feature size at least $1/2$. The only differences from the Angular Gadget are vertices $u,v,w$ and edges $u v, v w$, so if a crossing or smaller feature size occurs, it must involve one of these five items. Because $2r \le R/2$, the smallest feature introduced by this parallelogram is at least
    \begin{equation*}
      |u-v|\cdot\left|\sin\angle vue\right|
      \ge 2r\cdot\sin\left(\frac{\pi}{2} - \eps\right)
      \ge r > \frac{1}{2},
    \end{equation*}
    as required.

  \item[Vector Creation Gadget.]
    \proofParNo{Constraint and Uniqueness}
    Rigid assemblies $b_2c_2d_2e_2$ and $c_3gd_3e_3$ must have been rotated by angle $\theta$, meaning $g = c_3 + w\cdot(\cos\theta,\sin\theta)$. The rest is an instance of the Angular Gadget, which uniquely positions all other marked vertices and locates $g$ at $(Q/2,Q/2) + R\cdot\Rect(\alpha,\beta)$. These two positions of $g$ must agree, proving that
    \begin{equation*}
      R\cdot\Rect(\alpha,\beta) = w\cdot(\cos\theta,\sin\theta) - w\cdot(1,0),
    \end{equation*}
    as required.

    \proofPar{Existence}
    Each marked vertex is determined by a rigidified path (from a vertex rigidly attached to the rigidifed background grid) of total length at most $\frac{6}{10} Q < Q$ where each edge rotates by at most $\pm\delta$, so each marked vertex is displaced by at most $\frac{6}{10} Q \delta < \eps/2^7 \cdot Q/40$ by Lemma~\ref{lem:rotating-bar-displacement}. Then Lemma~\ref{lem:parallel-gadget} (Parallel Gadget) shows that each Parallel Gadget (scaled up by $Q/40$) can connect its opposite edges such that each vertex moves by at most $\eps/2\cdot Q/40$ and each edge rotates by at most $\pm\eps/2$. Any feature not involving $c_3$ has length at least $\frac{1}{2}Q/40$ by Lemma~\ref{lem:feature-size}, and $c_3$'s smallest feature is the distance from $c_3$ to $fg$, which is at least $|c_3-g|\cdot \sin(\pi/2-\eps) \ge w/2 \ge 1/2$.
    
  \item[End Gadget.]
    This is an Angular Gadget with one additional edge keeping $g$ stationary at $g=(Q/2+w,Q)$. Because this chosen point is at most $R\delta/2$ from the cell's center, the Angular Gadget case of Lemma~\ref{lem:scale-invariant-gadgets} guarantees that the required configuration exists and is unique. When $w=0$, this unique configuration agrees with the Angular Gadget`s initial configuration.\qedhere
    
  \end{descriptionflush}
  
\end{proof}

Finally, we demonstrate a few \emph{compound} gadgets, each inhabiting a constant number of grid cells instead of just one. As described in Section~\ref{sec:detailed-overview}, the background grid is pinned in place and consists of $Q\times Q$ cells with \term{transmission edges} incident to \term{transmission sliceform vertices} at the midpoints of the sides of each cell. Neighboring cells communicate through these sliceform vertices, and similarly, gadgets that are \emph{not} intended to interact must be placed in non-adjacent cells so as to not share a transmission vertex. Each transmission vertex will either be adjacent to one or more gadgets (meaning the angle at its transmission edges is determined by those gadgets) or it will be frozen, so there is no possibility for movement away from the gadgets. In the case of the Vector Average Gadget, a multi-cell cavity is carved into the background grid to make way for a larger assembly: cell edges that are not drawn are indeed missing from the grid, but all remaining cell edges have a transmission vertex and transmission edges at their midpoint, as usual.

\begin{figure}[phbt]
  \begin{subfigure}{\textwidth}
    \centering
    \includegraphics{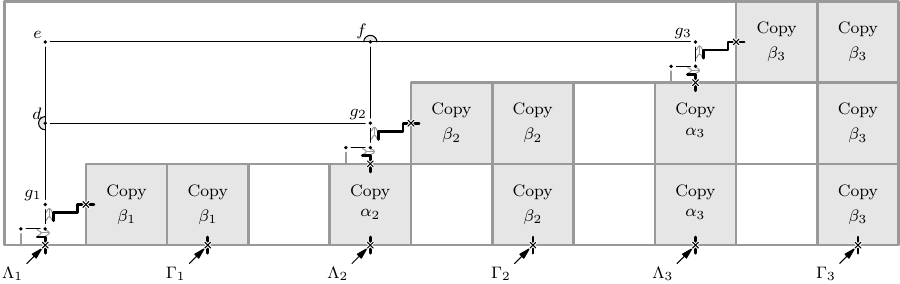}
    \caption{
      The Vector Average Gadget forces $g_2 = (g_1+g_3)/2$, i.e., $\Rect(\alpha_2,\beta_2) = (\Rect(\alpha_1,\beta_1)+\Rect(\alpha_3,\beta_3))/2$.
    }
    \label{fig:vector-average-gadget}
  \end{subfigure}
  \\
  \phantom{abc}\\
  \begin{subfigure}[t]{\gadgetWidth}
    \centering
    \includegraphics{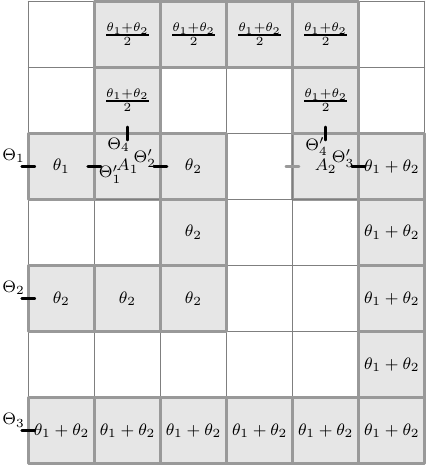}
    \caption{The Angle Sum Gadget forces $\theta_3 = \theta_1+\theta_2$.}
    \label{fig:angle-sum-gadget}
  \end{subfigure}
  \hfill
  \begin{subfigure}[t]{\gadgetWidth}
    \centering
    \includegraphics{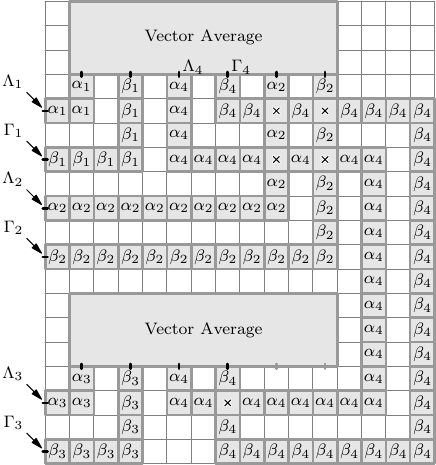}
    \caption{The Vector Sum Gadget forces $\Rect(\alpha_2,\beta_2) = \Rect(\alpha_1,\beta_1)+\Rect(\alpha_3,\beta_3)$}
    \label{fig:vector-sum-gadget}
  \end{subfigure}
  \caption{
    Compound gadgets for manipulating vectors and angles, built from earlier, more primitive gadgets. In these figures, $\Offset_{\Lambda}(C) = \alpha$, $\Offset_\Gamma(C) = \beta$, $\Offset_\Theta(C)=\theta$, and similarly with indices. Cells labelled with an angle are Copy Gadgets, those marked with an ``x'' are Crossover Gadgets, and cells labelled $A_j$ are Angle Average Gadgets. All grid edges (where present) have a transmission sliceform and transmission edges, but only important ones are shown here.
  }
  \label{fig:compound-gadgets}
\end{figure}

\begin{lemma}[Compound Gadgets]
  \label{lem:compound-gadgets}
  The abstract extended linkages described below are globally noncrossing with global minimum feature size at least $1/2$, and with their configuration spaces $U$ exactly parameterized as specified in each case.  Each gadget is described by an initial configuration where the lower-left corner is configured at $(0,0)$, all marked vertices have coordinates that are integer multiples of $Q/10$, and so all vertices have coordinates that are integer multiples of $Q/40$ (an integer).
  As in Lemma~\ref{lem:scale-invariant-gadgets},
  for some gadgets, we prove that the configuration space $U$ contains a conveniently shaped subset, identified below as a \term{useful neighborhood}.
  
  \begin{descriptionflush}

  \item[Vector Average Gadget.]
    The Vector Average Gadget $\cLvecavg$ (Figure~\ref{fig:vector-average-gadget}) constrains one vector to be the average of two others. Specifically, the function $\Offset_{(\Lambda_1,\Gamma_1,\Lambda_2,\Gamma_2,\Lambda_3,\Gamma_3)}$ is a homeomorphism from $\Conf(\cLvecavg)$ to
    \begin{equation*}
      \Uvecavg = \left\{
        (\alpha_1,\beta_1, \alpha_2,\beta_2,\alpha_3,\beta_3)\in[-\delta,\delta]^6
        \;\middle|\;
        2\Rect(\alpha_2,\beta_2) =
        \Rect(\alpha_1,\beta_1) + 
        \Rect(\alpha_3,\beta_3)
      \right\}.
    \end{equation*}
    Useful Neighborhood: whenever three vectors satisfying $v_1+v_3=2v_2$ are given with magnitudes at most $|v_j|\le \delta/2$ for $j=1,2,3$, there exists a (necessarily unique) configuration of $\cLvecsum$ with $\Rect(\alpha_j,\beta_j) = v_j$ for $j=1,2,3$.

  \item[Angle Sum Gadget.]
    The Angle Sum Gadget $\cLanglesum$ drawn in Figure~\ref{fig:angle-sum-gadget} constrains one sliceform angle to equal the sum of the other two: $\Offset_{(\Lambda_1,\Lambda_2,\Lambda_3)}$ is a homeomorphism from $\Conf(\cLanglesum)$ to
    \begin{equation*}
      \Uanglesum = 
      \left\{
        (\theta_1,\theta_2,\theta_3)
        \;\middle|\;
        \theta_1,\theta_2,\theta_3 \in [-\delta,\delta]
        \text{ and } \theta_3 = \theta_1 + \theta_2
      \right\}.
    \end{equation*}
    
  \item[Vector Sum Gadget.]
    The Vector Sum Gadget $\cLvecsum$ (Figure~\ref{fig:vector-sum-gadget}) constrains one vector to be the sum of two others. Specifically, $\Offset_{(\Lambda_1,\Gamma_1,\Lambda_2,\Gamma_2,\Lambda_3,\Gamma_3)}$ is a homeomorphism from $\Conf(\cLvecsum)$ to
    \begin{equation*}
      \Uvecsum = \left\{
        (\alpha_1,\beta_1, \alpha_2,\beta_2,\alpha_3,\beta_3)\in[-\delta,\delta]^6
        \;\middle|\;
        \Rect(\alpha_3,\beta_3) =
        \Rect(\alpha_1,\beta_1) + 
        \Rect(\alpha_2,\beta_2)
      \right\}.
    \end{equation*}
    Useful Neighborhood: whenever three vectors satisfying $v_1+v_2=v_3$ are given with magnitudes at most $|v_j|\le \delta/2$ for $j=1,2,3$, there exists a (necessarily unique) configuration of $\cLvecsum$ with $\Rect(\alpha_j,\beta_j) = v_j$ for $j=1,2,3$.

  \end{descriptionflush}
\end{lemma}

\begin{note}
  The Vector Average Gadget uses the well-known \emph{pantograph linkage} to keep $g_2$ at the midpoint of $g_1$ and $g_3$.
\end{note}

\begin{proof} We consider each gadget in turn.
  \begin{descriptionflush}

  \item[Vector Average Gadget.]
    The Vector Average Gadget, $\cLvecavg$, illustrated in Figure~\ref{fig:vector-average-gadget}, is formed within a $3\times 11$ rectangle of grid cells, where some grid edges have been removed as shown. Three Angular Gadgets (missing their top and left cell edges) are placed with lower-left corners at $(0,0)$, $(4Q,Q)$, and $(8Q,2Q)$, and their vertices $g_1,g_2,g_3$ are connected to parallelogram $defg_2$ as shown. The angles $\alpha_j$ and $\beta_j$ ($j=1,2,3$) at the three Angular Gadgets are duplicated by Copy Gadgets and output along the bottom edge. Observe that this gadget is scale invariant, like the gadgets from Lemma~\ref{lem:scale-invariant-gadgets}, so we may assume $Q=40$ in our analysis.

    \proofPar{Constraint and Uniqueness}
    Temporarily ignoring $d$, $e$, $f$, and their edges, the Angular and Copy Gadgets ensure that the six angles $\alpha_j$ and $\beta_j$ are each free to move within $[-\delta,\delta]$, and together they uniquely determine the positions of all vertices (other than $d$, $e$, and $f$). Lemma~\ref{lem:hook-linkage} (Hook Linkage) shows that the positions of $d$, $e$, and $f$ are determined by $g_1,g_2,g_3$. Furthermore, parallelogram $defg_2$ (which remains a parallelogram by Lemma~\ref{lem:parallelogram-linkage}) ensures that triangle $g_1dg_2$ remains similar to $g_1eg_3$, so $g_2$ must lie at the midpoint of $g_1g_3$. This is equivalent to $\Rect(\alpha_2,\beta_2)$ being the average of $\Rect(\alpha_1,\beta_1)$ and $\Rect(\alpha_3,\beta_3)$, as required. 

    \proofPar{Existence}
    Suppose angles $\alpha_j$ and $\beta_j$ are chosen for $j=1,2,3$ satisfying this average constraint, i.e., lying in $\Uvecavg$. As in Lemma~\ref{lem:scale-invariant-gadgets} (the Copy Gadget and Angular Gadget cases), all vertices other than $d,e,f$ may be configured at most $\eps/2$ from their initial positions such that all edges other than those incident to $d,e,f$ have rotated by at most $\pm\eps/2$ from their initial directions. Membership in $\Uvecavg$ guarantees that $g_2$ is the midpoint of $g_1g_3$. Lemma~\ref{lem:hook-linkage} (Hook Linkage) applied to $g_1dg_2$, which is applicable because $\max((\eps/2)/|g_1-d|,(\eps/2)/|d-g_2|) = \eps/80 < 1/32$, shows that $d$ may be configured at most $4\cdot\eps/2 = 2\eps$ from its initial position such that edges $g_1 d$ and $d g_2$ have each rotated by at most $\max\left((\eps/2)/|g_1-d|,(\eps/2)/|d-g_2|\right)\le\eps/20 < \eps/2$ from their initial directions. Similar computations apply to hooks $h_2fg_3$ and $g_1eg_3$, and because $g_2$ is the midpoint of $g_1g_3$, similar triangles ensure that $d$ and $f$ land at the midpoints of $g_1e$ and $g_3e$, respectively. The result is thus a valid configuration of $\cLvecavg$ with minimum feature size at least $1-4\eps \ge 1/2$.

    \proofPar{Useful Neighborhood}
    This final claim follows from the Angular Gadget case of Lemma~\ref{lem:scale-invariant-gadgets}.
    
  \item[Angle Sum Gadget.]
    The Angle Sum Gadget in Figure~\ref{fig:angle-sum-gadget} is constructed in a $6\times 7$ grid of $Q\times Q$ cells, with transmission sliceforms and transmission edges as described immediately preceding this Lemma (most not depicted). The cells labeled $A_1$ and $A_2$ are Angle Average Gadgets ($\cLangleavg$), where one of $A_2$'s transmission sliceforms (drawn in gray) has been frozen, fixing that input to~$0$. The other gray cells are Copy Gadgets, $\cLcopy$. This analysis follows straightforwardly from the characterizations of the Angle Average and Copy Gadgets from Lemma~\ref{lem:scale-invariant-gadgets}, as we now show. Define $\theta_j' = \theta_j'(C) = \Offset_{\Theta_j'}(C)$ for a given configuration $C$, analogous to the definitions of $\theta_j$.

    \proofPar{Constraint and Uniqueness}
    Suppose $\theta_1,\theta_2,\theta_3\in[-\delta,\delta]$ are given. We must show that they give rise to at most one configuration of $\cLanglesum$, and that $\theta_3 = \theta_1+\theta_2$ must hold for a configuration to exist. The Copy Gadgets labeled $\theta_1$ and $\theta_2$ show that $\theta_1' = \theta_1$ and $\theta_2' = \theta_2$. Then gadget $A_1$ forces $\theta_4 = (\theta_1+\theta_2)/2$. Copy Gadgets transfer this same value to $\theta_4'=\theta_4=(\theta_1+\theta_2)/2$, and then gadget $A_2$ and the remaining Copy Gadgets ensure that $\theta_4'=(\theta_1+\theta_2)/2$ is the average of $0$ and $\theta_3$, or in other words, $\theta_3 = \theta_1+\theta_2$. All copy and Angle Average Gadget configurations are uniquely determined by these values. Because transmission edges not part of a gadget on either side are frozen, angles $\theta_1,\theta_2,\theta_3$ uniquely determine the entire configuration of $\cLanglesum$ (when it exists), as required.

    \proofPar{Existence}
    The conditions $\theta_j\in[-\delta,\delta]$ for $j=1,2,3$ show that $\theta_4 = (\theta_1+\theta_2)/2$ also lies in $[-\delta,\delta]$, and the condition $\theta_3=\theta_1+\theta_2$ ensures that the inputs at $A_1$ and $A_2$ satisfy the average constraint required by each Angle Average Gadget, so the desired configuration indeed exists by Lemma~\ref{lem:scale-invariant-gadgets}.
    
  \item[Vector Sum Gadget.]
    In this gadget illustrated in Figure~\ref{fig:vector-sum-gadget}, each $Q\times Q$ cell labeled with an angle $\alpha_j$ or $\beta_j$ is a Copy Gadget, and each cell with an ``x'' is a Crossover Gadget. As with the previous compound gadgets, all unused transmission edges are frozen, as are the two transmission sliceforms illustrated in gray instead of black at the right of the bottom-most Vector Average Gadget. The analysis is similar to that of the Angle Sum Gadget, so we omit the details. \qedhere

  \end{descriptionflush}
\end{proof}

\subsection{Combining the Gadgets}
\label{sec:combining-gadgets}

We are now prepared to use these gadgets to construct an extended linkage that draws a piece of $Z(F)$, up to a translation. Recall that $F = \{f_1,\ldots,f_s\}$ is a family of polynomials in $\bR[x_1,y_1,\ldots,x_m,y_m]$, each with total degree at most~$d$. We apply the change of coordinates $(x_k,y_k) = 2r\cdot\Rect(\alpha_k,\beta_k)$ (for $1\le k\le m$) to write each polynomial $f_j$ (for $1\le j\le s$) in angular form as in Lemma~\ref{lem:angular-coefficients}:
\begin{equation}\label{eq:angular-form-in-proof-at-end}
  f_j(\widevec{xy}(\widevec{\alpha\beta})) - f_j(\vec 0) = \sum_{u=0}^3 \sum_{I\in\Coeffs(2m,d)} i^u\cdot d_{j,u,I}\cdot \left(\exp(i\cdot(I\cdot\widevec{\alpha\beta})) - 1\right),
\end{equation}
where the numbers $d_{j,u,I}$ are nonnegative. If the polynomials in $F$ do not have integer coefficients, we may need to modify the input slightly: scale up $F$ uniformly (which scales the $d_{j,u,I}$ coefficients by the same amount) until each nonzero term $d_{j,u,I}$ is at least $1$, and let $M$ be an upper bound on the scaled up coefficients in $F$ (as standard polynomials in $\widevec{xy}$). If $F$ originally had integer coefficients, the $d_{j,u,I}$ were already integers, so this scaling is unnecessary.

For each $1\le j\le s$, the numbers $d_{j,u,I}$ are nonnegative and add to at most $6^d \cdot r^d\cdot M\cdot \binom{2m+d}{d} \le R\delta/2$, by Lemma~\ref{lem:angular-coefficients} and by the choice of $R$ in Section~\ref{sec:parameters}. This shows that every partial sum of the right-hand side of~\eqref{eq:angular-form-in-proof-at-end} has magnitude at most $R\delta/2$. In particular, the Vector Creation, Vector Rotation, Vector Average, and Vector Sum Gadgets can all safely construct these terms and partial sums, according to the gadget specifications in Lemmas~\ref{lem:scale-invariant-gadgets} through~\ref{lem:compound-gadgets}. The same holds for the value $f_j(\vec 0)$ and the End Gadget, because $|f_j(\vec 0)|\le M \le R\delta/2$.

The extended linkage $\cE = \cE(F)$ that will be built in this section will draw a translation of $Z(F)\cap \left(2r\cdot \Rect([-\delta/d,\delta/d]^2)\right)^m$, which by Lemma~\ref{lem:rect-contains-box} contains $Z(F)\cap[-1,1]^2$ because $2r\cdot \delta/(2d) \ge 1$ by choice of $r$.

\paragraph{Step 1: Grid.}

To begin, define extended linkage $\cE$ as a rectangular grid of $Q\times Q$ cells, with $O(\poly(m^d,d^d,s))$ cells on each side. Each cell is given tolerance-$0$ corners and \term{transmission edges} attached at cell edge midpoints. The transmission vertices (which are sliceforms) initially have frozen corners. Add pins to three noncollinear vertices of this grid, so the entire grid $\cE$ is globally rigid. There will be no other pins anywhere in $\cE$; pins illustrated in the gadgets above are removed before use here, though they are \emph{effectively} still present because the grid $\cE$ is globally rigid.

In the rest of the construction, gadgets will be added to this grid, with their transmission vertices' corners upgraded to tolerance-$\delta$ as required by the gadgets.

\paragraph{\boldmath Step 2: Start Gadgets to set up $\alpha_1,\beta_1,\ldots,\alpha_m,\beta_m$.}

Add $m$ Start Gadgets $\cLstart(k), 1\le k\le m$ to $\cE$ as illustrated at the top of Figure~\ref{fig:example-angle-sum}, and let $v_k$ be the vertex in $\cLstart(k)$ corresponding to vertex~$v$ in Figure~\ref{fig:start-gadget}; these vertices $X = \{v_1,\ldots,v_m\}$ will be the drawing vertices of $\cE$. Likewise define corners $\Lambda_k$ and $\Gamma_k$ as the corners from $\cLstart(k)$ that are labelled $\Lambda$ and $\Gamma$ in Figure~\ref{fig:start-gadget}, and set $(\alpha_k,\beta_k) = (\alpha_k(C),\beta_k(C)) = \Offset_{(\Lambda_k,\Gamma_k)}(C)$ for a configuration $C$, as usual.

With $\cE$ in this state, by Lemma~\ref{lem:unscaled-gadgets} (Start Gadget), $\Conf(\cE)$ is perfectly described by corners $Y = (\Lambda_1,\Gamma_1,\ldots,\Lambda_m,\Gamma_m)$: the map $\Offset_Y$ is a homeomorphism of $\Conf(\cE)$ with $[-\delta,\delta]^{2m}$.

\begin{figure}[hbt]
  \centering
  \includegraphics{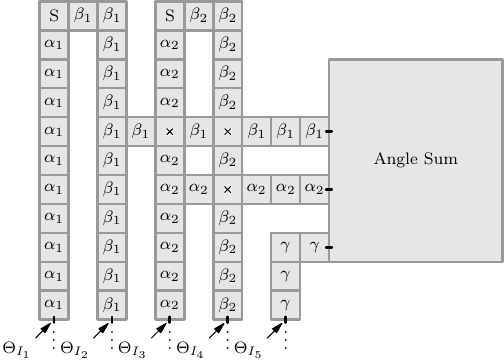}
  \caption{Constructing the sum $\gamma = \beta_1+\alpha_2$. Cells with ``S'' are Start Gadgets; those with ``x'' are Crossover Gadgets; and those with $\alpha_k$, $\beta_k$, or $\gamma$ are Copy Gadgets.}
  \label{fig:example-angle-sum}
\end{figure}

\paragraph{\boldmath Step 3: Construct all linear combinations $\theta_I$.}

In the next step, we modify $\cE$ further to construct all the angles $\theta_I := I\cdot\widevec{\alpha\beta}$, where $\widevec{\alpha\beta} = (\alpha_1,\beta_1,\ldots,\alpha_m,\beta_m)$ and $I$ ranges over all vectors of $\Coeffs(2m,d)$.

We start with base cases of the angle $0$ (corresponding to $I=(0,\ldots,0)$) and the angles $\alpha_k$ and $\beta_k$ (corresponding to $I = (0,\ldots,0,1,0,\ldots,0)$).
The angle $0$ can be constructed by placing a $0$ angle constraint on any transmission edge.
The angles $\alpha_k$ and $\beta_k$ are already constructed by the Start Gadgets.
Use Copy Gadgets to build \term{wires} that copy all of these base cases vertically along columns, as shown in Figure~\ref{fig:example-angle-sum}. We call this collection of wires the \term{wire column}, and it will accumulate other wires constructing other angles as the construction continues.

By ordering all the remaining vectors $I\in\Coeffs(2m,d)$ according to the sum of the absolute values of entries in $I$, each successive $\theta_I$ may then be computed as $\theta_I = \theta_{I'}\pm\alpha_k$ or $\theta_I = \theta_{I'}\pm\beta_k$, which can be constructed using a single Addition Gadget. To do this, as illustrated in Figure~\ref{fig:example-angle-sum}, we use Crossover Gadgets to transmit the desired input angles to the Addition Gadget, and then the constructed $\theta_I$ joins the wire column. Repeat this for each remaining $I\in\Coeffs(2m,d)$.

For each $I$, let $\Theta_I$ be any transmission corner in the wire coming from the construction of $\theta_I$ above. We claim that $\Conf(\cE)$ is still perfectly described by the corners $Y = (\Lambda_1,\Gamma_1,\ldots,\Lambda_m,\Gamma_m)$, but this time on a smaller range: we will show that $\Offset_Y$ is a homeomorphism of $\Conf(\cE)$ with $[-\delta/d,\delta/d]^{2m}$, and furthermore that the corners $\Theta_I$ indeed have angle $\Offset_{\Theta_I}(C) = \theta_I(C)$. The latter claim follows from induction on $I$. To see that the $\alpha_k$ and $\beta_k$ are each precisely (and independently) constrained to the interval $[-\delta/d,\delta/d]$, first consider the corner $\Theta_{I}$ where $I = (d,0,0,\ldots,0)$. We have $\theta_I = d\cdot \alpha_1$, and the $\delta$-tolerance at corner $\Theta_I$ requires this value to remain in $[-\delta,\delta]$, so $\alpha_1$ must remain in $[-\delta/d,\delta/d]$. The other entries in $\widevec{\alpha\beta}$ are likewise constrained by analogous arguments. In the other direction, for any vector $\widevec{\alpha\beta} \in [-\delta/d,\delta/d]^{2m}$ and any $I\in\Coeffs(2m,d)$, the value $I\cdot\widevec{\alpha\beta}$ lies in the interval $[-\delta,\delta]$ by triangle inequality. Thus, by our gadget analyses (Lemmas~\ref{lem:scale-invariant-gadgets}, \ref{lem:unscaled-gadgets}, and \ref{lem:compound-gadgets}), any such vector $\widevec{\alpha\beta}$ may be realized by a unique configuration of~$\cE$.\looseness=-1 %

\paragraph{\boldmath Step 4: Construct the terms in the angular expansions of $f_1,\ldots,f_s$.}

For each $1\le j\le s$ and each nonzero coefficient $d_{j,u,I}$ in the angular representation of $f_j(\widevec{xy}(\widevec{\alpha\beta}))$, insert a Vector Creation Gadget using angle $\theta_I$ and length $w = d_{j,u,I}$ (which indeed satisfies $1\le |w|\le R\delta/2$), and if $u\ne 0$, chain this with $u$~Vector Rotation Gadgets. The result constructs the desired vector $i^u\cdot d_{j,u,I}\cdot (\exp(i\cdot\theta_I) - 1)$: specifically, the result is a pair of wires with transmission corners $\Lambda'_{j,u,I}$ and $\Gamma'_{j,u,I}$ whose angles $(\alpha'_{j,u,I}(C),\beta'_{j,u,I}(C)) = \Offset_{\Lambda'_{j,u,I},\Gamma'_{j,u,I})}(C)$ satisfy
\begin{equation*}
  R\cdot\Rect(\alpha'_{j,u,I},\beta'_{j,u,I}) = 
  i^u\cdot d_{j,u,I}\cdot (\exp(i\cdot\theta_I) - 1).
\end{equation*}
These additional gadgets do not constrain $\cE$'s movement in any way, so the map $\Offset_Y$ is still a homeomorphism of $\Conf(\cE)$ with $[-\delta/d,\delta/d]^{2m}$. 

\paragraph{\boldmath Step 5: Add the vectors to construct $f_j(\protect\widevec{\alpha\beta}) - f_j(\vec 0)$.}
\label{step:5}

For each $1\le j\le s$, use Vector Addition Gadgets to successively construct the sum of the nonzero vectors among ${R\cdot\Rect(\alpha'_{j,u,I},\beta'_{j,u,I})}$, resulting in a single pair of wires with transmission corners $\Lambda'_j,\Gamma'_j$ whose angle offsets $\alpha'_j,\beta'_j$ satisfy
\begin{equation*}
  R\cdot\Rect(\alpha'_j,\beta'_j) = \left(f_j(\widevec{\alpha\beta}) - f_j(\vec 0), 0\right).
\end{equation*}
(Join these wires to the wire column, as usual.) As above, these additional gadgets do not constrain $\cE$'s movement, so $\Offset_Y:\Conf(\cE)\to[-\delta/d,\delta/d]^{2m}$ is still a homeomorphism.

\paragraph{\boldmath Step 6: Conclude with End Gadgets}

Finally, for each $1\le j\le s$, feed the $\alpha'_j$ and $\beta'_j$ wires into an End Gadget that enforces the constraints
\begin{equation*}
  R\cdot\Rect(\alpha'_j,\beta'_j) = \left(-f_j(\vec 0), 0\right).
\end{equation*}
The configuration space of the resulting linkage~$\cE$ is then homeomorphic to the subset of $\widevec{\alpha\beta}\in[-\delta/d,\delta/d]^{2m}$ that satisfies the constraints of these End Gadgets, i.e.,
\begin{equation*}
  \Conf(\cE) =
  \left\{
    \widevec{\alpha\beta} \in [-\delta/d,\delta/d]^{2m} \mid
    f_j(\widevec{\alpha\beta}) = 0, 1\le j\le s
  \right\}.
\end{equation*}
And because the offset of each vertex $v_k$ from the lower-left corner of its cell is precisely $(Q/5,Q/5) + (2r,2r) + 2r\cdot\Rect(\alpha_k,\beta_k) = (Q/5,Q/5) + (2r,2r) + (x_k,y_k)$, these vertices $\{v_1,\ldots,v_m\}$ indeed draw a translation of
\begin{equation*}
  Z(F)\cap\left(2r\cdot \Rect([-\delta/d,\delta/d],[-\delta/d,\delta/d])\right)^m,
\end{equation*}
as required.

\subsection{Sliceform and Angle Restrictor Gadgets}
\label{sec:sliceform-and-angle-restrictor-gadgets}

In Section~\ref{sec:implementing-with-partially-rigidified}, we will convert the extended linkage $\cE(F)$ constructed above into a partially rigidified linkage (with no other constraints).
In this section, we describe two gadgets to help with this process: the \emph{Sliceform Gadget} obviates the need for sliceform vertices, and the \emph{Angle Restrictor Gadget} enforces $\cE$'s angle constraint, after replacing each edge with a rigidified tree.

\begin{lemma}[Sliceform Gadget]
  \label{lem:sliceform-gadget}
  The extended linkage that has four edges meeting at a sliceform vertex with tolerance $\theta$ (where $\theta$ is in $\{0,\delta,\eps\}$) is perfectly simulated by the extended linkage $\cLslice(\theta)$, where all non-frozen corners are given tolerance $\theta$ (see Figure~\ref{fig:sliceform-gadget}).
  Furthermore, the global minimum feature size is at least $|v-w_j|/8$ (which is the same for all $j=1,2,3,4$ and in all configurations),
  and all vertices lie within distance $\frac{3}{2} |v-w_j|$ of~$v$,
  in all configurations.
\end{lemma}

\begin{figure}
  \centering
  \includegraphics{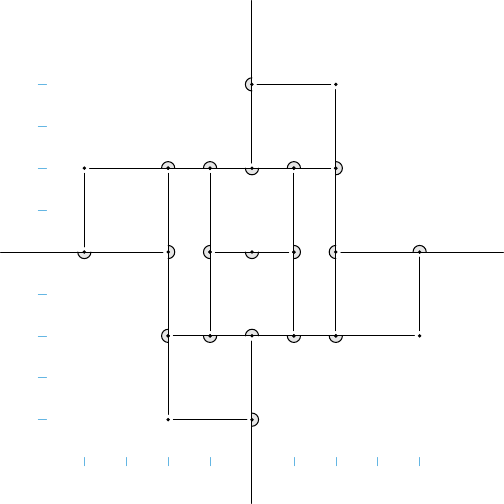}
  \caption{The Sliceform Gadget keeps $w_1,v,w_3$ and $w_2,v,w_4$ collinear.
    All marked vertices have coordinates at integer multiples of $|v-w_j|/4$, indicated by cyan axis notches.}
  \label{fig:sliceform-gadget}
\end{figure}

\begin{proof}
  For convenience, scale so that $|v-w_1|$ is $4$ in the initial configuration, so that the cyan notches in Figure~\ref{fig:sliceform-gadget} have unit distance between them, and all edge lengths are integers.
  By Lemma~\ref{lem:parallelogram-linkage} applied to the many rectangles of Figure~\ref{fig:sliceform-gadget}, all parallelograms in the figure remain parallelograms in all configurations. This, paired with the frozen angles keeping collinear edges collinear, ensures the desired collinearity and that the distances $|v-w_j|$ remain constant, namely~$4$. The gadget admits the desired motion by the simultaneous motion of all the parallelograms.
  All edge lengths are at least $1$, so by Lemma~\ref{lem:parallelogram-linkage}, the global minimum feature size is at least $1/2 = |v-w_j|/8$.
  By triangle inequality, all vertices remain at distance at most $6 = \frac{3}{2} |v-w_j|$ from~$v$.
\end{proof}

The last gadget is built from \emph{partially rigidified} linkages instead of \emph{extended} linkages. Recall that a \term{partially rigidified} linkage is a constrained linkage where all constraints are \term{rigid constraints}, $\RigidCon(H,C_H)$, each forcing a certain subgraph $H$ to maintain a chosen, rigid shape given by a specified configuration $C_H$.

\begin{figure}[t]
  \centering
  \includegraphics{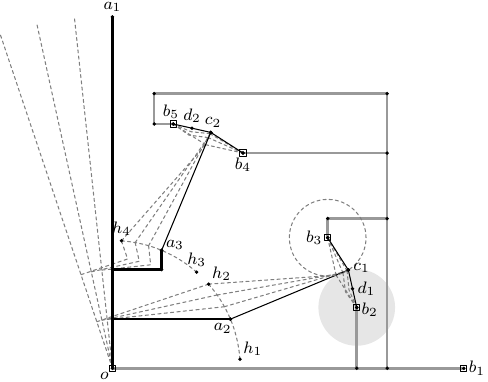}
  \hfill
  \includegraphics{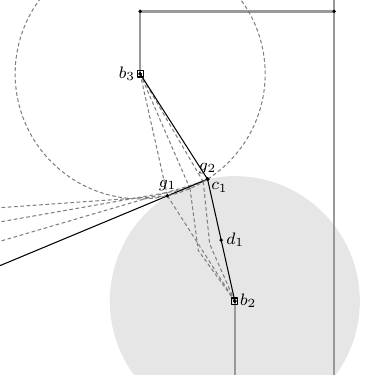}
  \caption{Angle Restrictor Gadget, $\cLanglerestrictor(n)$, shown in full (left) and closeup (right). Tree $A$ is bold and black, with leaves at $o$ and $a_j$, $j=1,2,3$; $B$ is bold and gray, with leaves at $o$ and $b_j$, $j=1,2,3,4,5$.
  Coordinates depend on input integer $n$ such that
  the range of motion of the gadget approaches zero as $n$ increases.
  Dashed segments show alternative configurations for edges,
  with vertices constrained to dashed circles.}
  \label{fig:angle-restrictor-gadget}
\end{figure}

\begin{lemma}[Angle Restrictor Gadget]
  \label{lem:angle-restrictor-gadget}
  For an integer $n \ge n_\eps = 5000$ (from Equation~\ref{eq:eps}), let $\cLanglerestrictor = \cLanglerestrictor(n)$ be the partially rigidified linkage shown in Figure~\ref{fig:angle-restrictor-gadget}, with two rigid constraints fixing orthogonal trees $A$ and $B$ in the configurations shown (but not governing their orientations or relative position). Tree $A$ is drawn with bold black lines, while $B$ has bold gray lines. Tree $B$ is also pinned to the plane for convenience, with $o$ at $(0,0)$.

  Then $\cLanglerestrictor$ is globally noncrossing, and vertex $a_1$ draws the locus $$\{\ell(o a_1)\cdot(\cos\theta,\sin\theta)\mid \pi/2-\gamma\le \theta \le \pi/2+\gamma\}$$
  liftably and rigidly, where $\gamma = \gamma(n)$ is given by
  \begin{equation*}
    \gamma(n) = \cos\inv\left(1 - \frac{3}{10}\cdot\frac{2n}{n^2+1}\right).
  \end{equation*}
  Furthermore,

  \begin{enumerate}
  \item The initial configuration shown in the figure has rational coordinates.

  \item\label{thmpart:angle-restrictor-gadget-combinatorial-embedding} All configurations are noncrossing and agree with the same combinatorial embedding. In particular, tree $A$ is always configured with the same orientation as in Figure~\ref{fig:angle-restrictor-gadget}.

  \item\label{thmpart:angle-restrictor-gadget-60-240} At every vertex that is not interior to one of the rigidified trees (in other words, every vertex given a name in Figure~\ref{fig:angle-restrictor-gadget}), every pair of consecutive edges forms an angle strictly between $60^\circ$ and $240^\circ$ in every configuration, except for the external angle $\angle a_1ob_1 \in [3\pi/2-\gamma,3\pi/2+\gamma]$.

  \item\label{thmpart:angle-restrictor-gadget-unique-90}
    There is a \emph{unique} configuration of $\cLanglerestrictor$ having $\angle b_1oa_1 = \pi/2$.

  \item Finally, the limiting case $\cLanglerestrictor(\infty)$ is well-defined, satisfies all of the above, and is globally rigid.
  \end{enumerate}
\end{lemma}

\begin{note}
  The assumption of pinning $B$ to the plane is equivalent to choosing a
  coordinate frame relative to $B$'s configuration.
  Thus, if we do not pin $B$, then this lemma still describes tree $A$'s position relative to that of~$B$.
  In particular, if either tree is known to have orientation matching that of Figure~\ref{fig:angle-restrictor-gadget}, then they both must---one cannot ``flip'' relative to the other---and $\angle b_1oa_1$ must lie in $[\pi/2 - \gamma, \pi/2+\gamma]$.
\end{note}

\begin{figure}[hbt]
  \centering
  \includegraphics{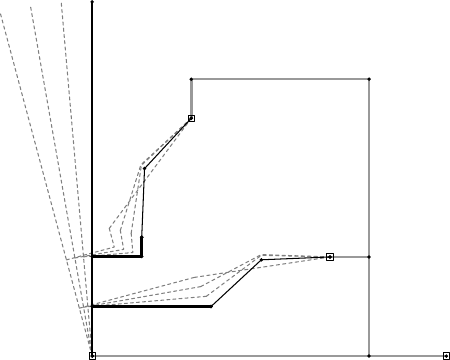}
  \caption{A simpler Angle Restrictor Gadget is possible if uniqueness at $\theta=\pi/2$ is not required. This is shown only for contrast with Figure~\ref{fig:angle-restrictor-gadget}; this simpler version is not formally analyzed or used in this paper.}
  \label{fig:easy-angle-restrictor-gadget}
\end{figure}

\begin{note}
  The uniqueness property of Part~\ref{thmpart:angle-restrictor-gadget-unique-90} is the crucial step in proving hardness of \emph{global} rigidity, but it also complicates the gadget. Without this requirement, a much simpler construction is possible, such as that shown in Figure~\ref{fig:easy-angle-restrictor-gadget}.
\end{note}

\begin{proof}
  First we specify the initial configuration $C_0$, drawn solid in Figure~\ref{fig:angle-restrictor-gadget}, more precisely. Vertex $o$ rests at the origin, $(0,0)$. Let $\kappa = \sin\inv(5/13)$, which is a \term{Pythagorean angle}, meaning $\sin(\kappa) = 5/13$ and $\cos(\kappa) = 12/13$ are both rational, and define another Pythagorean angle $\phi = \phi(n) = \sin\inv(2n/(n^2+1))$ (its cosine is $(n^2-1)/(n^2+1)\in\bQ$), where $\phi(n) \le \phi(n_\eps) < 0.0004$. In the initial configuration, $a_2 = (\cos\kappa,\sin\kappa)$, $c_1 = 2a_2$, $\angle b_3c_1o = \angle oc_1b_2 =  \pi/2-\phi$, and $d_1$ is the midpoint of $c_1b_2$. Edges $c_1 d_1$, $d_1 b_2$, and $c_1b_3$ have length $t,t,2t$ respectively, where $t = 3/20$. (This arbitrary choice of $t$ allowed for a prettier figure than others we tried.)
  Vertices $a_3,c_2,b_4,b_5,d_2$ are initially configured at the reflections through the line $y=x$ of the initial positions of vertices $a_2,c_1,b_3,b_2,d_1$ respectively.
  The upper right vertex of $B$ has coordinates $(2+t,2+t)$.
  The initial configuration has minimum feature size realized by point $d_1$ to segment $a_2 c_1$ (and $d_2$ to $a_3 c_2$), namely $t \sin \angle a_2 c_1 d_1 = t \cos \phi \geq t \cos \phi(n_\eps) > 0.14999$.
  Because we chose Pythagorean angles, all coordinates are initially rational. We will proceed as if $|o-a_1|=|o-b_1|=2.75$ (as drawn), but they may be longer without affecting the global minimum feature size.
  
  We now investigate a general configuration $C$ of $\cLanglerestrictor$. Bar $b_3c_1$ restricts $c_1$ to the dashed circle with radius $2t$, and bars $b_2d_1$ and $d_1c_1$ further restrict vertex $c_1$ to the disk with radius $2t$, as shown. As a result, $c_1$ must lie on the small circular arc $g_1g_2$ (see Figure~\ref{fig:angle-restrictor-gadget}), centered at $b_3$ with angle $2\phi$. (Note that $g_1$ and $g_2$ are fixed points in the plane, not linkage vertices.) In turn, because $|o-a_2| = |a_2-c_1| = 1$, vertex $a_2$ is confined to the circular arc $h_1h_2$ centered at $o$ with radius $1$, where the endpoints of confinement, $h_1$ and $h_2$, are determined by rhombus $oh_1g_1h_2$ of side length~$1$. Because $o$, $g_1$, and $g_2$ are collinear, we may compute that $|o-g_1| = 2 - 4t\sin\phi$, and so $\gamma := \frac{1}{2}\angle h_1oh_2$ is given by $\gamma = \cos\inv(1-2t\sin\phi)$. We have $\gamma(n) \le \gamma(n_\eps) < 0.0155$, which is less than the fixed value $|\angle a_2oa_3| = \pi/2 - 2\kappa > 0.78$, so arcs $h_1h_2$ and $h_3h_4$ are disjoint. This forces $\angle a_2oa_3$ to have counterclockwise orientation, and so tree $A$ must maintain its orientation as claimed, and $\theta = \angle b_1oa_1$ must remain in the closed interval between $\pi/2\pm\gamma$.

  We must show that $\cLanglerestrictor$ cannot move far enough to intersect itself or to violate the angle conditions in Part~\ref{thmpart:angle-restrictor-gadget-60-240}. Toward this goal, let us first argue that the angle of each edge in configuration $C$ differs from that in $C_0$ by at most $\pm 1/60$ radians. The edges in tree $A$ can rotate by at most $\pm\gamma$, and likewise edges $a_2c_1$ and $a_3c_2$ can change their direction by at most $\pm\gamma$, where $\gamma < 1/60$ as above. Edges $b_3 c_1$ and $b_4 c_2$ can rotate by at most $2\phi \le 2\phi(n_\eps) < 0.0008 < 1/60$. Only the edges incident to $d_1$ and $d_2$ remain, and by symmetry, we may consider only those at $d_1$. We will measure these edges relative to the vector $b_2c_1$, which is not an edge of the linkage but is nevertheless a useful reference. This vector $b_2c_1$ can change its direction by at most $2\phi \le 2\phi(n_\eps) < 1/120$. Its length is at least $|b_2-b_3| - 2t = 2t(2\cos(\phi) - 1)\ge 2t(2\cos(\phi(n_\eps))-1)$, which may be checked to be at least $2t\cdot\cos\frac{1}{120}$. Isosceles triangle $b_2d_1c_1$ then shows that $|\angle d_1b_2c_1| = |\angle d_1c_1b_2| \le 1/120$ radians. So each of $d_1b_2$ and $d_1c_1$ can rotate by at most $1/120+1/120=1/60$, as claimed.

  Each vertex is connected to a vertex in the stationary tree $B$ by a path of length at most $3$, so by Lemma~\ref{lem:rotating-bar-displacement}, each vertex has been offset from its initial position by at most $3\cdot 1/60 = 1/20$. Because $C_0$ has minimum feature size $> 0.14999$, we conclude by Lemma~\ref{lem:feature-size} that $C$ is noncrossing (with minimum feature size at least $0.14999 - 2/20 > 0.04999 > 0$). Each corner at each named vertex initially had an angle between $\pi/2 - \phi$ and $\pi + \kappa + \phi$, and these angles are at most $2\cdot 1/60 = 1/30$ different in configuration $C$, which confirms that the corresponding angles in $C$ lie safely between $\pi/3$ and $4\pi/3$: indeed,
  \begin{equation*}
    \frac{\pi}{2} - \phi - \frac{1}{30} > 88.06^\circ > 60^\circ
    \qquad\text{and}\qquad
    \pi+\kappa+\phi+\frac{1}{30} < 204.56^\circ < 240^\circ,
  \end{equation*}
  confirming Part~\ref{thmpart:angle-restrictor-gadget-60-240}. It also proves that $C$ and $C_0$ agree with the same combinatorial embedding (Part~\ref{thmpart:angle-restrictor-gadget-combinatorial-embedding}).

  Conversely, given any angle $\theta\in[\pi/2-\gamma,\pi/2+\gamma]$, we will argue there are either $1$ or $4$ configurations $C$ having $\angle b_1oa_1 = \theta$. This angle determines the position of $a_2$ along arc $h_1h_2$, and then there is a unique point $c_1$ on arc $g_1g_2$ with $|a_2-c_1| = 1$. If $c_1$ lies at $g_2$ (i.e., $\theta = \pi/2$), then because $|b_2-c_1| = 2t$, vertex $d_1$ must be configured at the midpoint of $b_2c_1$. The case is similar if $c_1$ is at $g_1$, i.e., $\theta = \pi/2 \pm \gamma$. Otherwise, $d_1$ lies in the interior of the solid gray disk, so there are two possible choices for the location of $d_1$. The same is true for the other assembly anchored at $a_3$, $b_4$, and $b_5$, and furthermore, vertices $d_1$ and $d_2$ may vary continuously with $\theta$. This proves that vertex $a_1$ draws liftably and rigidly.

  In the limiting case $n = \infty$, the angles $\phi$ and $\gamma$ are $0$, so the resulting linkage $\cLanglerestrictor(\infty)$ is globally rigid.
\end{proof}

\subsection{Implementing Extended Linkages with Partially Rigidified Linkages}
\label{sec:implementing-with-partially-rigidified}

We now use the gadgets from Section~\ref{sec:sliceform-and-angle-restrictor-gadgets}
to transform the extended linkage $\cE$, constructed in Section~\ref{sec:combining-gadgets}, into a partially rigidified linkage $\cL = \cL(F)$ that simulates $\cE$---not perfectly, but at least liftably and rigidly.

Extended linkage $\cE$ has all edge lengths at least $1$ and has global minimal feature size at least $1/2$, because each gadget individually has these properties.
Construct a new linkage $\cE' = \cE'(F)$ by replacing each sliceform vertex $v$ in $\cE$ with a Sliceform Gadget $\cLslice$ of Lemma~\ref{lem:sliceform-gadget} with the same angle tolerance, scaled so that $|v-w_j| = 1/8$.
By Lemma~\ref{lem:sliceform-gadget}, each sliceform gadget individually has minimum feature size at least $|v-w_j|/8 = 1/64$ and lives within a disk of radius $\frac{3}{2} |v-w_j| = 3/16$ centered at the replaced vertex.
Thus the feature size between different sliceform gadgets is at least $\frac{1}{2} - 2 \cdot \frac{3}{16} = \frac{1}{8}$.
Therefore $\cE'$ perfectly simulates $\cE$, is globally noncrossing, and has global minimum feature size at least $1/64$.

\newcommand\vcentertrick[1]{\raisebox{-0.5\height}{#1}}
\begin{figure}[phbt]
  \centering
  \vcentertrick{\includegraphics{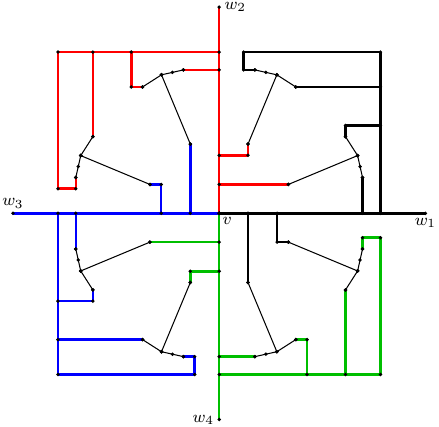}}
  \hfil
  \vcentertrick{\includegraphics{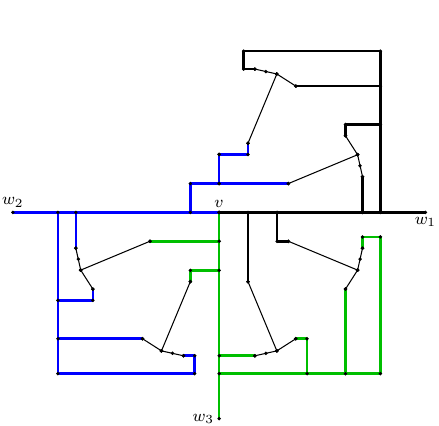}}
  \\
  \vcentertrick{\includegraphics{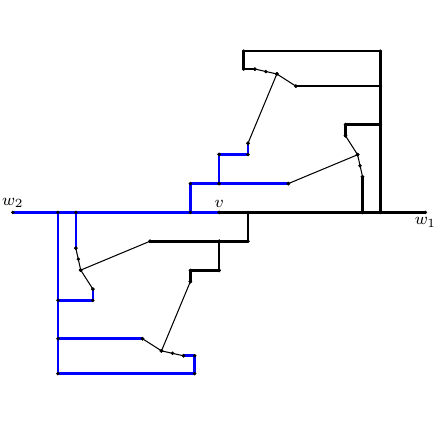}}
  \hfil
  \vcentertrick{\includegraphics{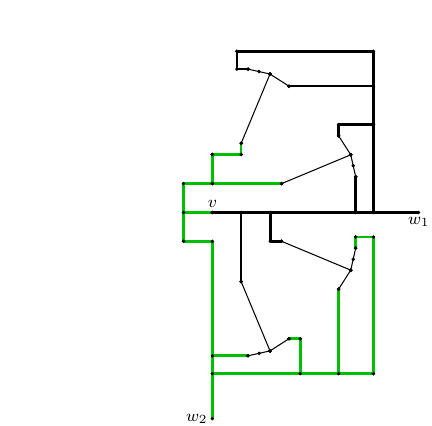}}
  \caption{To turn extended linkage $\cE'$ into $\cL$, a neighborhood of each vertex is replaced with an \emph{Angle Restrictor Gadget neighborhood} tailored to the base angles of its corners, as shown. The possible vertex types are $(90^\circ,90^\circ,90^\circ,90^\circ)$ (top left), $(180^\circ,90^\circ,90^\circ)$ (top right), $(180^\circ,180^\circ)$ (bottom left), $(270^\circ,90^\circ)$ (bottom right), and $(360^\circ)$ (not shown).
  In the bottom two cases of degree $2$, a single Angle Restrictor Gadget would suffice, but we draw both for consistency.
  Each rigidified tree is a different color of thick edges.}
  \label{fig:angle-restrictor-gadget-neighborhood}
\end{figure}

We transform $\cE'$ into a partially rigidified linkage $\cL = \cL(F)$ as follows: first, consider each edge $u v$ of $\cE'$ as a rigidified tree $T_{u v}$ initially containing just one edge. For each vertex $v$ with neighbors $w_j$ (for $1\le j\le \deg v$), we will modify the trees $T_{v w_j}$ in a small neighborhood of $v$ with (scaled down) Angle Restrictor Gadgets (from Lemma~\ref{lem:angle-restrictor-gadget}) adapted as necessary to $v$'s structure, to faithfully implement $v$'s angle constraints. Specifically, look at the angles $A(\Lambda)$ for corners $\Lambda$ around $v$: the $5$ possibilities, up to cyclic reordering, are  $(90^\circ,90^\circ,90^\circ,90^\circ)$ $(180^\circ,90^\circ,90^\circ)$, $(180^\circ,180^\circ)$, $(270^\circ,90^\circ)$, or $(360^\circ)$. The last case corresponds to $\deg v = 1$, in which case its angle constraint does nothing and may be ignored. In each of the first four cases, we modify the trees $T_{v w_j}$ in a small region around $v$ with the appropriate \term{Angle Restrictor Gadget neighborhood} from Figure~\ref{fig:angle-restrictor-gadget-neighborhood}. Each Angle Restrictor Gadget within this neighborhood is given the shape of $\cLanglerestrictor(n_\eps)$, $\cLanglerestrictor(n_\delta)$, or $\cLanglerestrictor(\infty)$ depending on whether the corner $\Lambda$ has $\Delta(\Lambda) = \eps$, $\delta$, or $0$ respectively.
(Recall $n_\eps$ and $n_\delta$ are given by Equations~\ref{eq:eps}--\ref{eq:delta}.)
Finally, remove the three pins of $\cE'$ (that serve only to fix the background grid in place), and replace them in $\cL$ with three noncollinear pins in one of the rigidified trees $T_{u_1 u_2}$ built from an edge $u_1 u_2$ in the background grid.

We claim that $\cL$ simulates $\cE'$ (and hence $\cE$) liftably and rigidly. For each rigidified tree $T_{u v}$, at least one of $u$ or $v$ has degree greater than $1$ in $\cE'$ because $\cE'$ is connected, so $T_{u v}$ participates in at least one Angle Restrictor Gadget neighborhood. As shown in Lemma~\ref{lem:angle-restrictor-gadget} (Angle Restrictor Gadget), rigidified trees connected to each other by an Angle Restrictor Gadget cannot change orientation relative to each other. Because tree $T_{u_1 u_2}$ is pinned rigidly in place and all rigidified trees in $\cL$ are connected to each other by a sequence of Angle Restrictor Gadgets, none of the rigidified trees in $\cL$ can change orientation at all. This same lemma guarantees that all angle constraints are enforced by the Angle Restrictor Gadgets. In particular, the entire background grid in $\cL$ is fixed rigidly in place by the pins in $T_{u_1 u_2}$. The Angle Restrictor Gadgets maintain their constraints in a liftable and rigid fashion, so $\cL$ indeed simulates $\cE'$ liftably and rigidly.

Below we will show that this linkage $\cL = \cL(F)$ satisfies all requirements of the Main Theorem.

\begin{theorem}
  \label{thm:main-construction-works}
  This linkage $\cL = \cL(F)$, built from $F = \{f_1,\ldots,f_s\}$, satisfies all of the requirements of Theorem~\ref{thm:main-theorem}, with
  \begin{equation*}
    D = 2^{12}\cdot 13\cdot 20\cdot \left(n_\eps^2+1\right)\cdot \left(n_\delta^2+1\right) \approx 4.26\cdot 10^{42}.
  \end{equation*}
\end{theorem}

\begin{proof}

  \begin{descriptionflush}
  \item[Part I.]
    \begin{enumerate}
    \item %
      Define $X\subset V(\cL)$ as the set of vertices of $\cL$ that simulate vertices $v_1,\ldots,v_m$ of $\cE$. Define translation $T(\widevec{xy}) = \widevec{xy}+(a_1,b_1,\ldots,a_m,b_m)$ where for each $1\le k\le m$, $(a_k,b_k)-(Q/5+2r,Q/5+2r)$ are the coordinates of the lower left corner of the cell containing Start Gadget $\cLstart(k)$; in other words, $(a_k,b_k)$ is the position of vertex $v_k$ that would correspond to $(\alpha_k,\beta_k)=(0,0)$. We showed in Section~\ref{sec:combining-gadgets} that $\cE$ perfectly draws the region $T(Z(F)\cap U)$, where $U = (2r\cdot\Rect([-d/\delta,d/\delta]))^m$. Because $\cL$ liftably and rigidly simulates $\cE$, $\cL$ indeed liftably and rigidly draws the same trace.
      
    \item %
      Following the construction above, there are $m$ Start Gadgets,  $O(|\Coeffs(2m,d)|) = O(\poly(m^d,d^d))$ Angle Addition, Vector Creation, Vector Rotation, and Vector Addition Gadgets, and $s$ End Gadgets. Stacked one on top of the other to the side of the wire column (other than the Start Gadgets above the wire column), these gadgets require a height of $O(\poly(m^d,d^d,s))$. The wire column likewise has no more than $O(\poly(m^d,d^d,s))$ wires, so the entire construction fits in an $O(\poly(m^d,d^d,s))\times O(\poly(m^d,d^d,s))$ grid of cells. Each cell contains $O(1)$ vertices and edges, proving the claim.
      
    \item %
      We already concluded that $\cE'$ is globally noncrossing with global minimum feature size at least $2^{-8}$. In the insertion of Angle Restrictor Gadget neighborhoods to transform $\cE'$ to $\cL$, scale each gadget by a factor of $2^{-12}$, so that it fits within a $2^{-10}$ neighborhood of its vertex and (on its own) has a global minimum feature size at least $2^{-12}\cdot 1/20$. $\cE'$ ensures that each Angle Restrictor Gadget neighborhood stays well separated from other edges or Angle Restrictor Gadgets, so $2^{-12}\cdot 1/20 > 1/D$ is a lower bound on the global minimum feature size of $\cL$, as desired.

    \item %
      Each rigidified tree $T_{u v}$ in $\cL$ is constructed by attaching a portion of an Angle Restrictor Gadget neighborhood (as illustrated in Figure~\ref{fig:angle-restrictor-gadget-neighborhood}) to one or both endpoints of edge $u v$. These rigidified trees are indeed orthogonal and edge disjoint, possess at least three noncollinear vertices, and only connect to other trees or edges via leaves.

    \item %
      The combinatorial embedding $\sigma$ of $\cL$ follows from the combinatorial embedding of extended linkage $\cE'$ (which is unique by Lemma~\ref{lem:sliceform-gadget}), augmented with the unique combinatorial embedding of Angle Restrictor Gadgets (Lemma~\ref{lem:angle-restrictor-gadget}, Part~\ref{thmpart:angle-restrictor-gadget-combinatorial-embedding}). Any vertex $u$ of $\cL$ that is \emph{not} internal to a rigidified tree must either belong to an Angle Restrictor Gadget neighborhood, or come from a vertex of $\cE'$ of degree $1$. In the latter case, $u$ still has degree $1$ in $\cL$, so we need not consider it. If instead $u$ belongs to the Angle Restrictor Gadget neighborhood of some vertex $v\in\cE'$, it may be seen in Figure~\ref{fig:angle-restrictor-gadget-neighborhood} that $u$'s angles remain within $(60^\circ,240^\circ)$: indeed, by Lemma~\ref{lem:angle-restrictor-gadget} (Angle Restrictor Gadget), the only place this condition might be violated is at the center of the neighborhood, $u=v$, but Figure~\ref{fig:angle-restrictor-gadget-neighborhood} shows that the angles at $v$'s corners remain in $\pi/2\pm\gamma$ or $\pi\pm\gamma$.
      
    \item %
      This is true by construction.

    \end{enumerate}

  \item[Part II.]
    \begin{enumerate}[resume]
    \item %
      When the coefficients of $f_1,\ldots,f_s$ are integers, all coefficients $d_{j,u,I}$ are also integers bounded by $O(\poly(m^d,d^d,s,M))$ (by Lemma~\ref{lem:angular-coefficients}). Parameters $r,Q,R$ are also integers bounded by $O(\poly(m^d,d^d,s,M))$, so it may be seen that all edge lengths of $\cE$ are \emph{integers} no greater than $Q$ (including the Vector Creation Gadgets with integers $w=d_{j,u,I}$ and End Gadgets with integers $w=-f_j(0)$).
      The edges of $\cE'$ have lengths in $\frac{1}{32}\bZ$, because the scaled Sliceform Gadgets $\cLslice$ that were inserted into $\cE$ to form $\cE'$ have these edge lengths. Finally, the scaled Angle Restrictor Gadget neighborhoods (with each Angle Restrictor Gadget having shape $\cLanglerestrictor(n_\eps)$, $\cLanglerestrictor(n_\delta)$, or $\cLanglerestrictor(\infty)$) may be checked to have edge lengths in $\frac{1}{D}\bZ$, so the same can be said for $\cL$.

      The only edges of $\cL$ not contained in a rigidified tree correspond to those edges in Figure~\ref{fig:angle-restrictor-gadget} (Angle Restrictor Gadget) that are drawn with thin lines. In $\cL$, these edges all have lengths $2^{-12}$, $2^{-12}\cdot 3/20$, and $2^{-12}\cdot 3/10$, which are less than $D$.

    \item %
      Angular forms $g_j(\widevec{\alpha\beta})$ may be computed from polynomials $f_j(\widevec{xy})$ in deterministic time $O(\poly(m^d,d^d,s,M))$ by Lemma~\ref{lem:angular-coefficients}, and the magnitudes of the coefficients never exceed $O(\poly(m^d,d^d,s,M))$. Our transformations from $f_1,\ldots,f_m$ to $\cL(F)$ are explicit and deterministic.

    \end{enumerate}

  \item[Part III.]
    
    \begin{enumerate}[resume]
    \item %
      When the given polynomials $f_j$ satisfy $f_j(\widevec{0}) = 0$, our construction of $\cE(F)$ uses End Gadgets only with input $w=0$, which come with  initial configurations with coordinates that are integer multiples of $Q/40$ and are therefore integers, as in Lemma~\ref{lem:unscaled-gadgets}. The Vector Creation Gadgets, with integer input $w=d_{j,u,I}$, likewise have integer coordinates initially. All other gadgets start at integer coordinates unconditionally. Together, these specify an initial configuration of $\cE(F)$ corresponding to $\widevec{\alpha\beta}=\vec{0}$, i.e., $\widevec{xy} = \vec{0}$, which indeed has integer coordinates.

      The initial configurations of the Sliceform Gadgets and Angle Restrictor Gadgets were also illustrated in Figures~\ref{fig:sliceform-gadget} and~\ref{fig:angle-restrictor-gadget} and may be checked to have rational coordinates in $\frac{1}{D}\bZ$. These induce the desired initial configuration $C_0$ of $\cL(F)$.

    \item %
      Linkage $\cE'$ perfectly draws $T(Z(F)\cap U)$, meaning there is only one configuration of $\cE'$ mapping to $T(\vec{0})$. In this unique configuration, each Angle Restrictor Gadget must be in its initial state (corresponding to $\theta=\pi/2$ in Lemma~\ref{lem:angle-restrictor-gadget}), and there is only one such configuration for each gadget. So $C_0$ is indeed unique.      

    \item %
      All edges of $\cE'$ are initially axis aligned, so the rigidified trees built from these edges are likewise axis aligned.
      
    \item %
      As described above, our construction of $C_0$ is explicit and deterministic.

    \end{enumerate}

  \end{descriptionflush}

  \noindent All properties have been verified, so this concludes the proof.
\end{proof}

\subsection{Modifications for Strong Matchstick Universality}
\label{sec:matchstick-universality-full}

We may subtly modify the above proof of Theorem~\ref{thm:main-theorem} (Main Theorem) to prove that the proper subsets of $\bR^2$ drawn by matchstick linkages are \emph{exactly} the bounded semialgebraic sets. We will use one extra cell gadget when constructing extended linkage $\cE(F)$, the Crossing End Gadget (Figure~\ref{fig:crossing-end-gadget}, Lemma~\ref{lem:crossing-end-gadget}), which creates a crossing precisely when $g(\widevec{xy}) = 0$ for a given polynomial $g$. When linkage $\cE(F)$ is simulated by a matchstick linkage $\cM(F)$ as described in Section~\ref{sec:matchstick}, all of $\cE(F)$'s noncrossing configurations transfer to $\cM(F)$, i.e., thickening does not introduce unintended crossings. This allows us to draw semialgebraic sets of the form
\begin{equation*}
  \{\vec x\in\bR^k\in\bR^2\mid f_1(\vec x) = \cdots = f_s(\vec x) = 0, g_1(\vec x) \ne 0, \ldots,g_r(\vec x)\ne 0\},
\end{equation*}
as well as coordinate projections thereof. This is sufficient to draw any bounded semialgebraic set in the plane. Details follow.

\begin{figure}
  \centering
  \includegraphics{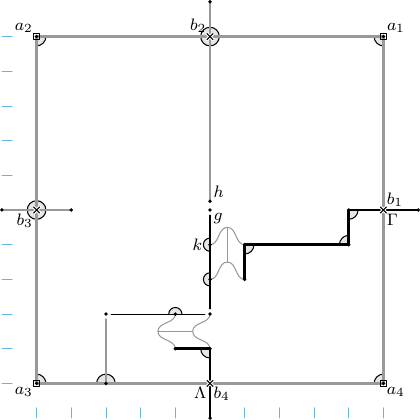}
  \caption{The Crossing End Gadget creates a crossing at $g=h$ precisely when $\alpha=\beta=0$, assuming $g$ remains on the line $y=Q/2$.}
  \label{fig:crossing-end-gadget}
\end{figure}

\begin{theorem}[Universality of Matchstick Linkages]
  \label{thm:matchstick-universality-full}
  The proper subsets $R\subsetneq\bR^2$ that are drawable by a matchstick linkage are exactly the bounded semialgebraic sets.
\end{theorem}

\begin{proof}
  Let $\cM = (\cL,\NXCon_{\cL})$ be a matchstick linkage that is connected and has at least one pin. For the underlying, unconstrained linkage $\cL$, the set $\Conf(\cL)$ is described by closed polynomial conditions and is therefore compact and algebraic. Noncrossing can be described by the nonvanishing of polynomials, so $\Conf(\cM) = \NXConf(\cL) \subset \Conf(\cL)$ is bounded and semialgebraic (not necessarily closed). Then, for any vertex $v$, the trace $\pi_{v}(\Conf(\cM))$ is also bounded and semialgebraic.

  Conversely, to show that every bounded, semialgebraic set can be drawn with a matchstick linkage, it suffices (by Lemma~\ref{lem:semialg-projection-ne0}) to show that every bounded, basic semialgebraic set of the form
  \begin{equation*}
    U = \{\vec x\in\bR^{2m} \mid f_1(\vec x) = \cdots = f_s(\vec x) = 0, g_1(\vec x) \ne 0, \ldots, g_r(\vec x)\ne 0\}
  \end{equation*}
  can be drawn by a matchstick linkage, up to translation. We may assume $U$ does not contain the origin (by translating if necessary), and by replacing each $g_j(\vec x)$ by $g_j(\vec x)\cdot|\vec x|^2$ (which does not modify $U$), we may further assume that each $g_j$ satisfies $g_j(\vec 0) = 0$. We first prove the special case where $U\subset[-1,1]^{2m}$; the general case is discussed at the end.

  By defining $F_j(\widevec{xy}) := f_j(D\cdot\widevec{xy})$ and $G_j(\widevec{xy}) := g_j(D\cdot\widevec{xy})$, we may write
  \begin{equation} \label{eq:semialg-solution-set}
    \frac{1}{D}\cdot U = \{\widevec{xy}\mid F_j(\widevec{xy})=0\text{ for $1\le j\le s$ and } G_j(\widevec{xy})\ne 0\text{ for $1\le j\le r$}\}.
  \end{equation}
  By scaling the polynomials $F_j$ and $G_j$ (which does not affect the solution set in \eqref{eq:semialg-solution-set}), we may further assume that their coefficients are in $[-1,1]$.

  Now we will use the ``coefficients as variables'' trick from the proof of Theorem~\ref{thm:unit-universality}.
  Namely, for each polynomial $P \in \{F_1,\dots,F_s,G_1,\dots,G_r\}$,
  and for each monomial $c_{P,J} \widevec{xy}^J$ in $P$,
  we create new variables $a_{P,J}$ and $b_{P,J}$, gather all of these new variables into a vector $\widevec{ab}$, and define the new polynomials
  \begin{equation*}
    P'(\widevec{xy},\widevec{ab}) := \sum_{J\text{ such that }c_{P,J}\ne 0} a_{P,J}\widevec{xy}^J.
  \end{equation*}
  Polynomials $F'_1,\dots,F'_s,G'_1,\dots,G'_r$ now have integer coefficients: in fact, all coefficients equal~$1$.
  It remains to implement the equations $F'_j(\widevec{xy},\widevec{ab}) = 0$, $G'_j(\widevec{xy},\widevec{ab}) \neq 0$, $(a_{F_j,J},b_{F_j,J}) = (c_{F_j,J},0)$, and $(a_{G_j,J},b_{G_j,J}) = (c_{G_j,J},0)$, for all $j$ and $J$,
  which exactly recover the solutions $\widevec{xy}$ from \eqref{eq:semialg-solution-set}.

  To implement each non-equality $G'_j(\widevec{xy},\widevec{ab}) \neq 0$,
  we create a custom extended linkage cell gadget, the \term{Crossing End Gadget} $\cLendcrossing$, that is \emph{not} globally noncrossing, but instead takes as input a real number in some range (namely, the horizontal offset of vertex~$g$) and induces a crossing precisely when that input is zero:

  \begin{lemma}[Crossing End Gadget]
    \label{lem:crossing-end-gadget}
    The Crossing End Gadget, $\cLendcrossing$, is a copy of $\cLangular$ with one additional edge $b_2 h$ of length $Q/2$. (It is drawn short to emphasize that $g$ and $h$ are distinct and do not share an edge, but in the initial configuration depicted, $g$ and $h$ actually overlap.) Let $H\subset\Conf(\cLendcrossing)$ consist of those configurations in which $g$ has $y$ coordinate equal to $Q/2$. Then the only configuration in $H$ with a crossing has $g = (Q/2,Q/2)$. All configurations in $H$ have minimum feature size at least $1/2$ if the distance between edges $b_2 h$ and $k g$ is ignored.
  \end{lemma}

  \begin{proof}
    This follows from Lemma~\ref{lem:scale-invariant-gadgets} (Angular Gadget).
  \end{proof}

  We may now proceed as in the proof of the Main Theorem (Theorem~\ref{thm:main-theorem}) using polynomials $F'_1,\ldots,F'_s$ and $G'_1,\ldots,G'_r$, and concluding each $G'_j$ with a Crossing End Gadget instead of an End Gadget. The rest of the linkage is already designed to ensure that the angles $\alpha,\beta$ fed into each Crossing End Gadget satisfy $R\cdot\Rect(\alpha,\beta) = (G'_j(\widevec{xy})-G'_j(\vec 0),0)$ (see Step~5 of Section~\ref{step:5}), so only configurations in the set $H$ from Lemma~\ref{lem:crossing-end-gadget} are possible.
  Then, as in Theorem~\ref{thm:unit-universality}, we pin all the vertices $v_{P,J}$ in the plane to force variables $(a_{P,J},b_{P,J})$ to take the values $(c_{P,J},0)$.
  For the resulting extended linkage $\cE$ with drawing vertices $X$, we conclude that the map $\pi_X$ is a homeomorphism between $\NXConf(\cE)$ and a translation of $\frac{1}{D}\cdot U$, and furthermore, the only crossings in any configurations of $\cE$ come from the Crossing End Gadgets.

  As in Section~\ref{sec:implementing-with-partially-rigidified}, eliminate sliceforms in $\cE$ to form $\cE'$ and process this further (as before) into partially rigidified linkage $\cL$. Finally, modify $\cL$ into a matchstick linkage $\cM$ as in Construction~\ref{con:matchstick-main-theorem}.

  We claim that $\cM$ simulates $D\cdot\cE$. This means that each noncrossing configuration $C$ of $\cE$ gives rise to a noncrossing configuration of $\cM$, i.e., no \emph{unintended} crossings arise in $\cM$'s underlying linkage. Supposing the contrary, if $C$ is a crossing configuration, the crossing must happen in the vicinity of a Crossing End Gadget by Lemma~\ref{lem:crossing-end-gadget}: the edge polyiamonds ultimately built from edges $k g$ and $b_2 h$ of $\cLendcrossing$ must intersect. But the former edge polyiamond lies entirely in the lower half of the gadget, except for its endpoint on the line $y=Q/2$: indeed, in the Crossing End Gadget, vertex $g$ remains on the line $y=Q/2$ and the direction of edge $k g$ changes by at most $\pm\delta$, far less than the $30^\circ$ required for any other part of the edge polyiamond of $k g$ to reach the line $y=Q/2$. For similar reasons, the edge polyiamond for $b_2h$ remains above $y=Q/2$ except for its stationary endpoint at $h=(Q/2,Q/2)$, so they could intersect only at~$h$. But then $C$ is a crossing configuration of $\cE$, contrary to assumption. So $\cM$ indeed simulates $D\cdot\cE$ and thus draws the set $D\cdot\frac{1}{D}\cdot U = U$.

  Finally, we consider the general case where $U$ does not necessarily reside in $[-1,1]^{2m}$. First choose an integer $n$ large enough so that $\frac{1}{n}\cdot U$ lies in $[-1,1]^{2m}$, and construct a matchstick linkage $\cL$ that draws $\frac{1}{n}\cdot U$ as detailed above. Build a new matchstick linkage $\cL'$ that copies $\cL$ except that each edge polyiamond is $n$-times longer. This linkage $\cL'$ will indeed draw $U$, as required.
\end{proof}

\section{Open Problems}

Table~\ref{tab:results} settles most problems in this area, but a few
interesting open problems remain.
The one combination in the table that remains unsolved is
the complexity of deciding global rigidity in graphs
with unit edge lengths, allowing crossings.
In particular, are there \emph{any} such graphs larger than a triangle?
(If not, the decision problem has an easy algorithm!)

We could also consider additional graph types (rows) in
Table~\ref{tab:results}.  For example, we considered unit edge lengths
and edge lengths in $\{1,2\}$, both when allowing crossings and when
forbidding crossings, but we did not consider globally noncrossing graphs
with edge lengths restricted to $\{1\}$ or $\{1,2\}$.  Do these linkages
remain universal for compact semialgebraic sets?  Is realizing them
$\exists\bR$-complete?  Is testing their rigidity and global rigidity
$\forall\bR$-complete?  Our results for globally noncrossing graphs
use only integer edge lengths bounded by a universal constant~$D$,
so we are ``only'' a constant factor away from a bound of $1$ or~$2$.

On the practical side, $D \approx 10^{42}$ is rather large.
We have been loose with our constants to keep the analysis
as simple as possible, but it would be interesting to tune the constants
and see how practical the whole construction can be made.
After all, a primary motivation for avoiding crossings was to make it
possible to physically construct universal linkages.

We also introduced the class of globally noncrossing graphs.
Is it $\forall\bR$-complete to determine whether a graph with edge-length
constraints is globally noncrossing, that is, whether all its realizations
are noncrossing?
Our $\exists\bR$-completeness proof for realizing globally noncrossing graphs
shows that distinguishing between unrealizable graphs and realizable
globally noncrossing graphs is $\exists\bR$-complete, but both cases
are technically ``globally noncrossing''.

What if we relax our edge-length constraints to allow approximate
solutions, modelling a small amount of pliancy in the bars or
tolerance at the hinges? Approximate realizations where each edge can
be stretched by at most some $\alpha$ factor are considered
extensively in the field of metric embedding, with many interesting
approximation algorithms.  Saxe \cite{Saxe-1979} proved that it is
strongly NP-complete to distinguish between graphs realizable with
stretch $\alpha = 1 + {1 \over 18}$ from graphs realizable with
stretch $\alpha = 1 + {1 \over 9}$, when embedding into one dimension.
What about embedding into two dimensions?  What about when restricted
to globally noncrossing graphs, matchstick graphs, graphs with unit
edge lengths, or graphs with edge lengths in $\{1,2\}$?  (Saxe's proof
uses edge lengths in $\{1,2,3,4\}$.)

\section*{Acknowledgments}

We are grateful to the anonymous referees whose thorough suggestions greatly improved the presentation of this paper.

\bibliography{kempe}
\bibliographystyle{plainurl}

\end{document}